\documentclass[sigconf, nonacm]{acmart} 
\usepackage{import}
\usepackage{subfig}
\usepackage{overpic}
\usepackage{makecell}
\usepackage{balance}
\usepackage[a-2b]{pdfx}[2018/12/22]
\usepackage[ruled,vlined,linesnumbered]{algorithm2e} 
\usepackage{multirow}
\makeatletter
\renewcommand\@fnsymbol[1]{\ifcase#1\or †\or \dagger\or \ddagger\else\@arabic{#1}\fi}
\makeatother

\AtBeginDocument{%
  \providecommand\BibTeX{{%
    \normalfont B\kern-0.5em{\scshape i\kern-0.25em b}\kern-0.8em\TeX}}}

\begin{document}

\title{Federated Heavy Hitter Analytics with Local Differential Privacy}

\author{Yuemin Zhang}
\affiliation{
  \institution{The Hong Kong Polytechnic University}
  \country{China}
}
\email{ymin.zhang@connect.polyu.hk}

\author{Qingqing Ye}
\authornote{Corresponding author.}
\affiliation{
  \institution{The Hong Kong Polytechnic University}
  \country{China}
}
\email{qqing.ye@polyu.edu.hk}

\author{Haibo Hu}
\affiliation{
  \institution{The Hong Kong Polytechnic University}
  \country{China}
}
\email{haibo.hu@polyu.edu.hk}

\begin{abstract}
Federated heavy hitter analytics enables service providers to better understand the preferences of cross-party users by analyzing the most frequent items. As with federated learning, it faces challenges of privacy concerns, statistical heterogeneity, and expensive communication. Local differential privacy (LDP), as the \textit{de facto} standard for privacy-preserving data collection, solves the privacy challenge by letting each user perturb her data locally and report the sanitized version. However, in federated settings, applying LDP complicates the other two challenges, due to the deteriorated utility by the injected LDP noise or increasing communication/computation costs by perturbation mechanism. To tackle these problems, we propose a novel target-aligning prefix tree mechanism satisfying $\epsilon$-LDP, for federated heavy hitter analytics. In particular, we propose an adaptive extension strategy to address the inconsistencies between covering necessary prefixes and estimating heavy hitters within a party to enhance the utility. We also present a consensus-based pruning strategy that utilizes noisy prior knowledge from other parties to further align the inconsistency between finding heavy hitters in each party and providing reasonable frequency information to identify the global ones. To the best of our knowledge, our study is the first solution to the federated heavy hitter analytics in a cross-party setting while satisfying the stringent $\epsilon$-LDP. Comprehensive experiments on both real-world and synthetic datasets confirm the effectiveness of our proposed mechanism.

\end{abstract}

\keywords{Local Differential Privacy,  Heavy Hitters, Federated Analytics}

\maketitle
\section{Introduction}
Federated learning~\cite{kairouz2021advances}, which was introduced in 2016, allows developers to train machine learning models across multiple data parties without centralized data collection. Following the success of this computing paradigm, federated analytics (FA), which was coined by Google in 2020 for its Gboard application~\cite{R2020}, involves utilizing federated technologies to answer more fundamental queries about decentralized data that often do not involve machine learning. Federated analytics mainly considers analytical queries, especially statistical queries~\cite{elkordy2023federated}, such as heavy hitters (i.e., the most frequent items) identification~\cite{ZKM2020, CB2022, WLJ2019}. As with federated learning, FA faces the challenges of {\bf privacy concerns, statistical heterogeneity, and expensive communication}~\cite{WSZ2021}. Several privacy-preserving techniques (e.g., differential privacy, secure multi-party computing) have been investigated to tackle privacy concerns in statistical queries~\cite{xu2020collecting, acharya2019communication, BK2021, wang2022fedfpm}. However, most of them are restricted to a simplified scenario where the client is semi-trusted and holds users' raw data~\cite{DSZ2022}, or each client only includes a single user~\cite{wang2022fedfpm}. 

In practice, users are often distributed across different parties~\cite{SHH2023}. For instance, Amazon may require its branches in different regions (e.g., Europe and America) to collaboratively identify the top $k$ items frequently purchased during the Christmas holiday, to develop more effective marketing plans. For the sake of privacy, users are not willing to upload their raw data to the branches since they may not be fully trusted. Local differential privacy~\cite{KLN2011} (LDP), as the golden standard for private data collection, allows users to perturb their data locally and send the sanitized data to the server, making it well-suited for privacy-preserving federated data analytics.

In this work, we study the problem of federated (a.k.a., cross-party) heavy hitter analytics in the context of LDP. In fact, LDP has been proposed and adopted in many distributed systems, such as Emoji and Safari usage data collection in Apple~\cite{A2017} and Chrome usage data analysis in Google~\cite{EPK2014}. While the privacy challenge of federated heavy hitter identification appears to be addressable by LDP, it unfortunately complicates the other two challenges, namely statistical heterogeneity and expensive communication/computation. In particular, as the data across different parties usually exhibit statistical heterogeneity (e.g., non-IID nature), simply aggregating the noisy data perturbed by the LDP mechanism may not accurately reflect the global representation of the data distribution. \textcolor{black}{On the other hand, LDP mechanisms (e.g., unary encoding~\cite{EPK2014}) often involve encoding a single value into a vector (e.g., with a domain length $|\mathcal{X}|=2^{64}$), resulting in high communication costs (e.g., up to $O(|\mathcal{U}|\cdot |\mathcal{X}|)$) in the server side where $|\mathcal{U}|$ is the user population size. 
Furthermore, such mechanisms identify heavy hitters by making $|\mathcal{X}|$ queries to estimate the frequency of every value in $\mathcal{X}$, however, issuing $|\mathcal{X}|$ queries is computationally infeasible~\cite{WLJ2019,CMM2021,LLL2024}. 
As such, parties are not expected to transmit the large amount of user data (even being perturbed) to the central server, as it may consume large network bandwidth and excessive computation cost.} 

Fortunately, most of FA tasks can be disentangled into several sub-tasks, building on the natural basis that the overall results are composed of various partial results held by different parties and users. For example, in federated heavy hitter identification, an item is identified as a global top $k$ item if the sum of its counts on Amazon Europe and Amazon America is ranked within the top $k$. 
As such, the central server can break down the task of federated heavy hitter identification into some sub-tasks and delegates them to each party instead of requiring them to upload the local datasets.

Following this design, a naive idea is to implement an existing LDP solution (e.g., PEM~\cite{WLJ2019}) in each party to identify local heavy hitters and then the central server aggregates these results by counting. However, it is non-trivial to design such a mechanism with good utility in the federated setting, because of the non-IID nature of decentralized data. First of all, data heterogeneity across different parties can lead to skewed (and noisy) frequency distributions, resulting in many false positives. 
On the other hand, the non-IID issue also complicates the estimation and aggregation processes. The local heavy hitters identified within a party may not contribute to identifying the global ones. For instance, some globally infrequent items may be popular within a specific party, while overwhelming the other globally frequent ones and resulting in inconsistencies. 

In this paper, we propose a target-aligning prefix tree mechanism for federated heavy hitter analytics, motivated by the typical using of a prefix tree (a.k.a, \textit{trie}) structure to iteratively identify heavy hitters~\cite{ZKM2020,CB2022} in the literature. In particular, we develop a two-phase process to align local and global targets. In the first phase, we design a shared shallow trie construction strategy, where each party independently estimates its trie with the same level. The aggregated leaves of this trie serve as a warm start for subsequent independent estimation of the second phase in each party, assisting in the filtration of false positive prefixes at a shallow level. Meanwhile, we introduce an adaptive trie extension strategy, by integrating the frequency distribution and the LDP noise scale as constraints to enhance the estimation. This strategy serves to adaptively determine the prefixes to be extended at different levels within each party.

We further optimize the estimation accuracy by introducing a consensus-based pruning strategy. This is because, in the trie construction process, the extended domains are likely to contain numerous unnecessary candidates, resulting in large domain sizes and thus introducing excessive LDP noises. Firstly, globally infrequent prefixes or items do not serve as useful candidates. Secondly, some prefixes or items may be popular globally but infrequent in specific parties due to the non-IID issue. Both types of prefixes fail to provide reliable frequency information. In the FA setting, it is a natural idea to seek priori knowledge from other parties to improve the current party's estimation. Therefore, our consensus-based pruning strategy is designed to implement the second phase sequentially, so that each party adaptively prunes unnecessary candidates suggested by the previous one. 

The key contributions of our study are summarized as follows:
\begin{itemize}
    \item[$\bullet$] This is the first to explore stringent $\epsilon$-LDP solution to the federated heavy hitter analytics under the cross-party setting. A novel target-aligning prefix tree mechanism is designed to adaptively extend the prefix tree at different levels.
    
    \item[$\bullet$] We further propose a consensus-based pruning strategy that enhances the alignment between the local and global targets. This strategy utilizes the prior knowledge from other parties to assist the pruning in the current party effectively.  
    
    \item[$\bullet$] We conduct comprehensive experiments on real-world and synthetic datasets to confirm the superiority of our mechanism in comparison with existing solutions.
\end{itemize}

The remainder of this paper is organized as follows. We discuss related works in Section~\ref{related work}. In Sections~\ref{preliminary} and~\ref{limitations}, we provide necessary background information and the problem statement. In Sections~\ref{method} and~\ref{Optimization}, we detail the proposed mechanism. In Section~\ref{experiment}, we present the experimental results. Finally, we conclude our paper in Section~\ref{conclusion}. 
\section{Related Work}\label{related work}
Differential privacy (DP)~\cite{DMN2006} has been considered the golden standard for privacy protection in recent years, and has widely used in many tasks such as frequent subgraph mining~\cite{CSX2018}, high-dimensional data publishing~\cite{ZCP2017, CXZ2015, STC2016, CTS2020}, and learning problems~\cite{fu2023dpsur}. As it always relies on a trusted server, local differential privacy (LDP)~\cite{KLN2011,li2019mobile} is proposed to allow users to locally perturb their data and report a noisy version. A number of studies have been conducted on the application of LDP, e.g., key-value data collection~\cite{YHM2019} and time-series data release~\cite{ye2023stateful}, and several companies have implemented LDP in their products~\cite{EPK2014, A2017, DKY2017}. Some recent works study data poisoning attacks to LDP~\cite{cao2021data, sun2024ldprecover,huang2024ldpguard}.

Federated analytics~\cite{R2020} allows individuals to collaboratively contribute to analytical tasks that do not require training, such as histogram construction and heavy hitters~\cite{WSZ2021}, in contrast to federated learning, which requires training a neural network. The problem of identifying privacy-preserving heavy hitters is a fundamental non-training application and has been studied under both DP and LDP settings. Zhu \textit{et al}.~\cite{ZKM2020} propose TrieHH that satisfies DP. It focuses on single-party settings and does not satisfy LDP. 
In contrast, our study focuses on and well addresses challenges raised in the multi-party setting while achieving a rigorous LDP guarantee. TrieHH++~\cite{CB2022} extends TrieHH and provides ($\epsilon, \delta$)-aggregated differential privacy. 
Chadha \textit{et al}.~\cite{CCD2023} address a scenario where each user possesses multiple data points under aggregate DP. 
Li \textit{et al}.~\cite{LLL2024} explore tracking heavy hitters on data streams at bounded memory expense under LDP. 
Wang \textit{et al}.~\cite{WLJ2019} introduce a protocol under LDP for identifying heavy hitters by discovering the popular prefixes with increasing lengths, named the prefix extending method (PEM), which is also designed for the single-party setting. 
 
The existing studies that comply with LDP predominantly focus on a traditional local setting with a single server, disregarding the more practical scenario where users' data are distributed across multiple servers. 
The most relevant work~\cite{SHH2023} to ours employs a hierarchical approach to identify local and global heavy hitters in cross-party settings, but does not satisfy LDP. In this work, each user utilizes its own data and its prefix to select the  random item $X$ to add to the domain. It aims to address the out-of-domain issue when a user's prefix lies outside the predefined domain. However, it results in different output domains rely on each user's sensitive data, undermining the LDP guarantee. In response, our work introduces a federated heavy hitter analytics mechanism designed to detect heavy hitters across multiple parties while mitigating the negative impacts of non-IID data and satisfying $\epsilon$-LDP.

\section{Preliminaries}\label{preliminary}

\subsection{Local Differential Privacy}
\label{subsec:ldp}
As the local variant of differential privacy (DP)~\cite{DMN2006}, LDP~\cite{KLN2011} eliminates the need for a trusted data collector, making it a viable privacy model for numerous applications in distributed settings, such as frequency and mean estimation~\cite{WBL2017,WXY2019}, trajectory data synthesis and collection~\cite{DHZ2023,ZYC2023}, and preference ranking analysis~\cite{YCS2022}. The following provides a formal definition of LDP.

\begin{definition}[$\epsilon$-LDP]
\label{def:ep_ldp}
Given a privacy budget $\epsilon>0$, a randomized mechanism $\mathcal{M}:\mathcal{X}\to\mathcal{Y}$ provides $\epsilon$-LDP if and only if, for any two inputs $x, x^{\prime}\in \mathcal{X}$ and any possible output $y\in \mathcal{Y}$, the following inequality holds: 
$\mathrm{Pr}[\mathcal{M}(x)=y]\leq e^{\epsilon} \times \mathrm{Pr}[\mathcal{M}(x^{\prime})=y].$
\end{definition}

The plausible deniability of LDP is achieved by perturbing users' data locally and is controlled by privacy budget $\epsilon$. Similar to centralized settings, LDP enjoys the desired property of post-processing~\cite{DR2014}.

\begin{theorem}
    (Post-Processing) Let $\mathcal{M}:\mathcal{X}\to\mathcal{Y}$ be a mechanism that satisfies $\epsilon$-LDP, and $\mathcal{F}:\mathcal{Y}\to\mathcal{Y}^\prime$ be an arbitrary randomized mapping. Then \textcolor{black}{$\mathcal{F}\circ\mathcal{M}:\mathcal{X}\to\mathcal{Y}^\prime$ satisfies $\epsilon$-LDP}, where $\circ$ denotes the composition of $\mathcal{F}$ and $\mathcal{M}$.
\end{theorem}

The post-processing property guarantees that no privacy leakage occurs when perturbed data undergo further processing.

\subsection{Frequency Oracles}
\label{subsec:fo}
\textcolor{black}{A \textit{frequency oracle} (FO) is an LDP mechanism for frequency estimation. 
In this subsection, we briefly revisit three classic FOs, namely $k$-ary randomized response ($k$-RR)~\cite{W1965, KOV2014}, optimized unary encoding (OUE)~\cite{WBL2017}, and optimized local hashing (OLH)~\cite{WBL2017}.}

\textbf{$\boldsymbol{k}$-RR. }
Given a privacy budget $\epsilon$, $k$-RR mechanism $\mathcal{M}:\mathcal{X}\to\mathcal{X}$ satisfies $\epsilon$-LDP, if, for any input $x \in \mathcal{X}$, the output $y\in\mathcal{X}$ is sampled from the following distribution:
\begin{equation}\label{eq2}
\mathrm{Pr}[\mathcal{M}(x)=y]=
\begin{cases}
  p=\frac{e^{\epsilon}}{\mathcal{|X|}-1+e^{\epsilon}} & \text{ if } y=x, \\
  q=\frac{1}{\mathcal{|X|}-1+e^{\epsilon}} & \text{ otherwise}.
\end{cases}
\end{equation}

By Equation~\ref{eq2}, $k$-RR consistently reports the original value (i.e., $y=x$) with a higher probability $p$. To estimate the frequency of $x \in \mathcal{X}$---that is, the proportion of users whose private value is $x$ relative to the total user count---one counts the occurrences of $x$ being reported, represented as $c_x$, and then computes $f^{k\text{-RR}}_x = \frac{c_x/n - q}{p - q},$ 
where $n$ signifies the total number of users. The variance for this unbiased estimation~\cite{WBL2017} is $Var[f^{k\text{-RR}}_x] = \frac{\mathcal{|X|}-2+e^{\epsilon}}{(e^{\epsilon}-1)^2\cdot n}$.
It has been proven that $k$-RR achieves good performance for domain sizes $k$ (i.e., $\mathcal{|X|}$) smaller than $3e^\epsilon + 2$~\cite{WBL2017}. For a larger domain size, the state-of-the-art OUE~\cite{WBL2017} mechanism has been proven that can outperform $k$-RR.

\textbf{OUE. }
Given a privacy budget $\epsilon$, the OUE mechanism satisfies $\epsilon$-LDP if it first encodes the input $x \in \mathcal{X}$ as a length-$|\mathcal{X}|$ one-hot binary vector $B=Encode(x)=[0,\cdots,0,1,0,\cdots,0]\in\mathcal{B}$, where only the $x$-th position is 1. It then perturbs each position of $B$, and the output $B^\prime\in \mathcal{B}$ is sampled from the following distribution:
\begin{equation}\notag
    Pr[B^\prime[x]=1]=
    \begin{cases}
      p=\frac{1}{2} & \text{ if } B[x]=1, \\
      q=\frac{1}{e^{\epsilon}+1} & \text{ if }B[x]=0.
    \end{cases}
\end{equation}
The perturbed vector is viewed as supporting an input $x$ if $B[x]=1$. Let $c_x$ denote the counts of the supports of $x$ among all the reported vectors, then frequency $f^{\text{OUE}}_x$ of $x$ is the same with $k$-RR. The variance of this estimation~\cite{WBL2017} is $Var[f^{\text{OUE}}_x] = \frac{4e^{\epsilon}}{(e^{\epsilon}-1)^2\cdot n}$. 
However, for significantly larger domain sizes, OUE can result in excessive communication costs. Therefore, OLH~\cite{WBL2017} is recommended.

\textbf{OLH. }
Let $\mathbb{H}$ be a universal hash function family, with each $H\in \mathbb{H}$ outputting a value in $[d^{\prime}]$, where $d^{\prime}=\lceil e^\epsilon+1\rceil$. Given a privacy budget $\epsilon$, the OLH mechanism satisfies $\epsilon$-LDP if it first encodes the input $x \in \mathcal{X}$ as $v=H(x)$, with $H$ chosen uniformly at random. It then perturbs $v$, and the output $y\in [d^{\prime}]$ is sampled according to the following distribution:
\begin{equation}\label{eq5}\notag
\mathrm{Pr}[y=v^{\prime}]=
\begin{cases}
  p=\frac{e^{\epsilon}}{d^{\prime}-1+e^{\epsilon}} & \text{ if } v^{\prime}=v, \\
  q=\frac{1}{d^{\prime}-1+e^{\epsilon}} & \text{ otherwise}.
\end{cases}
\end{equation}

For each $x \in \mathcal{X}$, let $c_x = |\{u \mid H^{u}(x) = y^u\}|$ denote the number of reports supporting the input as $x$, where $y^u$ represents the output of user $u$, perturbed by the randomly chosen hash function $H^{u}$. The frequency estimation of $x$ and the variance of this estimation~\cite{WBL2017} are the same as the OUE mechanism. \textcolor{black}{However, the computation cost of OLH is much larger, especially when the domain size is overly large~\cite{WLJ2019,CMM2021}. }
All the above three mechanisms can be used as an FO, each subject to different practical constraints. In addressing the heavy hitter problem, the FO is typically treated as a black box. Since the strategy for identifying heavy hitters being the primary challenge, which is also the most significant and impactful aspect.

\subsection{Heavy Hitter Identification}
\label{subsec:heavy hitter}
Identifying heavy hitters under LDP in the local setting focuses on the most frequent items, especially when the item domain size is extremely large. Under such conditions, it is impractical to directly apply FOs (e.g., OUE), \textcolor{black}{due to prohibitive communication and computation costs~\cite{CMM2021,LLL2024,WLJ2019}.} There are some existing works that consider the set-valued LDP setting (e.g., LDPMiner~\cite{QYY2016} and SVIM~\cite{WLJ2018}), which also incur high communication costs. To reduce these costs while still achieving good utility, prefix-tree-based iterative algorithms are commonly employed~\cite{WXY2018,WLJ2019,BS2015}. For instance, PEM~\cite{WLJ2019}, the state-of-the-art heavy hitter identification mechanism under LDP, identifies heavy hitters by constructing a prefix tree iteratively. Concretely, PEM assigns the users into different groups to use an FO to report prefixes with different lengths of their encoded items. The server then sequentially processes each group's reports to determine which top $t$ popular prefixes (where $t$ refers to the extension number) should be extended in the next group for longer candidate prefixes. The heavy hitters are identified after the final group's estimations are completed. 

Although the prefix-tree based algorithms are efficiently in identifying heavy hitters in local settings, the previous studies all consider a single-party setting, leaving a gap in the practical federated setting that the users are always distributed across multiple parties. 
In this setting, each party can only obtain the noisy answers from its groups, but cannot access the results from other parties.

\section{Problem Statement and Challenges}
\label{limitations}

\subsection{Problem Statement}
\label{subsec:problem}
We consider a federated setting, where a central server wants to identify the heavy hitters in the data distributed across a set of parties $\mathcal{P}=\{P_1, \cdots, P_{|\mathcal{P}|}\}$. Each party $P_i\in \mathcal{P}$ comprises a distinct set of users $U^{i}$, without overlap with the other parties. Each user $u_j\in U^{i}$ holds an item $x_j\in \mathcal{X}$, where $\mathcal{X}$ is the item domain. 
Since the parties are not entirely trusted, users will locally perturb their private values via LDP mechanisms and only send noisy values to the parties. Upon receiving the sanitized data, each party computes and sends some partial results to the server for identifying the top-$k$ federated heavy hitters, which are formally defined as follows. 

\begin{definition}[Top-$k$ Federated Heavy Hitter]
Given a set of parties $\mathcal{P}$ and the value domain $\mathcal{X}$, an item $x\in \mathcal{X}$ is a top-$k$ federated heavy hitter if its frequency, 
\begin{equation}\label{eq8}\notag
f_{x}=\frac{\sum_{P_i\in \mathcal{P}}|\{j|x_j=x, u_j\in U^i\}|}{\sum_{P_i\in \mathcal{P}}|U^{i}|},
\end{equation}
is ranked within top $k$ over frequencies of all possible values in $\mathcal{X}$ across all parties $\mathcal{P}$.
\end{definition}

Intuitively, upon receiving perturbed data from the users, the party can directly send the local dataset to the central server for heavy hitter identification. However, this may produce the excessive communication cost. In practice, both the user population and the item domain are usually overly large, e.g., reaching the million level in recommendation scenarios~\cite{GZL2023,MZH2018}. Assume we have $5,000,000$ users and $|\mathcal{X}|=2,000,000$, the state-of-the-art FO, OUE~\cite{WBL2017}, needs to encode an input value into a $2,000,000$-bit vector, then perturbs and submits it to the server. Thus, the communication cost in the server side is $1\times 10^{13}$ bits, which is unacceptable in many applications. One can use other communication efficient FOs, e.g., OLH, but communication costs are still large and decoding a single randomized data (via such mechanisms) on the server side requires a comprehensive scan of the entire domain $\mathcal{X}$, which is computationally infeasible for large data domains~\cite{WLJ2019,CMM2021,LLL2024}. To mitigate the expensive communications and computations, the central server alternatively breaks down the task of federated heavy hitter identification into some sub-tasks and delegates them to each party instead of requiring them to upload the distributed data.

\subsection{Straw-Man Solution and Challenges}
\label{subsec:limitations}
In the literature, some existing works study the simplified scenario of federated heavy hitter identification under differential privacy. A line of work assumes that the distributed parties are semi-trusted, so that they can collect users' raw data, inject DP noise into query results and then send the noisy version to the central server~\cite{BK2021,JKK2023}. In the context of LDP, all the existing works assume each client only includes a single user, so that users can directly send the perturbed value to the central server~\cite{WXY2018,WLJ2019,BS2015}.

In practice, users are always distributed among multiple untrusted parties. For instance, Amazon may require its branches in different regions (e.g., Europe and America) to collaboratively identify the top $k$ items frequently purchased or liked during the Christmas holiday to develop more effective marketing plans. 
Under this circumstance, we alternatively turn to LDP. Specifically, users send the perturbed data to the party, which can then estimate some partial results (e.g., local heavy hitters~\footnote{For ease of understanding, we refer to heavy hitters identified by a single party $P_i$ as local heavy hitters.}). Then parties send the partial results to the server for identifying federated heavy hitters. 

This idea naturally leads to a straightforward solution, as described in Algorithm~\ref{FAPEM}. Specifically, each party can employ existing approaches for heavy hitter identification (e.g., PEM or SPM~\cite{FPE2016}) to estimate the local heavy hitters, and then submit the identified top-$k$ heavy hitters along with their estimated frequencies to the central server. The server then aggregates the local heavy hitters from different parties to determine the federated heavy hitters. 

\begin{algorithm}[t]
	\footnotesize
	\caption{PEM for Federated Heavy Hitters (FedPEM)}\label{FAPEM}
	\KwIn{All parties $\mathcal{P}$, value domain $\mathcal{X}$, maximum binary length $m$, granularity $g$, query $k$, and privacy budget $\epsilon$}
	\KwOut{Top-$k$ federated heavy hitters $R^k$}
	\tcp{Party side:}
	Each party $P_i \in \mathcal{P}$ calls PEM($U^i, m, k, \varepsilon$, g) (Algorithm 1 in~\cite{WLJ2019}), where $U^i$ denotes the set of users in the $i$-th party\;
	
	Each party $P_i$ obtains the local top-$k$ heavy hitters $R^k_i$ and reports to server\;
	\tcp{Server side:}
	Initialize the federated heavy hitter dictionary $R^k=\{\}$\;
	Server aggregates and obtains the top-$k$ federated heavy hitters $R^k$\;
		
	\Return  $R^k$\;	
\end{algorithm}

However, such a solution overlooks the non-IID problem, which poses challenges for accurate aggregation on the server side. Federated heavy hitter identification faces two significant issues related to non-IID data. First of all, data heterogeneity incurs varying frequency distributions across parties. For instance, the differing frequency distributions of values can lead to skewed frequency distributions of prefixes. However, existing mechanisms in LDP overlook this issue and consistently apply the same extension number of $t=k$ in each iteration when extending popular prefixes to a longer length.
On the other hand, \textcolor{black}{the non-IID nature of the data complicates the estimation and aggregation processes.} The local heavy hitters identified within a party may not contribute to identifying the global~\footnote{In the sequel, global and federated are interchangable.} heavy hitters. For example, certain items may be frequent in Amazon but not popular, even not exist, in Shopee. Such items do not offer valuable information for federated heavy hitter identification and, to some extent, act as noise. 
This oversight can lead to false positives and the neglect of some significant popular prefixes. In the following sections, we will delve into these issues in more detail and introduce our mechanism to mitigate them.

\section{Methodology}\label{method}
This section presents our proposed solution. We first introduce our design rationale in Section~\ref{subsec: design rationale} and then overview our target-aligning prefix tree mechanism in Section~\ref{subsec:mechanism overview}. Then the implementation details are illustrated in Sections~\ref{subsec:shared shallow trie construction}--\ref{subsec:tap mechanism}, followed by an analysis on privacy guarantee and utility in Section~\ref{subsec:privacy}.

\subsection{Design Rationale}
\label{subsec: design rationale}

The effectiveness of constructing a prefix tree (a.k.a., \textit{trie}) for identifying heavy hitters has been verified in the typical LDP setting~\cite{WXY2018,WLJ2019}. Compared with other structures (e.g., multiple channels~\cite{BS2015} and segment based structure~\cite{FPE2016}), the prefix tree encodes an item into a binary tree and thus utilizes the sequential feature for pruning. \textcolor{black}{For example, given the item domain size $2^{64}$, each item can be encoded into a 64-bit vector, based on which a 64-level binary trie can be constructed. The inherent sequential feature allows for the iterative pruning of unnecessary items by filtering prefixes (e.g., preserve only items whose first two-bit prefix is $00$ or $10$).} This not only avoids the excessive communication cost but also reduces the scale of LDP noises by iteratively narrowing down the candidate domain. This inspires us to design an effective LDP-enabled prefix tree mechanism to accurately identify federated heavy hitters.

\begin{figure}[t]
    \centering
    \includegraphics[width=0.95\columnwidth]{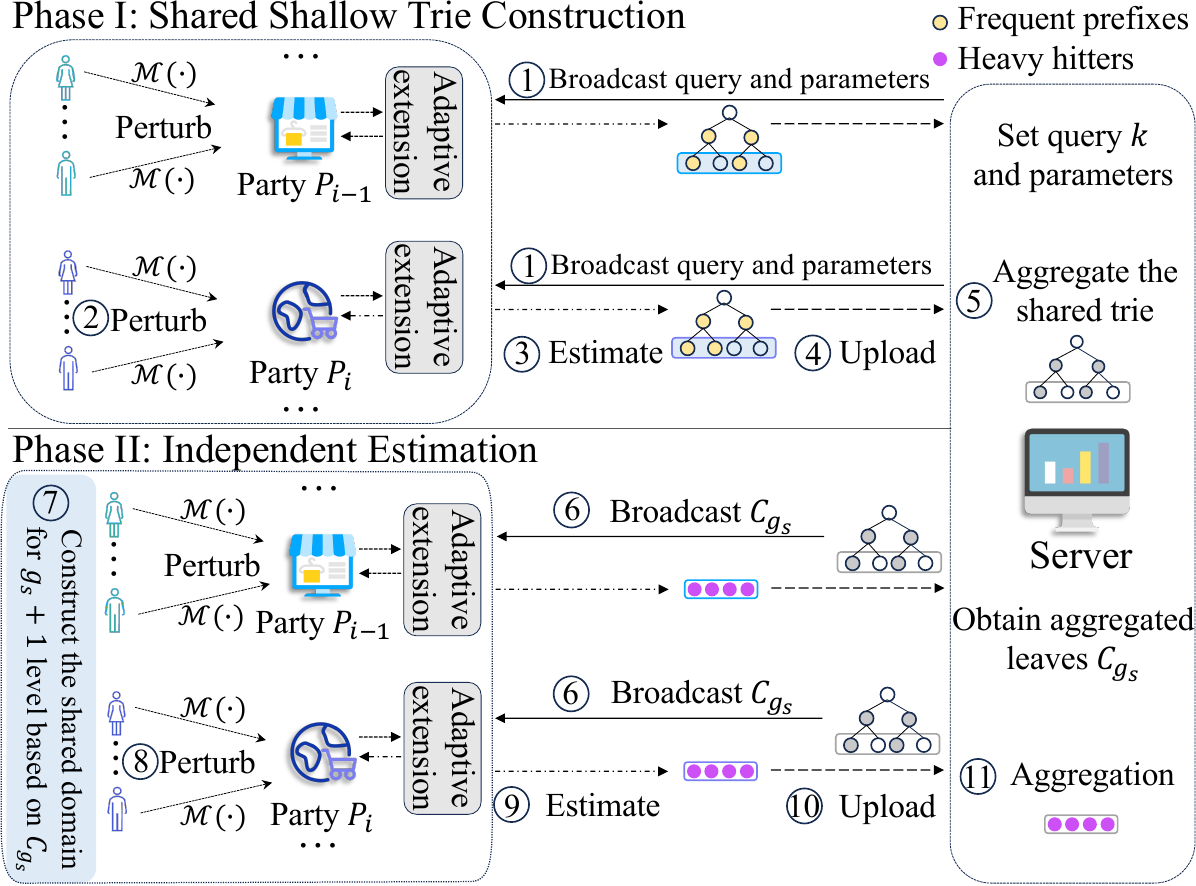}
    \vspace{-0.1in}
    \caption{An overview of the TAP mechanism.
    }
	\vspace{-0.2in}
    \label{fig_overview_tap}
\end{figure}

According to the Apriori property~\cite{AS1994}, the prefix of a popular value is likely to be popular as well. However, due to LDP noises and the non-IID nature of data in federated settings, some globally frequent values might be overwhelmed by certain prefixes and pruned in local parties. This results in inconsistencies between local and global targets, especially at shallow levels. As such, we propose a strategy of constructing a shared shallow trie that aggregates frequent prefixes across parties at a shallow level, serves as a warm start, and aligns such inconsistencies collaboratively. 

On the other hand, \textcolor{black}{in federated settings}, each party will maintain their own tries by minimizing the local trie construction errors, \textcolor{black}{which are inherently affected by both LDP noise and the underlying frequency distribution.} Therefore, we propose an adaptive trie extension strategy. As opposed to existing studies that apply fixed extensions, we adaptively determine the extension sub-processes during trie construction. This approach not only reduces LDP noises but also mitigates the inconsistency between identifying local heavy hitters and global targets (i.e., covering necessary prefixes).

\subsection{Mechanism Overview}
\label{subsec:mechanism overview}

The above design rationale leads to our \underline{t}arget \underline{a}ligning \underline{p}refix tree (TAP) mechanism. As shown in Figure~\ref{fig_overview_tap}, the TAP mechanism consists of two phases. In phase I, the server first broadcasts the query and parameters to all parties (step \textcircled{\small 1}). Upon receiving them, all users in each party are divided into several groups uniformly at random based on the increasing lengths of prefixes. Each group corresponds to a level of the trie, and each user in that group reports a perturbed prefix of her value. In each party, some frequent prefixes estimated by the current group are then extended to construct the candidate domain for the next group, to estimate longer prefixes. Each user reports only once, which avoids dividing the privacy budget and thus achieving more accurate estimations. During phase I, the first few user groups in all parties collaboratively construct a shared shallow trie (steps \textcircled{\small 2}--\textcircled{\small 5}), by aggregating the partial results from all parties, whose construction will be detailed in Section~\ref{subsec:shared shallow trie construction}. 

In phase II, the server first broadcasts the leaves of the constructed shared trie to each party (step \textcircled{\small 6}), based on which the parties extend the same leaves to construct candidate domains (step \textcircled{\small 7}). Then local heavy hitters in different parties can be identified (steps \textcircled{\small 8}--\textcircled{\footnotesize 10}). 
Note that trie extensions in both phases are always adaptive to the noisy frequency distributions, which will be further elaborated in Section~\ref{subsec:adaptive trie extension}. 
Finally, the server aggregates these findings to identify the federated heavy hitters (step \textcircled{\footnotesize 11}). 

\subsection{Shared Shallow Trie Construction}
\label{subsec:shared shallow trie construction}

As aforementioned, the non-IID problem may cause inconsistencies to the estimation among parties. 
Numerous infrequent values can share identical shorter prefixes, which might appear frequent at shallow levels. 
Due to LDP noises and the non-IID problem, values that are genuinely frequent across all parties might be overshadowed by these common prefixes and get prematurely pruned by local parties, especially at shallow levels. For instance, in Figure~\ref{fig_examples}(a), there are two parties (i.e., party A and party B) \textcolor{black}{that encode the entire domain into a trie (e.g., $2^{64}$ items represents by a 64-bit binary string) where the prefixes at the second level of the trie, namely the first two-bit prefix domain, are \{$00$, $01$, $10$, $11$\},} and the query is for the top $k=2$ heavy hitters. 
At the current level (shaded in blue), both parties select the yellow prefixes to be extended based on their own noisy observations. That is, party A extends the yellow prefixes $00$, $01$, and $11$, whereas party B chooses $00$ and $10$. We can observe that at the current level of the trie, the globally frequent prefixes (in grey) are $00$ and $10$. Unfortunately, without intervention, party A will miss $10$ (in pink) and further incur useless extensions at each deeper levels, which could introduce extra noise in identifying global heavy hitters. Therefore, we propose a strategy of constructing a shared shallow trie to mitigate this issue by aligning the local and global targets.

Algorithm~\ref{TAP_phase_I} shows the workflow. The server assigns each party $P_i\in \mathcal{P}$ a granularity $g$ and a shared trie level $g_s$ (line 1). Each $P_i$ instructs the first $g_s$ groups of users to estimate the frequent prefixes at the initial $g_s$ levels (lines 4--8). We apply adaptive extension here by considering the underlying frequency distribution and LDP noises, see Section~\ref{subsec:adaptive trie extension}. Then, each party sends the candidates with non-zero estimated counts at the $g_s$-th leves to the server (line 9). The server calculates the top $k$ frequent prefixes (denoted as $C_{g_s}$) based on the reports (line 10). For instance, as shown in Figure~\ref{fig_examples}(a), with $g_s=2$ and $k=2$, party A selects all prefixes as the candidates $C^A_{g_s}$ after completing the estimation at the second level (shown in blue), while $C^B_{g_s}$ consists of $00$ and $10$. They then report $C^A_{g_s}$, $C^B_{g_s}$, along with the counts of all prefixes to the server. The server aggregates them and broadcasts the top $k=2$ frequent prefixes, $C_{g_s}=\{00,10\}$, to all parties. Then in phase II, after receiving these globally frequent prefixes $C_{g_s}$, each party constructs the prefix domain $\Lambda_{g_s+1}$ by extending $C_{g_s}$ and continues their estimations independently until they identify their own local heavy hitters.

\begin{algorithm}[t]
  \footnotesize
  \caption{Shared Shallow Trie Construction (STC)}\label{TAP_phase_I}
  \KwIn{All parties $\mathcal{P}$, maximum binary length $m$, granularity $g$, shared trie level $g_s$, query $k$, and privacy budget $\epsilon$}
  \KwOut{Globally frequent prefixes $C_{g_s}$}
  The server broadcasts granularity $g$ and shared level $g_s$ to all parties;

  \For{$P_i \in \mathcal{P}$}{
  	Initialize $C^i_0 = \emptyset$, $h_0=0$\;
  	
  	Divide its users $U^i$ into $g$ groups $\{U^i_1, ..., U^i_g\}$ randomly\;
  	
  	\For{$h \in \{1, ..., g_s\}$}{
  		Compute the prefix length of $h$-th level $l_h =\lceil h \times \frac{m}{g} \rceil$\;
  		
  		Construct candidate domain $\Lambda_h$ = \texttt{Construct} ($l_h, l_{h-1}, C^i_{h-1}$)\;
  		
  		Get user reports to estimate ($C^i_h, \mathcal{I}^i_h$) = \texttt{Estimate} ($l_h, \epsilon, \Lambda_h, U^i_h$)\;
  	}
  
    Report prefix candidates $C^i_{g_s}$ and their counts $\mathcal{I}^i_h$ at $g_s$-th level to the server\;
  }

  The server aggregates the top $k$ frequent prefixes $C_{g_s}$ from $\{C^1_{g_s}, ..., C^{|\mathcal{P}|}_{g_s}\}$ based on $\{\mathcal{I}^1_{g_s}, ..., \mathcal{I}^{|\mathcal{P}|}_{g_s}\}$;

  \Return  $C_{g_s}$\;

  \SetKwFunction{FMain}{Construct}
    \SetKwProg{Fn}{Procedure}{:}{}
    \Fn{\FMain{$l_h, l_{h-1}, C_{h-1}$}}{

    \Return $\Lambda_h = C_{h-1} \times \{0, 1\}^{l_h - l_{h-1}}$\;
  }

  \SetKwFunction{FMain}{Estimate}
	\SetKwProg{Fn}{Procedure}{:}{}
	\Fn{\FMain{$l_h, \epsilon, \Lambda_h, U_h$}}{
	Each user in $U_h$ reports her noisy prefix via the given FO and $\epsilon$\;
	
	Estimate frequencies of all prefixes $\lambda_h \in \Lambda_h$\ with $l_h$ bits length\;
	
	Construct $C_h$ as $t=k^*+\eta$ highest estimated values (with their counts $\mathcal{I}_{h}$) by solving Equations~\ref{eq9} and ~\ref{eq10} (see Section~\ref{subsec:adaptive trie extension})\;
	
	\Return ($C_h$, $\mathcal{I}_h$)\;
}
\end{algorithm}

\subsection{Adaptive Trie Extension}
\label{subsec:adaptive trie extension}

In previous studies, the server extends the top $t=k$ (i.e., the extension number) frequent prefixes at each level to construct the candidate prefix domain for the next level. This extension number $t=k$ relies solely on the query $k$ and is independent of the current frequency distribution. We argue that such fixed extensions, influenced by LDP noise and the inherent prefix frequency distribution, may not yield an appropriate perturbation domain for the next level. Concretely, the extension number $t$ represents the number of child nodes to be extended, and the number of candidates grows exponentially with increasing $t$. Consequently, a large $t$ results in a large candidate domain, including many unnecessary prefixes to be extended (i.e., those that are not prefixes of heavy hitters). When applying LDP perturbation, the noise scale increases with the domain size, leading to excessive LDP noise. Thus, the goal of adaptive trie extension is to strike a balance between the amount of important prefixes and the injected LDP noise. 
Specifically, for a given party and the candidate prefix domain $\Lambda_h$ at level $h$, the party sorts the estimated frequencies $\hat{F}_h$ of all candidates after receiving the perturbed data from users, and then adaptively decides on the extension number $t$. This decision considers two factors: a less frequent prefix among the top $k$ serving as an anchor, and the amount of LDP noise. We first identify this anchor prefix as

\begin{equation}\label{eq9}
    \underset{k^{\ast}}{\arg\max} \ \frac{\sum_{1 < j \leq k^{\ast}} \hat{f_j}}{k^{\ast}} - \frac{\sum_{k^{\ast} < s \leq k+1} \hat{f_s}}{k+1 - k^{\ast}},
\end{equation}
where $1< k^{\ast}\leq k$, $\hat{f_j}$ (resp. $\hat{f_s}$) denotes the noisy frequency of the $j$-th (resp. the $s$-th) popular prefix. 
In Equation~\ref{eq9}, the first term calculates the average of the first $k^{\ast}$ prefixes' frequencies (except for the largest one since it must be preserved), and the latter represents the average frequency of the other relatively infrequent prefixes within the top $k+1$ prefixes at current level. We include the $(k+1)$-th frequent prefix, which represents the upper bound of the average frequency of the rest prefixes and avoids dividing by zero. 
The underlying rationale for this objective function is that such an anchor prefix, along with any more frequent prefixes, will likely cover the least frequent prefix among the final top $k$ heavy hitters at this level. The identification of $k^{\ast}$ distinguishes the prefix with a significantly influential frequency, which serves as the boundary between the prefixes of top $k$ at current level and others. However, the LDP noise must be took into account.

\begin{figure}[t]
    \centering
    \subfloat[An instance of the shared shallow trie construction.]{
    \includegraphics[width=0.9\columnwidth]{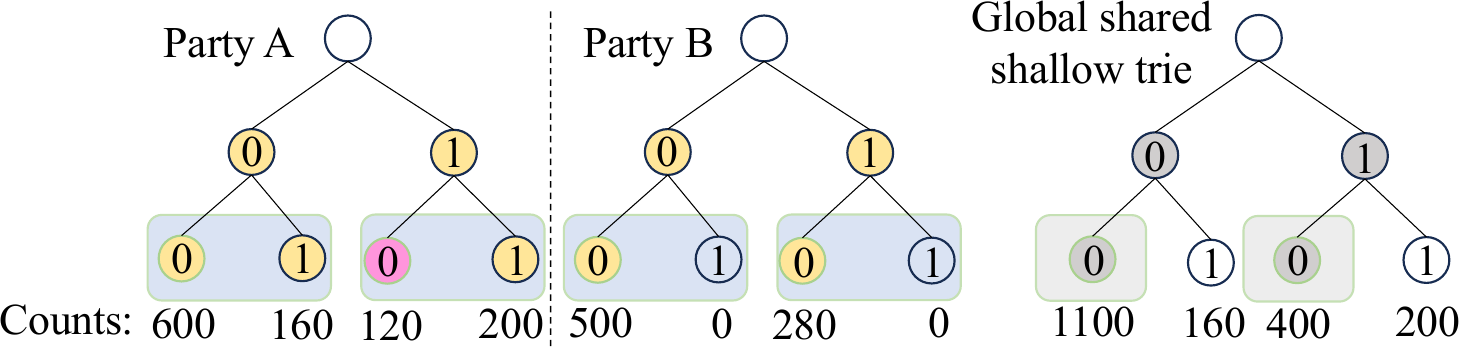}
    }\vspace{-0.1in}
    \hfill
    \subfloat[An example of the adaptive trie extension.]{
    \includegraphics[width=0.8\columnwidth]{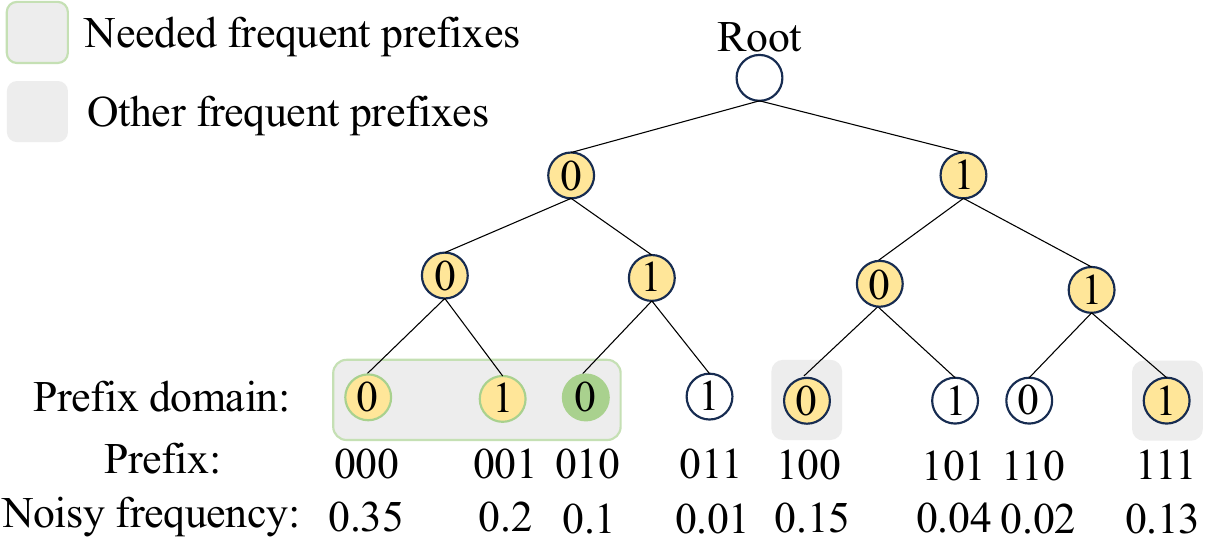}
    }
    \caption{Toy examples of the TAP mechanism.
    \vspace{-0.1in}
    }
    \label{fig_examples}
\end{figure}

As our objective is to increase the cover rate of necessary prefixes at each level, it is natural to estimate how far the anchor prefix can drift under the LDP noise. Such an estimation, along with the anchor $k^*$, is considered an approximated worst case for the extension number that covers the most infrequent prefix among the needed prefixes. Since an estimated frequency can be approximated by a Gaussian distribution~\cite{WLJ2019}, let $X_{k^*}$ and $X_{k^*+x}$ denote the frequency variables of the $k^*$-th and $(k^*+x)$-th observed popular prefixes, respectively. We can calculate the drift distance $\eta$ as follows:

\begin{equation}\label{eq10}
    \eta = \min(k, \mathbb{E}(x)), 
\end{equation}
where the $\mathbb{E}(x)$ is the expectation of $x$, derived as 
\begin{align*}
    & \mathbb{E}(x) = \sum\nolimits_{x=k^*-k+1}^{\min(k, \pi_{p}^{i} - k)} x \cdot \Pr(X_{k^*} \leq X_{k^*+x}), \\
    & \Pr(X_{k^*} \leq X_{k^*+x}) = \Pr(X_{k^*} - X_{k^*+x} \leq 0) \\ & \approx \frac{1}{\sqrt{4\pi}} \int_{-\infty}^{0} \frac{1}{\sigma}\exp\left( -\frac{(t-(\hat{f}_{k^*} - \hat{f}_{k^*+x}))^2}{4\sigma^2} \right) dt,
\end{align*}
where $\hat{f}_{k^*}$ and $\hat{f}_{k^*+x}$ are the observed noisy $k^*$-th and $(k^*+x)$-th frequencies and $\sigma$ denotes the standard deviation of the FO used for frequency estimation. 
In the above equations, we approximately calculate the expectation that the $k^*$-th prefix drifts to the $(k^*+\eta)$-th position under LDP noise. Specifically, we approximate the means of the frequency variables $X_{k^*}$ and $X_{k^*+x}$ using the noisy $k^*$-th and $(k^*+x)$-th frequencies, setting the standard deviation according to the specific FO. Then we simulate the perturbation and calculate the probability that $X_{k^*}\leq X_{k^*+x}$. Following this, we derive the expectation and set the drift distance to $\eta$ accordingly. Note that we bound $\eta$ that cannot exceed $k$ since the drift distance should not be too large. Ultimately, we determine the extension number as $t=k^*+\eta$. Thus, the current party can extend these top $t$ frequent prefixes to a longer length for further estimation. 
For example in Figure~\ref{fig_examples}(b), consider the $h$-level prefix domain as $\{000, 001, 010, 011, 100, 101, 110, 111\}$ and the query $k=4$. The necessary prefixes, namely, those of the top $k$ heavy hitters at the $h$-th level, are 000, 001, and 010.  After sorting the frequencies, we observe that the most frequent $k=4$ prefixes are $000$, $001$, $100$, and $111$ in yellow, excluding the needed prefix 010, shown in green, which is ranked below $4$ due to the LDP noise. 
In contrast to using a fixed extension number $t=k=4$, which would overlook the necessary prefix, our adaptive strategy sets $t=k^*+\eta=5$, thus encompassing the required ones.

\subsection{Putting Things Together}
\label{subsec:tap mechanism}
We now summarize the entire pipeline of the proposed TAP mechanism. As shown in Algorithm~\ref{TAP}, each party constructs the shared shallow trie to obtain the global frequent prefixes at the $g_s$-th level (line 1), as described in Algorithm~\ref{TAP_phase_I}. Then, each party continues the estimation process based on $C_{g_s}$ (line 2). They extend these prefixes and use the corresponding users from each group to estimate the longer popular prefixes, repeating this procedure until they ultimately report the local heavy hitters at the $g$-th level to the server (lines 3-5). Note that the extensions are always adaptive. The server then derives the federated heavy hitters by counting (line 6). 

The TAP mechanism utilizes the adaptive trie extension and the shared shallow trie construction to enhance the consistency of the local and global targets, enabling the identified local heavy hitters to provide much more valuable and reasonable statistical information for mining global heavy hitters across parties. The former can be applied within each party and relies on the current frequency distribution. The latter provides a warm start for each party by integrating global constraints extracted from all parties. These constraints can help participants collect more useful prefixes by disregarding certain prefixes that may appear frequent at initial levels within specific parties but are infrequent on a global view.

\begin{algorithm}[t]
  \footnotesize
  \caption{Target-Aligning Prefix Tree Mechanism (TAP)}\label{TAP}
  \KwIn{Parties $\mathcal{P}$, maximum binary length $m$, granularity $g$, shared trie level $g_s$, query $k$, and privacy budget $\epsilon$}
  \KwOut{Top-$k$ federated heavy hitters $R^k$}
  \tcp{Phase I: shared shallow trie construction}
  $C_{g_s}=STC(\mathcal{P},m,g,g_s,k,\epsilon)$\;

  \tcp{Phase II: independent estimations with a warm start}
  
  \For{$P_i \in \mathcal{P}$}{
  	Based on the remaining $g-g_s$ groups, $P_i$ constructs candidate domain in each level $h\in \{g_s+1,..., g\}$ for estimatation, i.e., 
  	$\Lambda^i_h$ = \texttt{Construct}($l_h, l_{h-1}, C^i_{h-1}$), 
  	
  	$(C^i_{g}, \mathcal{I}^i_{g})$ = \texttt{Estimate}($l_h, \epsilon, \Lambda^i_h, U^i_h$)
  	
  	$P_i$ uploads local heavy hitters $C^i_{g}$ and counts $\mathcal{I}^i_{g}$ to server\;
  }

  The server derives the top-$k$ federated heavy hitters $R^k$ from $\{C^1_{g}, ..., C^{|\mathcal{P}|}_{g}\}$ based on $\{\mathcal{I}^1_{g}, ..., \mathcal{I}^{|\mathcal{P}|}_{g}\}$\;
  
  \Return  $R^k$\;
\end{algorithm}

\subsection{Privacy and Utility Analysis}
\label{subsec:privacy}
We first prove that the proposed TAP mechanism satisfies $\epsilon$-LDP.

\begin{theorem}\label{theorem4}
    The TAP mechanism satisfies $\epsilon$-LDP if it uses an FO that satisfies $\epsilon$-LDP.
\end{theorem}
\begin{proof} 
    In the TAP mechanism, each party randomly divides its users into $g$ groups, which does not raise privacy concerns. Each user in a group receives a candidate domain for perturbation and, given the privacy budget $\epsilon$ , perturbs their value via an FO. We first prove that the perturbation within a party satisfies $\epsilon$-LDP, and then show that candidate domain construction does not consume any additional privacy budget.
    
    For each user in the $h$-th group of a party, the report includes a level index $h$, which is decided independent of her private value, and the noisy prefix $\lambda_h$ perturbed by an FO, denoted by $\mathcal{M}(\cdot)$, with privacy budget $\epsilon$. For any two inputs $x_1, x_2 \in \mathcal{X}$, and for any specific tuple $\langle h, y\rangle \in \text{Range}(\text{TAP})$, we have: $ \max_{x_1,x_2,\langle h, y\rangle} \frac{\Pr[\text{TAP}(x_1) = \langle h, y\rangle]}{\Pr[\text{TAP}(x_2) = \langle h, y\rangle]} = \max_{x_1,x_2,\langle h, y\rangle} \frac{\Pr[\mathcal{M}(\lambda^1_h) = y]\cdot \Pr[h]}{\Pr[\mathcal{M}(\lambda^2_h) = y]\cdot \Pr[h]} = \max_{x_1,x_2,\langle h, y\rangle} \frac{\Pr[\mathcal{M}(\lambda^1_h) = y]}{\Pr[\mathcal{M}(\lambda^2_h) = y]} \leq e^\epsilon, $
    where $\lambda^i_h$ denotes the prefix of $x_i$ at the $h$-th level. 
    $Pr[h]=1/g$ is the probability that a user falls into the $h$-th group, which is chosen randomly and independent of the private value. 
    The last inequality holds because the FO $\mathcal{M}(\cdot)$ used for perturbation satisfies $\epsilon$-LDP.

    In Phase I, the candidate domain is constructed independently at the first level to ensure privacy. For levels $1 \leq h \leq g_s$, domains are derived by extending the top $t$ frequent prefixes from the $(h-1)$-th level, where $g_s$ is determined independently. Since frequencies are perturbed with an FO satisfying $\epsilon$-LDP, the resulting frequency distribution benefits from the post-processing theorem. Other parameters in Equations~\ref{eq9} and ~\ref{eq10} (e.g., user count and FO variance) are non-private, so determining $t$ also preserves LDP. Consequently, candidate domain construction at subsequent levels in Phase I also benefits from the post-processing theorem. 
    In Phase II, the $(g_s+1)$-th level domain is extended based on sanitized global results at the $g_s$ level. The privacy guarantee for levels $g_s+1 \leq h \leq g$ is consistent with that of Phase I. Thus, we complete the proof.
\end{proof}

We also provide a utility analysis on the adaptive extension. The potential failure of our adaptive extension would occur if $t$ is consistently set to the same constant at all levels. Therefore, we analyze the probability under the worst case of when the adaptive extension consistently selects a constant across all levels.

\begin{theorem}
Given the query $k$, the upper bound of the probability of event $A$ that the adaptive extension strategy always sets the extension number $t$ at $g$ iterations to the same constant $c$ is as follows:
\begin{equation}\notag
    Pr[A]\leq{(P_x)}^{g},
\end{equation}
where $P_x=Pr[\phi(\frac{-\delta_f}{2\sigma}) > \frac{2\sqrt{\pi}}{3k+1}]$, $\phi(\cdot)$ denotes the CDF of the Gaussian distribution, $\delta_f$ denotes the maximum difference between two neighboring frequencies of the last $2k-k^*$ prefixes/items out of the frequent $2k$, and $\sigma$ is the variance of the FO.
\end{theorem}
\begin{proof}
For the extension number $t$ at each iteration, it is set to the same constant $c$ when the selected $k^*+\eta= c$ (we assume $k^*\leq c$ has been determined since it highly depends on the distribution), that is, $Pr(k^*)\cdot Pr[min(k, \mathbb{E}(x)) = c - k^*]$. 
When $\mathbb{E}(x)\leq k$:
\begin{footnotesize}
\begin{flalign*}
    &\Pr(k^*=c_1)\cdot \Pr[min(k, \mathbb{E}(x)) = c - c_1]\\
    &=\Pr(k^*=c_1)\cdot \Pr[\mathbb{E}(x)\leq k]\cdot \Pr[\mathbb{E}(x) = c - c_1]\\
    &=\Pr(k^*=c_1)\cdot \Pr[\mathbb{E}(x)\leq k]\cdot \Pr[\frac{(c-c_1)\cdot(1+c_1+c)}{2\sqrt{\pi}}\cdot\phi(\frac{-\delta_f}{2\sigma}) = c - c_1]\\
    &=\Pr(k^*=c_1)\cdot \Pr[\mathbb{E}(x)\leq k]\cdot \Pr[\frac{(1+c_1+c)}{2\sqrt{\pi}}\cdot\phi(\frac{-\delta_f}{2\sigma}) = 1]\\
    &=\Pr(k^*=c_1)\cdot \Pr[\mathbb{E}(x)\leq k]\cdot \Pr[\phi(\frac{-\delta_f}{2\sigma}) = \frac{2\sqrt{\pi}}{1+k^*+c}]\\
    &\leq \Pr(k^*=c_1)\cdot \Pr[\mathbb{E}(x)\leq k]\cdot Pr[\phi(\frac{-\delta_f}{2\sigma}) \geq \frac{2\sqrt{\pi}}{1+k+c}]\\
    &\leq Pr[\phi(\frac{-\delta_f}{2\sigma}) \geq \frac{2\sqrt{\pi}}{1+k+c}]\leq\Pr[\phi(\frac{-\delta_f}{2\sigma}) > \frac{2\sqrt{\pi}}{1+3k}]; 
\end{flalign*}
\end{footnotesize}
when $\mathbb{E}(x) > k$: 
\begin{footnotesize}
\begin{align*}
    &\Pr(k^*=c_1)\cdot \Pr[min(k, \mathbb{E}(x)) = k]=\Pr(k^*=c_1)\cdot \Pr[\mathbb{E}(x) > k]\\
    &=\Pr(k^*=c_1)\cdot \Pr[\frac{(c-c_1)\cdot(1+c_1+c)}{2\sqrt{\pi}}\cdot\phi(\frac{-\delta_f}{2\sigma}) > k]\\
    &=\Pr(k^*=c_1)\cdot \Pr[\phi(\frac{-\delta_f}{2\sigma}) > \frac{2k\sqrt{\pi}}{(c-c_1)\cdot(1+c_1+c)}]\\
    &\leq \Pr(k^*=c_1)\cdot \Pr[\phi(\frac{-\delta_f}{2\sigma}) > \frac{2k\sqrt{\pi}}{k\cdot(1+k+c)}]\\
    &=\Pr(k^*=c_1)\cdot \Pr[\phi(\frac{-\delta_f}{2\sigma}) > \frac{2\sqrt{\pi}}{1+k+c}]\\
    &\leq\Pr[\phi(\frac{-\delta_f}{2\sigma}) > \frac{2\sqrt{\pi}}{1+k+c}]\leq\Pr[\phi(\frac{-\delta_f}{2\sigma}) > \frac{2\sqrt{\pi}}{1+3k}].
\end{align*}
\end{footnotesize}
Since we have $g$ iterations in total, we can obtain $Pr[t\geq c]\leq{(P_x)}^{g}$, where $P_x=Pr[\phi(\frac{-\delta_f}{2\sigma}) > \frac{2\sqrt{\pi}}{3k+1}]$. Thus, we complete the proof.
\end{proof}

We observe that when the query $k$, the maximum difference between two neighboring frequencies among the least frequent $2k-k^*$ items in the top $2k$, and the variance of the used FO are small, the probability of always setting $t$ to a same constant at different iterations is very small.

\section{ Consensus-Pruned Target-Aligning Prefix Tree Mechanism}
\label{Optimization}
Although TAP mitigates the adverse effects of data heterogeneity and inconsistency issues via shared shallow trie construction and adaptive trie extension strategies, the non-IID nature of the data still complicates the estimation and aggregation processes. In phase II, each party independently continues to identify heavy hitters. Unfortunately, the local heavy hitters identified within a party may not contribute to identifying the global ones, which also incurs inconsistencies between local and global objectives.  
To this end, in this section, we propose the \underline{t}arget-\underline{a}ligning \underline{p}refix tree mechanism with the consensu\underline{s}-based pruning strategy (TAPS).

\subsection{Sequential Estimation}
\label{subsec:different_workflow}
The TAPS mechanism integrates TAP and a novel consensus-based pruning strategy that further enhances the target alignment across parties. This mechanism also consists of two phases. Similar to TAP, after phase I, the server first broadcasts the global top $k$ frequent prefixes via the shared shallow trie (steps 1--6 in Figure~\ref{fig_overview_tap}). Then, in phase II, instead of identifying local heavy hitters independently, we propose our consensus-based pruning strategy. 
The motivation for this pruning strategy is that although TAP utilizes both global aggregated constraints and local frequency information in phase I to mitigate the inconsistencies, estimations in phase II do not consider the impact of data distributions from other parties.

Intuitively, such prior knowledge can assist the current party in filtering out some infrequent prefixes before its estimation, which reduces the domain size to reduce the LDP noises, thus leading to more accurate estimations. 
However, it is infeasible for a party's pruning procedure to receive assistance from all other parties simultaneously, as such information would necessitate completing the entire estimation or incur prohibitive communication costs. Therefore, we optimize the workflow in phase II of TAP to a sequential order, which circumvents this issue and mitigates the accumulation of errors. 
Since federated heavy hitter analytics represents a variant of the frequency estimation problem, we believe that the party with a larger user base is likely to contribute more significantly to pruning. Therefore, we sort the parties in descending order based on user population sizes. We use $\mathcal{P}^R$ to denote the sorted one, and each party sequentially performs the estimation in this order.

Assume the entire candidate domain of current level consists of all leaves in Figure~\ref{fig_overview}. Each of the $i$-th party $P_i$ initially receives some candidates from the server (sent by $P_{i-1}$, which has a larger user population) that may not need to be pruned entirely in the current party (step 6). Then, $P_i$ employs a consensus-based test to validate these candidates and prunes the necessary ones (steps 7--8). Subsequently, it proceeds with the heavy hitter identification process within the pruned domain and collects pruning candidates of this level for the next party $P_{i+1}$. For the first party with the most users but without previous input, it independently performs the estimation. 
Upon completion of the identification process by $P_i$, it forwards the local heavy hitters and candidates for pruning to the server. The server then sends these candidates to $P_{i+1}$, initiating a repeat of this process. Note that adaptive trie extension is always employed during the frequency estimation by each party and level. Once the server obtains the local heavy hitters from the last party, it aggregates these findings to determine the federated heavy hitters. 

\subsection{Consensus-based Pruning Strategy}
\label{subsec:pruning_strategy}

\begin{figure}[t]
    \centering
    \includegraphics[width=\columnwidth]{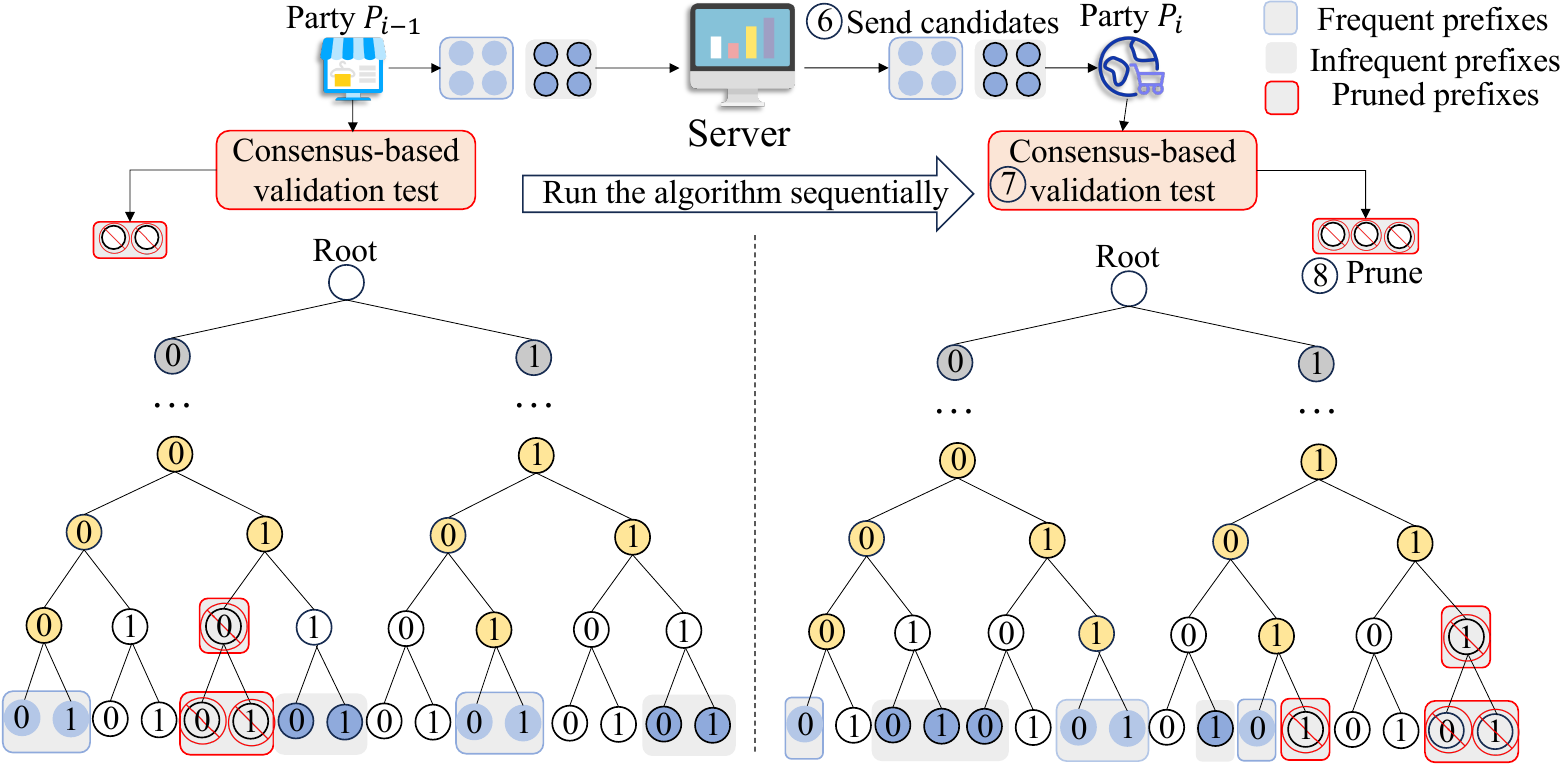}
    \vspace{-0.2in}
    \caption{Consensus-based Pruning Strategy.}
    \vspace{-0.1in}
    \label{fig_overview}
\end{figure}

Due to the non-IID issue, the infrequent prefixes which can be pruned include two types. The first type refers to the globally infrequent items, which cannot be the targets. The second type includes items that are extremely infrequent in the current party (although it may be frequent in other parties), which contributes little to frequency estimation but increases LDP noise scale.

\begin{algorithm}[t]
    \footnotesize
    \caption{Target-Aligning Prefix Tree Mechanism with the Consensus-based Pruning Strategy (TAPS)}\label{TAPS}
    \KwIn{Parties $\mathcal{P}$, maximum binary length $m$, granularity $g$, shared level $g_s$, dividing ratio $\beta$, query $k$, and privacy budget $\epsilon$}
    \KwOut{Top-$k$ federated heavy hitters $R^k$}
    \tcp{Phase I: shared shallow trie construction}
    $C_{g_s}=STC(\mathcal{P},m,g,g_s,k,\epsilon)$\;

    \tcp{Phase II: continue the construction with pruning}
    Sort the parties in a descending order as $\mathcal{P}^R$\;

    \For{$P_i \in \mathcal{P}^R$}{
      $P_i$ initiates the pruning dictionary $D_i=\{\}$ and $C^i_{g_s}=C_{g_s}$\;

      \For{$g_s+1\leq h\leq g$}{
      $P_i$ initiates $\tilde{U}^i_h=U^i_h$ and $\hat{\Lambda}^{i-1,i}_{h}=\emptyset$\;
      \uIf{$g-g_s\leq h\leq g \text{ or } g_s+1\leq h\leq 2g_s$}{
          \uIf{$i>1$}{
        $P_{i}$ divides the users at $h$-th level into $\hat{U}^{i}_{h,0}$= $\hat{U}^{i}_{h,1}=\beta\cdot U^i_h$, and $\tilde{U}^i_h=U^i_h-\hat{U}^{i}_{h,0}-\hat{U}^{i}_{h,1}$\;
        $P_{i}$ asks users in $\hat{U}^{i}_{h,0}$ and $\hat{U}^{i}_{h,1}$ to estimate the frequencies of pruning candidates in $D_{i-1}[h]$\;
        $P_{i}$ obtains the pruning set $\hat{\Lambda}^{i-1,i}_{h}=\hat{\Lambda}^{i-1,i}_{h,0}\cup\hat{\Lambda}^{i-1,i}_{h,1}$ by Equations~\ref{eq_k^prime}--\ref{eq_k_prime_intersection_1}\;
        
      }
      $\Lambda^i_h$=\texttt{Construct}($l_h, l_{h-1}, C^i_{h-1}$)\;
      $(C^i_h,\mathcal{I}^i_h)=$\texttt{Estimate}($l_h, \epsilon, h, \Lambda^i_h-\hat{\Lambda}^{i-1,i}_{h}, \tilde{U}^i_h$)\;
      \uIf{$i<|\mathcal{P}^R|$}{
        Select $D_i[h]=\Delta^{i}_{h}$, where $\Delta^{i}_{h}$ is from Equation~\ref{eq_delta}\;
      }
          }\Else{ 
          $\Lambda^i_h$=\texttt{Construct}($l_h, l_{h-1}, C^i_{h-1}$)\; $(C^i_h,\mathcal{I}^i_h)=$\texttt{Estimate}($l_h, \epsilon, h, \Lambda^i_h, \tilde{U}^i_h$)\;
          }
          }
          \uIf{$D_{i}.keys\neq\emptyset$}{
          $P_i$ uploads $D_i$ to the server\;
          \uIf{$i<|\mathcal{P}^R|$}{The server sends $D_i$ to the next party\;}
          
          }
        $P_i$ uploads local heavy hitters $C^i_{g}$ and counts $\mathcal{I}^i_g$ to server\;
      }
  The server derives the top-$k$ federated heavy hitters $R^k$ from $\{C^1_{g}, ..., C^{|\mathcal{P}|}_{g}\}$ based on $\{\mathcal{I}^1_{g}, ..., \mathcal{I}^{|\mathcal{P}|}_{g}\}$\;
  \Return  $R^k$\;
\end{algorithm}

For the $(i-1)$-th party $P_{i-1}$, it needs to submit candidates of two types at the $h$-th level~\footnote{Here we take the simplest example that only applies this pruning at the $h$-th level for ease of understanding.}. It first completes the estimation (based on the pruned domain) and obtains the noisy prefixes $C^{i-1}_h$ (sorted by the descending frequencies) along with their estimated frequencies. Then it selects the pruning candidates as follows:
	\begin{flalign}\label{eq_delta}
		\Delta^{i-1}_h&=\{\Delta^{i-1}_{h,0},\Delta^{i-1}_{h,1}\}, \\
		\Delta^{i-1}_{h,0}&=\{\lambda^j| Rank(f_j)=|C^{i-1}_h|-j, \lambda^j\in C^{i-1}_h, 1\leq j\leq 2k\}, \nonumber\\
		\Delta^{i-1}_{h,1}&=\{(\lambda^j, f_j)| Rank(f_j)=j, \lambda^j\in C^{i-1}_h, 1\leq j\leq 2k\}, \nonumber
	\end{flalign}
where $\Delta^{i-1}_{h,0}$ and $\Delta^{i-1}_{h,1}$ represent the infrequent and frequent candidate sets of $P_{i-1}$ at $h$-th level, $f_j$ denotes the frequency of the prefix $\lambda^j$, and $Rank(f_j)=j$ means $f_j$ is $j$-th largest in $C^{i-1}_h$.  
Inspired by~\cite{WLJ2018}, we set $|\Delta^{i-1}_{h,0}|$ and $|\Delta^{i-1}_{h,1}|$ to $2k$. Because the goal here is to provide the next party with both exceptionally frequent and infrequent pruning candidates (specifically, the first $k$ out of $2k$), thereby necessitating the inclusion of weaker candidate prefixes (the remaining ones in each $2k$) to introduce some randomness for the consensus-based validation test in $P_i$. Note that $|C^{i-1}_h|\succ2k$ if $t$ is not too small, since $|C^{i-1}_h|=t\cdot2^{l_h-l_{h-1}}$, where $l_h =\lceil \frac{h \times m}{g} \rceil$. Then $\Delta^{i-1}_h$ is sent to the server, who forwards it to $P_i$. 

For $P_i$, the consensus-based validation test is invoked to filter the candidates for pruning. The received $\Delta^{i-1}_h$ is regarded as the pruning predictions from $P_{i-1}$, which were estimated using most of its users $U^{i-1}_h$ at level $h$ that can be treated as the ``training set''. Thus, $P_i$ validates $\Delta^{i-1}_h$ using a small fraction of its users $U^i_h$ at level $h$, acting as the validation set for the consensus-based validation test. It then prunes the prefix domain for the `test set' (i.e., the remaining majority of users in $U^{i}_h$). 
Concretely, as shown in Algorithm~\ref{TAPS}, $P_i$ requires each user in $\hat{U}^{i}_{h,0}$ (resp. $\hat{U}^{i}_{h,1}$) to report her noisy prefix via an FO, then estimates the frequency of all prefixes in $\Delta^{i-1}_{h,0}$ (resp. $\Delta^{i-1}_{h,1}$) with the $l_h$ bits length (lines 9--10). 
Here, $\hat{U}^{i}_{h,0}=\hat{U}^{i}_{h,1}=\beta\cdot U^i_h$ represents the validation users of $P_i$ for $\Delta^{i-1}_{h,0}$ and $\Delta^{i-1}_{h,1}$, respectively, where $\beta$ is the dividing ratio. 
$P_i$ then sorts the results by frequencies in ascending order to obtain $\hat{\Delta}^{i-1,i}_{h,0}$ and $\hat{\Delta}^{i-1,i}_{h,1}$.

Although $P_i$ possesses the validated results, it is non-trivial to independently determine which prefixes should be pruned, since the entire frequency distribution at the $h$-th level is unknown. This uncertainty makes it difficult to discern the boundary between infrequent and frequent prefixes. While one could establish a fine-tuned threshold for pruning, the frequencies may be subject to instability due to the sampling errors from grouping users and the LDP noise. Moreover, the threshold might vary based on the underlying distributions, which differ across parties. Consequently, we rely on the rankings in both $P_{i-1}$ and $P_i$, rather than using the frequencies directly, to adaptively decide the consensus pruning set, which mitigates the aforementioned issues. For the first type of infrequent prefixes (those are globally infrequent), we can filter the pruning prefixes as follows:

\begin{flalign}
    &\underset{k^\prime}{\text{arg max}}\quad \frac{|\hat{\Lambda}^{i-1,i}_{h,0}|}{k^\prime \cdot (1 + \epsilon)^{k^\prime}} - \gamma \cdot \alpha^2,\label{eq_k^prime}\\
    &\hat{\Lambda}^{i-1,i}_{h,0} = \Delta^{i-1}_{h,0}[:k^\prime] \cap \hat{\Delta}^{i-1,i}_{h,0}[:k^\prime],\label{eq_k_prime_intersection}
\end{flalign}
where $1 \leq k^\prime \leq k$, $\gamma = \left(1 - \frac{|U^{i-1}|}{\sum\nolimits_{j=1}^{|\mathcal{P}|}|U^{j}|}\right)^2$, and $\alpha = \frac{k^\prime - |\hat{\Lambda}^{i-1,i}_{h,0}| + 1}{k^\prime + 1}$. In Equation~\ref{eq_k^prime}, the formula comprises two components: the intersection score and the penalty term. The intersection score, represented by $|\hat{\Lambda}^{i-1,i}_{h,0}|/k^\prime$, indicates the proportion of the first common $k^\prime$ prefixes in $\Delta^{i-1}_{h,0}$ and $\hat{\Delta}^{i-1,i}_{h,0}$. To minimize the inclusion of too many false positives, this score is adjusted by dividing by $ (1 + \epsilon)^{k^\prime}$, thus to refrain from always obtaining a large $k^\prime$ and provide the non-linear property. Because when $\epsilon$ is relatively small, which might result in inaccurate $\Delta^{i-1}_{h,0}$ and $\hat{\Delta}^{i-1,i}_{h,0}$ estimations, $k^\prime$ should not be too large. Regarding the penalty term, it consists of two factors. $\gamma$ denotes the population confidence from the previous party, calculated according to the users' population. The more users in the $P_{i-1}$, the smaller $\gamma$ value. The second factor, $\alpha$, accounts for the non-intersection ratio within the first $k^\prime$ prefixes of $\Delta^{i-1}_{h,0}$ and $\hat{\Delta}^{i-1,i}_{h,0}$. Through the $\arg \max$ operation, the chosen $k^\prime$ serves as an optimal boundary, thereby maximizing the consensus confidence of the infrequent prefixes ranked in first $k^\prime$ prefixes. 
Once the $k^\prime$ is determined, the intersection $\hat{\Lambda}^{i-1,i}_{h,0}$ becomes the consensus pruning prefix set for the first type. This set includes the infrequent prefixes selected by both $P_{i-1}$ and $P_i$, achieving an optimal trade-off between consensuses and differences in $\hat{\Lambda}^{i-1,i}_{h,0}$, as defined by Equation~\ref{eq_k^prime}.  

For the second type of infrequent prefixes, i.e., extremely infrequent or even absent in the current party but may be frequent in other parties, we introduce the frequency contrast scores as the measurement:

\begin{small}
\begin{equation}\label{eq_frequency_contrast_set}
    \hat{\Delta}^{i-1,i}_{h,c}\!=\!\{(\lambda^j, f^c_j)| Rank(f^c_j)\!=\!j, f^c_j\!=\!\varphi(\lambda_j), \lambda^j \! \in \! \Delta^{i-1}_{h,1}, 1\leq j\leq 2k\}, 
\end{equation}
\end{small}
where $\varphi(\lambda_j)=\frac{\Delta^{i-1}_{h,1}(\lambda_j)}{\hat{\Delta}^{i-1,i}_{h,1}(\lambda_j)+\tau}$, and $\tau=1\times 10^{-11}$ is to avoid dividing by zero, $\Delta^{i\!-\!1}_{h,1}(\lambda_j)$ (resp. $\hat{\Delta}^{i\!-\!1,i}_{h,1}(\lambda_j)$) denotes the frequency of the prefix $\lambda_j$ in $\Delta^{i\!-\!1}_{h,1}$ (resp. $\hat{\Delta}^{i\!-\!1,i}_{h,1}$). 
In Equation~\ref{eq_frequency_contrast_set}, we initially calculate the frequency contrast score for each $\lambda_j\in \Delta^{i-1}_{h,1}$. This score quantifies the frequency discrepancy of $\lambda_j$ between $P_i$ and $P_{i-1}$. A higher $f^c_j$ suggests that the prefix $\lambda_j$ is more prevalent in $P_{i-1}$ but significantly less so in $P_i$. For instance, if the best-selling item is only sold by the smallest party with the fewest users, e.g., the frequency of item $X$ is $0.001$ in the previous party and $0.7$ in the current party (i.e., most sold), this score would be $0.001$, which is relatively low. Conversely, if item $Y$’s frequency is $0.7$ in the previous party and $0.001$ in the current party (i.e., almost nonexistent), the score would be $700$, making it much more easier to identify. 
Then we arrange $\hat{\Delta}^{i-1,i}_{h,c}$ in descending order and feed it with $\hat{\Delta}^{i-1,i}_{h,1}$ into Equation~\ref{eq_k_prime_intersection}:

\begin{equation}\label{eq_k_prime_intersection_1}
    \hat{\Lambda}^{i-1,i}_{h,1} = \hat{\Delta}^{i-1,i}_{h,c}[:k^\prime] \cap \hat{\Delta}^{i-1,i}_{h,1}[:k^\prime].
\end{equation}

By substituting $\hat{\Lambda}^{i-1,i}_{h,0}$ with $\hat{\Lambda}^{i-1,i}_{h,1}$, we solve Equation~\ref{eq_k^prime} to derive the pruning set of infrequent prefixes $\hat{\Lambda}^{i-1,i}_{h}=\hat{\Lambda}^{i-1,i}_{h,0}\cup\hat{\Lambda}^{i-1,i}_{h,1}$. Note that because of descending order, the items with lower scores should not be selected. 
The current party then prunes the candidate domain, performs the estimation, and uploads its pruning candidates to the server for the pruning in the next party at this level as well (lines 11--15 in Algorithm~\ref{TAPS}).

We finally analyze the communication and computation costs of TAPS mechanism in comparison with two baselines (e.g., GTF and FedPEM) and the (infeasible) method of directly uploading all noisy answers perturbed by an FO (i.e., OUE or OLH), to the central server. We assume that each pair of a prefix/item and its count consumes $b$ bits, $g^*$ represents the number of iterations where the pruning strategy is applied, $|\mathcal{U}|$ is the user population size, and $|\mathcal{X}|$ is the overall domain size. As shown in Table~\ref{table_communication_cost}, the baselines incur a communication cost of $O(b\cdot k\cdot |\mathcal{P}|)$ and a computation cost of $O(k\cdot |\mathcal{P}|)$, as each party uploads only the top $k$ local heavy hitters and the server only needs to count and sort them. The method of uploading user data (i.e., via OUE or OLH) for estimating heavy hitters is impractical due to the large user population and item domain, resulting in prohibitive communication costs increase linearly with $|\mathcal{U}|$ (for OUE), and excessive computation costs $O(|\mathcal{U}|\cdot|\mathcal{X}|)$ (for both OUE and OLH)~\cite{WLJ2019,CMM2021}. As a comparison, our mechanism TAPS incurs significantly less communication and computation costs. On the other hand, although TAPS still consumes a bit more communication costs~\footnote{The number of iterations $g^*$ usually selects a small value, \textcolor{black}{e.g., $g/2$}.}
than two baselines GTF and FedPEM, its utility performance is significantly better than the baselines, as will be shown in Section~\ref{experiment}. This can be attributed to our design of mitigating non-IID issues in federated settings. 

\subsection{Privacy Analysis}
\label{subsec:privacy_TAPS}
Finally, we establish the privacy guarantee of TAPS mechanism.

\begin{theorem}\label{theorem5}
    The TAPS mechanism satisfies $\epsilon$-LDP if it uses an FO that satisfies $\epsilon$-LDP.
\end{theorem}
\begin{proof}
    In the TAPS mechanism, the privacy guarantee during phase I remains consistent with that of TAP. In phase II, each party further splits each group of users at the first and last $g_s$ levels into two set, namely, $\beta\cdot U_h$ for the consensus pruning test and $U_h-\beta\cdot U_h$ for the main estimation, where $\beta$ is the independently dividing ratio. Each user within the respective set perturbs her value/prefix using an FO that satisfies $\epsilon$-LDP, based on a predefined candidate domain. Similar to TAP, all candidate domains are established based on the sanitized results, thereby leveraging the post-processing property. Therefore, the TAPS mechanism satisfies $\epsilon$-LDP.
\end{proof}

\begin{table}[t]
\centering
\caption{Communication and computation costs.}
\vspace{-0.1in}
\resizebox{\columnwidth}{!}{\begin{tabular}{c|c|c|c|c|c}
\hline
     & GTF & FedPEM & OUE &OLH & TAPS \\ \hline
 Communication & $O(b\cdot k\cdot |\mathcal{P}|)$   & $O(b\cdot k\cdot |\mathcal{P}|)$ &$O(|\mathcal{U}|\cdot |\mathcal{X}|)$        & $O(b\cdot |\mathcal{U}|)$  & $O(b\cdot k\cdot |\mathcal{P}|\cdot g^*)$  \\ \hline
 
 Computation & $O(k\cdot |\mathcal{P}|)$   & $O(k\cdot |\mathcal{P}|)$    &$O(|\mathcal{U}|\cdot |\mathcal{X}|)$        & $O(|\mathcal{U}|\cdot |\mathcal{X}|)$  & $O(k\cdot |\mathcal{P}|)$  \\ \hline
\end{tabular}}
\vspace{-0.1in}
\label{table_communication_cost}
\end{table}
\section{Performance Evaluations}\label{experiment}
\subsection{Experimental Setting}

\begin{table}[t]
\centering
\caption{Datasets.}
\vspace{-0.15in}
\resizebox{\columnwidth}{!}{\begin{tabular}{c|c|c|c|c|c}
\hline
& Dataset description & \# users & \# total users & \# unique items & \# common items \\ \hline
\multirow{2}{*}{RDB} & Reddit~\cite{KSV2018reddit} (Comments) & 252,830 & \multirow{2}{*}{352,830} & 30,550 & \multirow{2}{*}{8,047} \\ \cline{2-3}\cline{5-5}
                           & IMDB~\cite{MDP2011imdb} (Movie reviews) & 100,000 &    & 15,470 &   \\ \hline
\multirow{4}{*}{YCM} & Yahoo~\cite{ZZL2015yelp_agnews_yahoo} (Q\&A) & 812,300 & \multirow{4}{*}{1,336,048} & 79,971 & \multirow{4}{*}{3,879} \\ \cline{2-3} \cline{5-5}
                           & CNN Dailymail~\cite{KTE2015cnndailymail,SLM2017cnndailymail} (News articles) & 287,113 &   & 32,162  &  \\ \cline{2-3} \cline{5-5}
                           & Mind~\cite{WQC2020mind} (Microsoft news) & 123,082  &  & 17,309  &  \\ \cline{2-3} \cline{5-5}
                           & SWAG~\cite{ZBS2018swag} (Video captions) & 113,553 &   & 7,656 &  \\ \hline
\multirow{6}{*}{TYS} & Twitter~\cite{GBH2009twitter} (Tweets)  & 658,549 & \multirow{6}{*}{2,119,776} & 80,126 & \multirow{6}{*}{2,175} \\ \cline{2-3} \cline{5-5}
                           & Yelp~\cite{ZZL2015yelp_agnews_yahoo} (Reviews) & 649,917 &  & 34,866  & \\ \cline{2-3} \cline{5-5}
                           & Scientific Papers~\cite{CDK2018scientificpaper} (Papers) & 349,119 &  & 27,372  &  \\ \cline{2-3} \cline{5-5}
                           & Amazon Arts~\cite{NLM2019amazon} (Reviews) & 200,000 &  & 8,914  &  \\ \cline{2-3} \cline{5-5}
                           & SQuAD~\cite{RJL2018squad,RZL2016squad} (Q\&A) & 142,192 &  & 19,895  &  \\ \cline{2-3} \cline{5-5}
                           & AG News~\cite{ZZL2015yelp_agnews_yahoo} (News articles) & 119,999 &  & 15,879 &  \\ \hline
\multirow{6}{*}{UBA} & UBA 0 (shopping interactions) & 1,476,546 & \multirow{6}{*}{6,483,370} & 162,833 & \multirow{6}{*}{975} \\ \cline{2-3} \cline{5-5}
                           & UBA 1 (shopping interactions) & 1,263,768 &  & 167,196  & \\ \cline{2-3} \cline{5-5}
                           & UBA 2 (shopping interactions) & 1,246,972 &  & 167,309  &  \\ \cline{2-3} \cline{5-5}
                           & UBA 3 (shopping interactions) & 1,117,376 &  & 58,087  &  \\ \cline{2-3} \cline{5-5}
                           & UBA 4 (shopping interactions) & 774,626 &  & 9,203  &  \\ \cline{2-3} \cline{5-5}
                           & UBA 5 (shopping interactions) & 604,082 &  & 4,979 &  \\ \hline
\multirow{8}{*}{SYN} & SYN 0, Poisson ($\lambda=10$) & 220,000 & \multirow{8}{*}{780,000} & 5,484 & \multirow{8}{*}{276} \\ \cline{2-3} \cline{5-5}
                           & SYN 1, Poisson ($\lambda=8$) & 170,000 &  & 6,260  &  \\ \cline{2-3} \cline{5-5}
                           & SYN 2, Zipf ($\alpha=1.1$) & 120,000 &  & 6,688  &  \\ \cline{2-3} \cline{5-5}
                           & SYN 3, Zipf ($\alpha=1.3$) & 80,000 &  & 6,796  &  \\ \cline{2-3} \cline{5-5}
                           & SYN 4, Poisson ($\lambda=6$) & 70,000 &  & 4,590  &  \\ \cline{2-3} \cline{5-5}
                           & SYN 5, Poisson ($\lambda=4$) & 60,000 &  & 4,429  &  \\ \cline{2-3} \cline{5-5}
                           & SYN 6, Zipf ($\alpha=1.5$) & 30,000 &  & 4,813  &  \\ \cline{2-3} \cline{5-5}
                           & SYN 7, Zipf ($\alpha=1.7$) & 30,000 & & 5,024 &  \\ \hline
\end{tabular}}
\vspace{-0.1in}
\label{table_dataset}
\end{table}

{\bf Datasets}. 
In the experiments, we use four real-world datasets (RDB, YCM, TYS, UBA) and one synthetic dataset (SYN). Each real-world dataset simulates heterogeneous parties across different application scenarios, e.g., RDB is designed to facilitate comparisons with baselines and simulate a typical heavy-hitter identification scenario that identifying the most frequent ``out-of-vocabulary'' words typed on keyboards~\cite{CCD2023}. RDB involves two parties --- comments from Reddit~\cite{KSV2018reddit} and movie reviews from IMDB~\cite{MDP2011imdb}. YCM comprises four distinct datasets including Yahoo~\cite{ZZL2015yelp_agnews_yahoo}, CNN Dailymail~\cite{KTE2015cnndailymail,SLM2017cnndailymail}, Mind~\cite{WQC2020mind}, and SWAG~\cite{ZBS2018swag}, which feature a variety of Q\&A pairs, news articles, and video captions. TYS incorporates six parties, namely Twitter~\cite{GBH2009twitter}, Yelp~\cite{ZZL2015yelp_agnews_yahoo}, Scientific Papers~\cite{CDK2018scientificpaper}, Amazon Arts~\cite{NLM2019amazon}, SQuAD~\cite{RJL2018squad,RZL2016squad}, and AG News~\cite{ZZL2015yelp_agnews_yahoo}, representing diverse text descriptions such as tweets, reviews, academic papers, Q\&A, and news articles. 
UBA is the user behaviour dataset from Alibaba~\cite{ZMF2019,PBZ2019,PZZ2020}, which is one of the largest recommender system datasets from real industrial scenes. Each record in UBA represents a specific and private user-item interaction. We randomly select and allocate approximately $6.5$ million users across six different parties. 
For synthetic dataset SYN, we follow the existing studies~\cite{LDC2022,YAG2019} and use a Dirichlet distribution to allocate different item domains for eight parties from the Tmall~\cite{tmall} dataset, which contains a vast array of anonymous shopping logs, into each party. We divide and sample all items into $N=6$ groups. For each party $P_i$, we sample $q\sim Dir_N(\beta=0.5)$, and allocate a $q_j$ proportion of the total items in group $j$ to construct item domain in $P_i$, where $\beta$ controls the imbalance level of the domain skew. Due to the small concentration parameter (0.5) of the Dirichlet distribution, some sampled item domains may not have any items of certain groups of item. We employ Zipf and Poisson distributions with different parameters to build the frequency distribution. Each user in a party holds only a single word or item, and multiple occurrences are sampled as one. We summarize the details in Table~\ref{table_dataset}.

\begin{figure}[t]
    \centering
    \includegraphics[width=0.5\columnwidth]{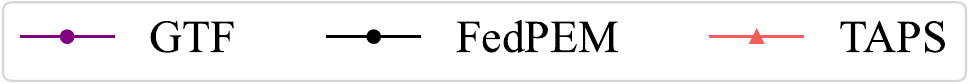} \\ 
    \vspace{-0.1in}
    \subfloat[RDB, $k$=10]{
    	\includegraphics[width=0.32\columnwidth]{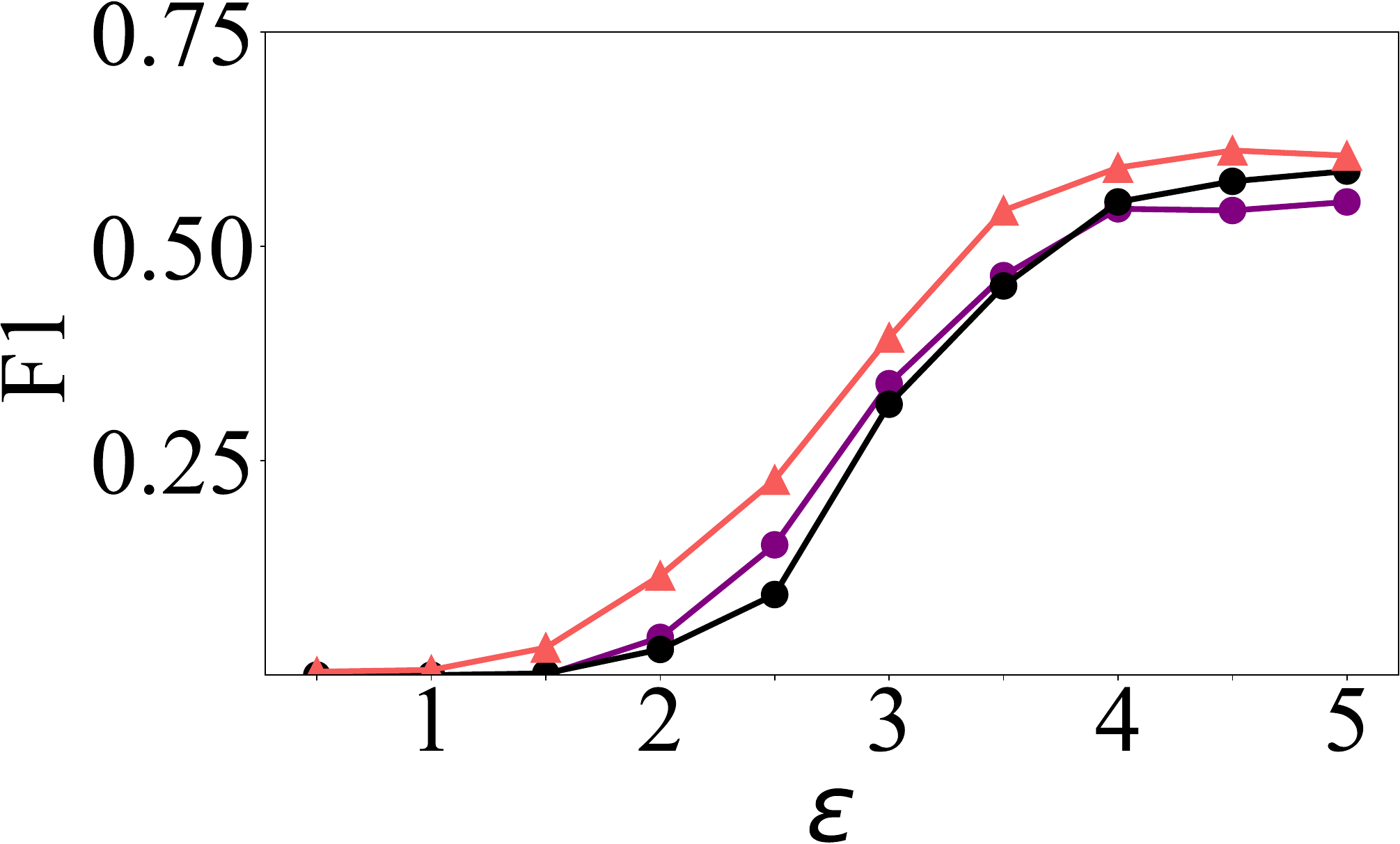}
    }
    \subfloat[RDB, $k$=20]{
    		\includegraphics[width=0.31\columnwidth]{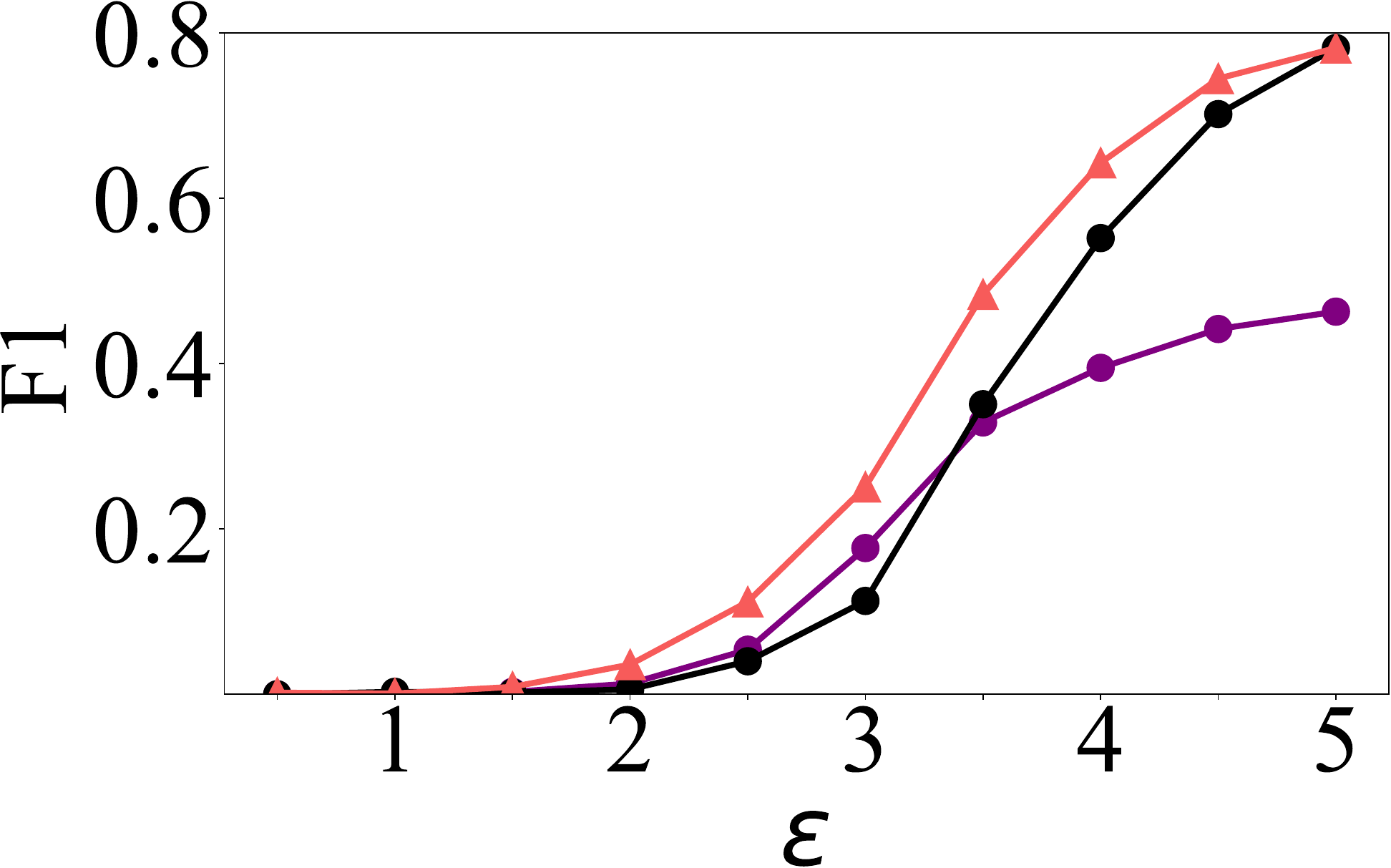}
    	}
     \subfloat[RDB, $k$=40]{
    		\includegraphics[width=0.31\columnwidth]{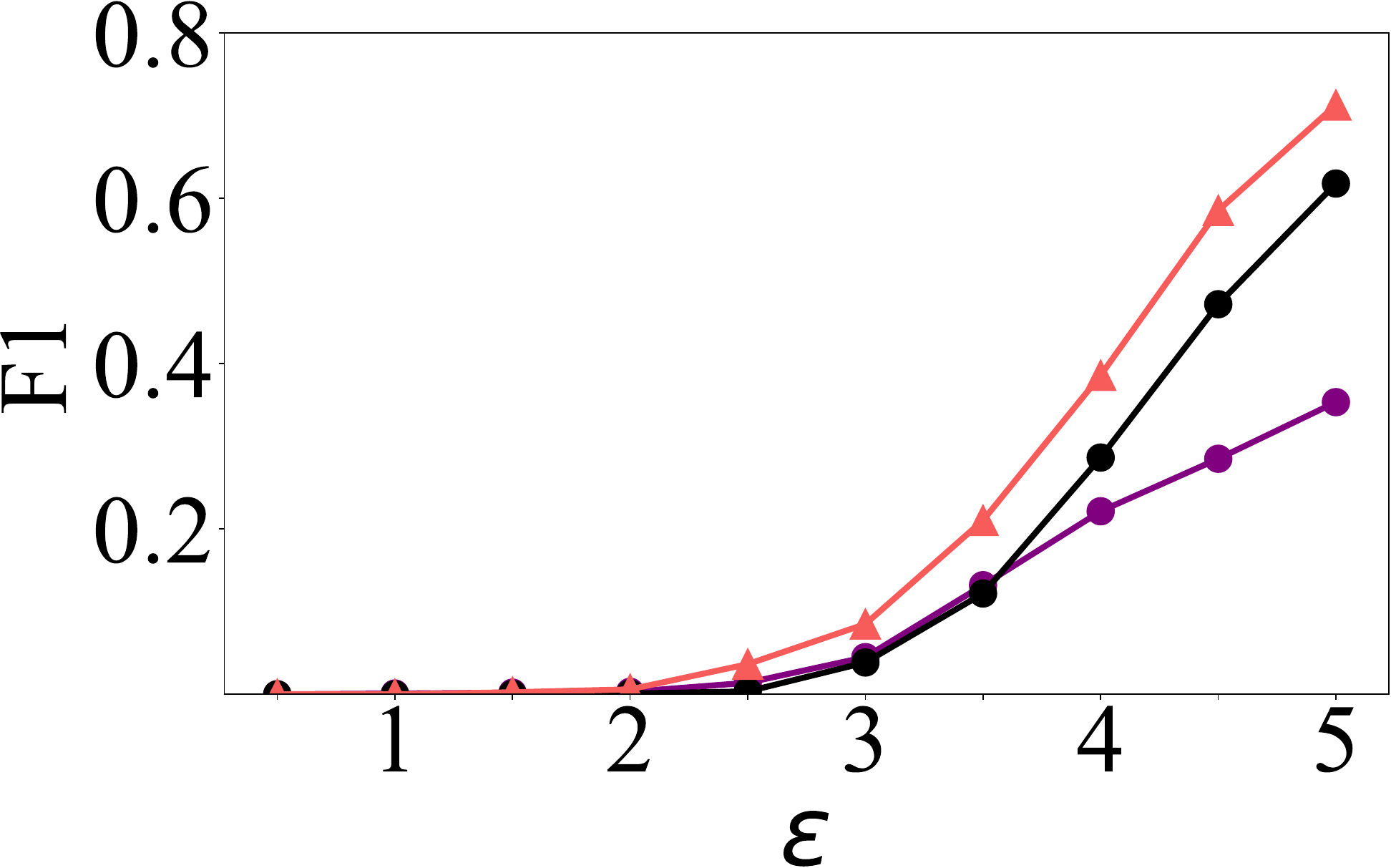}
    	}\hfill
     \vspace{-0.1in}
	\subfloat[YCM, $k$=10]{
		\includegraphics[width=0.3\columnwidth]{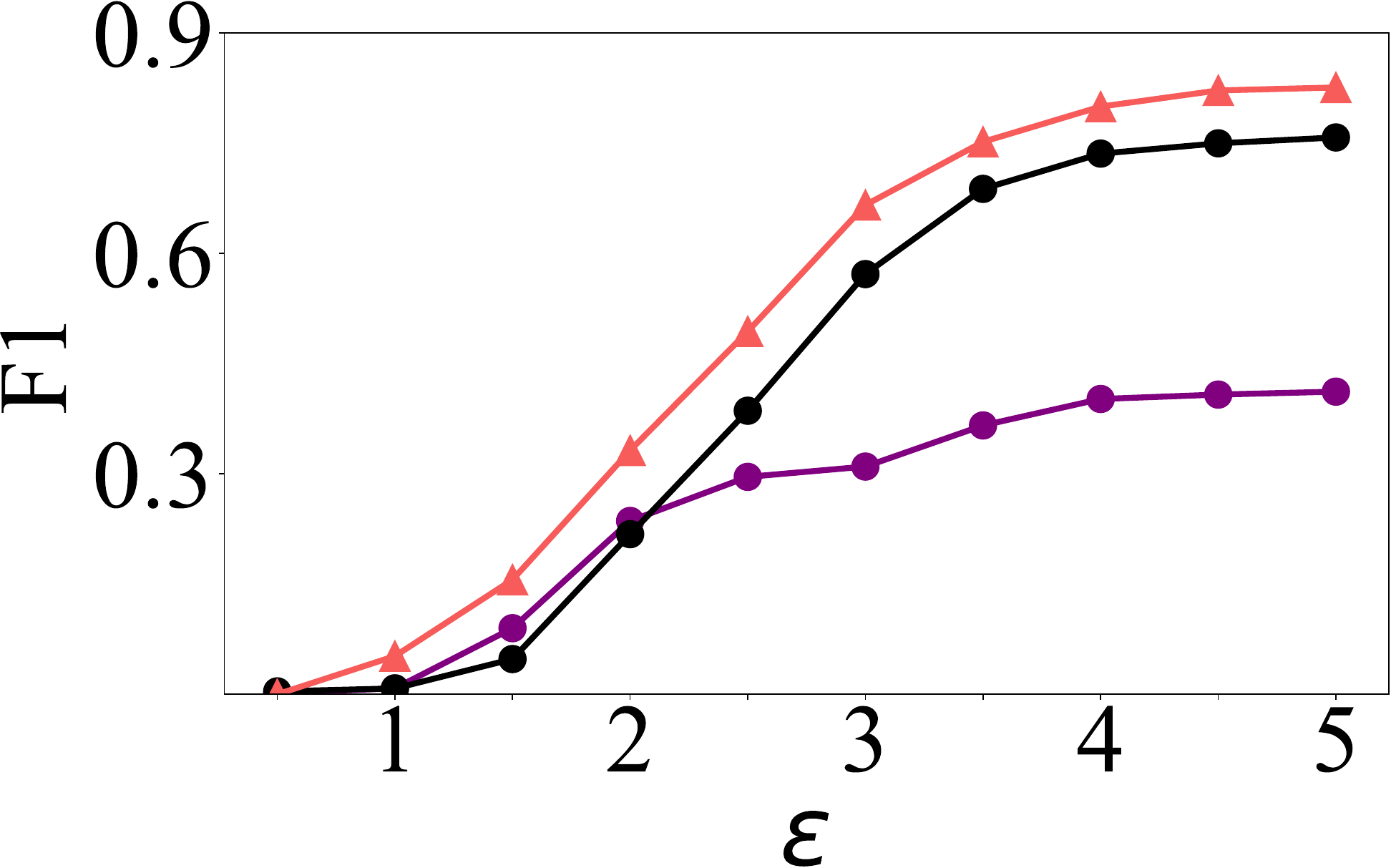}
	}
        \subfloat[YCM, $k$=20]{
    		\includegraphics[width=0.31\columnwidth]{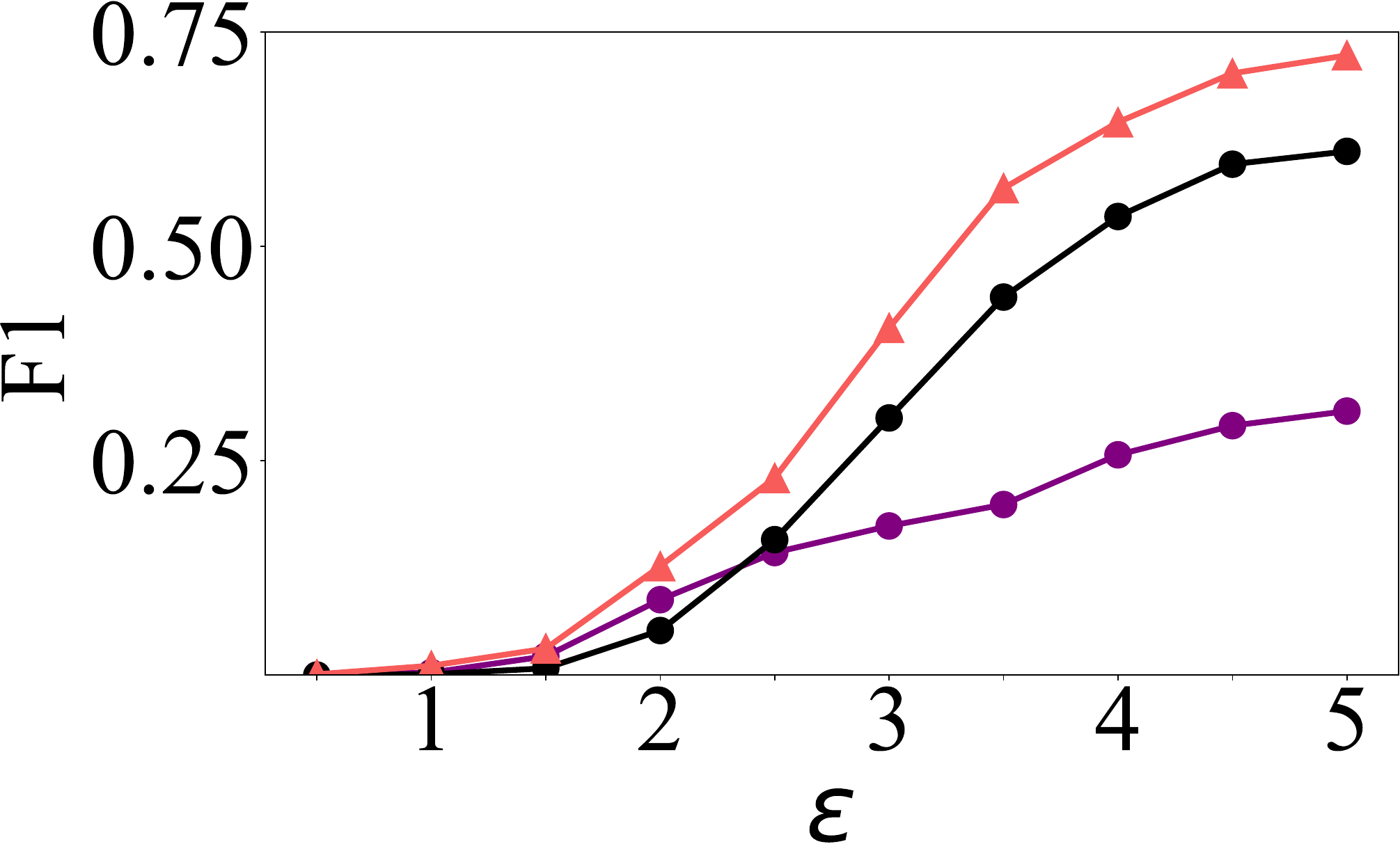}
    	}
     \subfloat[YCM, $k$=40]{
        	\includegraphics[width=0.31\columnwidth]{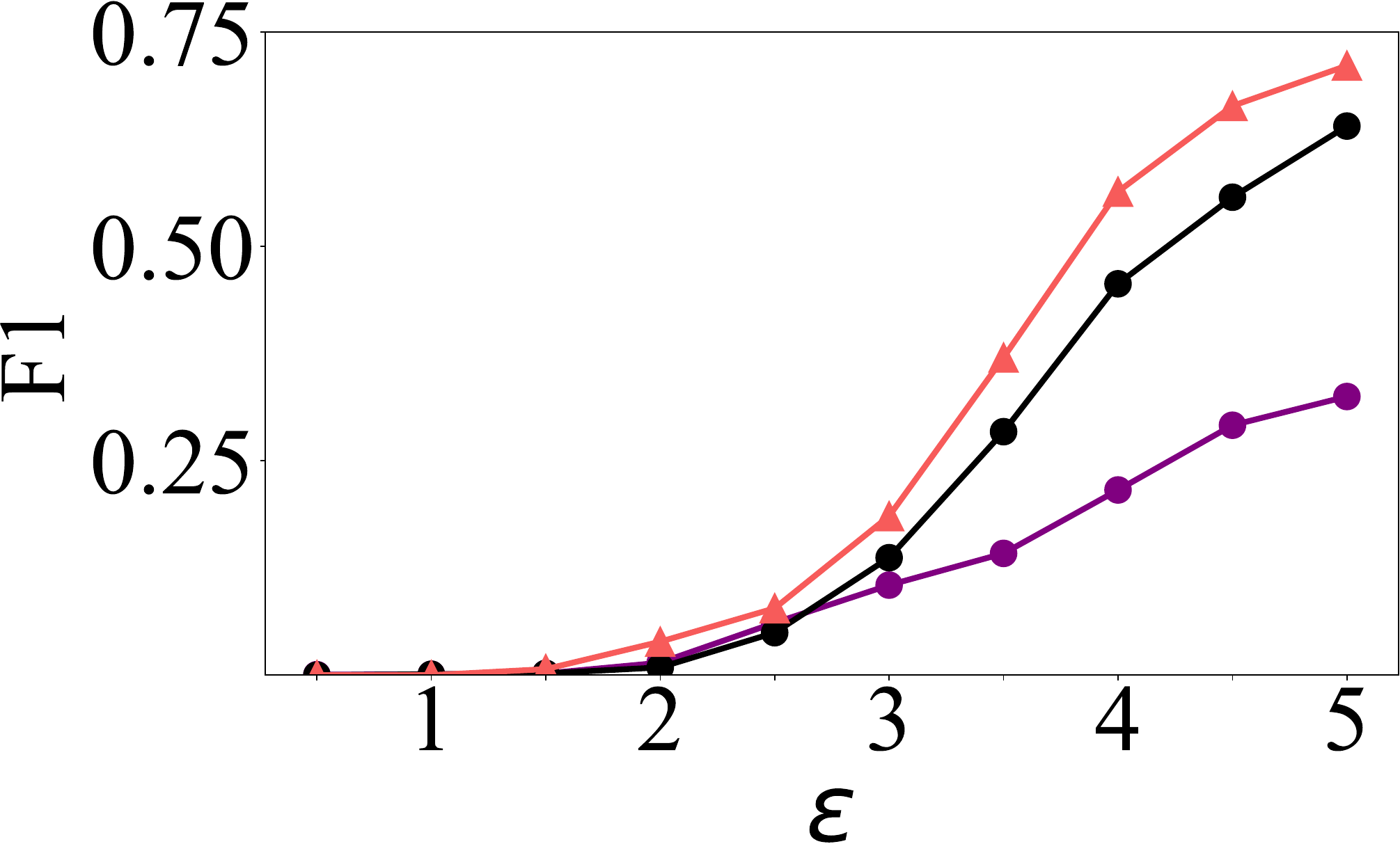}
        }\hfill
     \vspace{-0.1in}
	\subfloat[TYS, $k$=10]{
		\includegraphics[width=0.31\columnwidth]{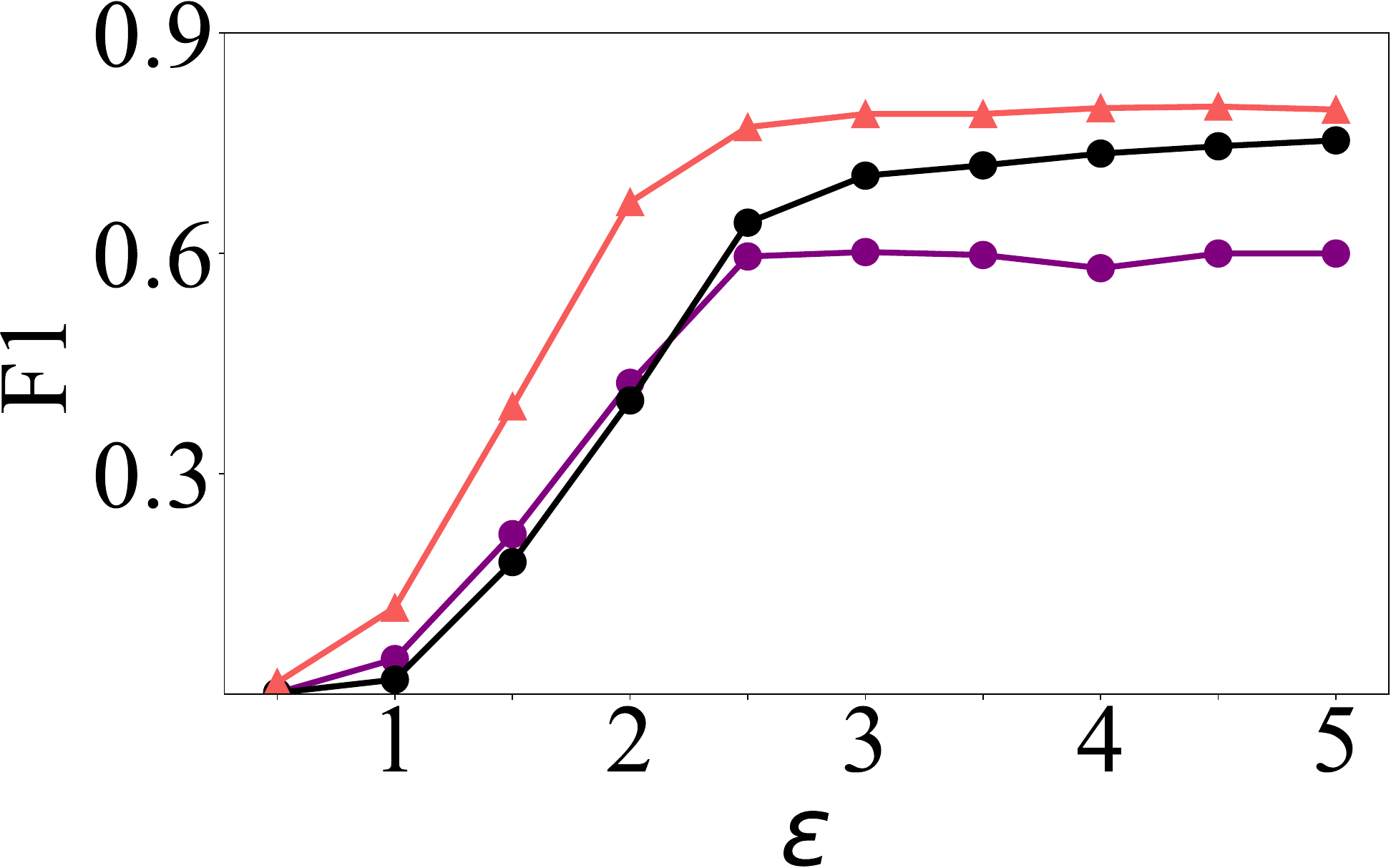}
	}
        \subfloat[TYS, $k$=20]{
        	\includegraphics[width=0.31\columnwidth]{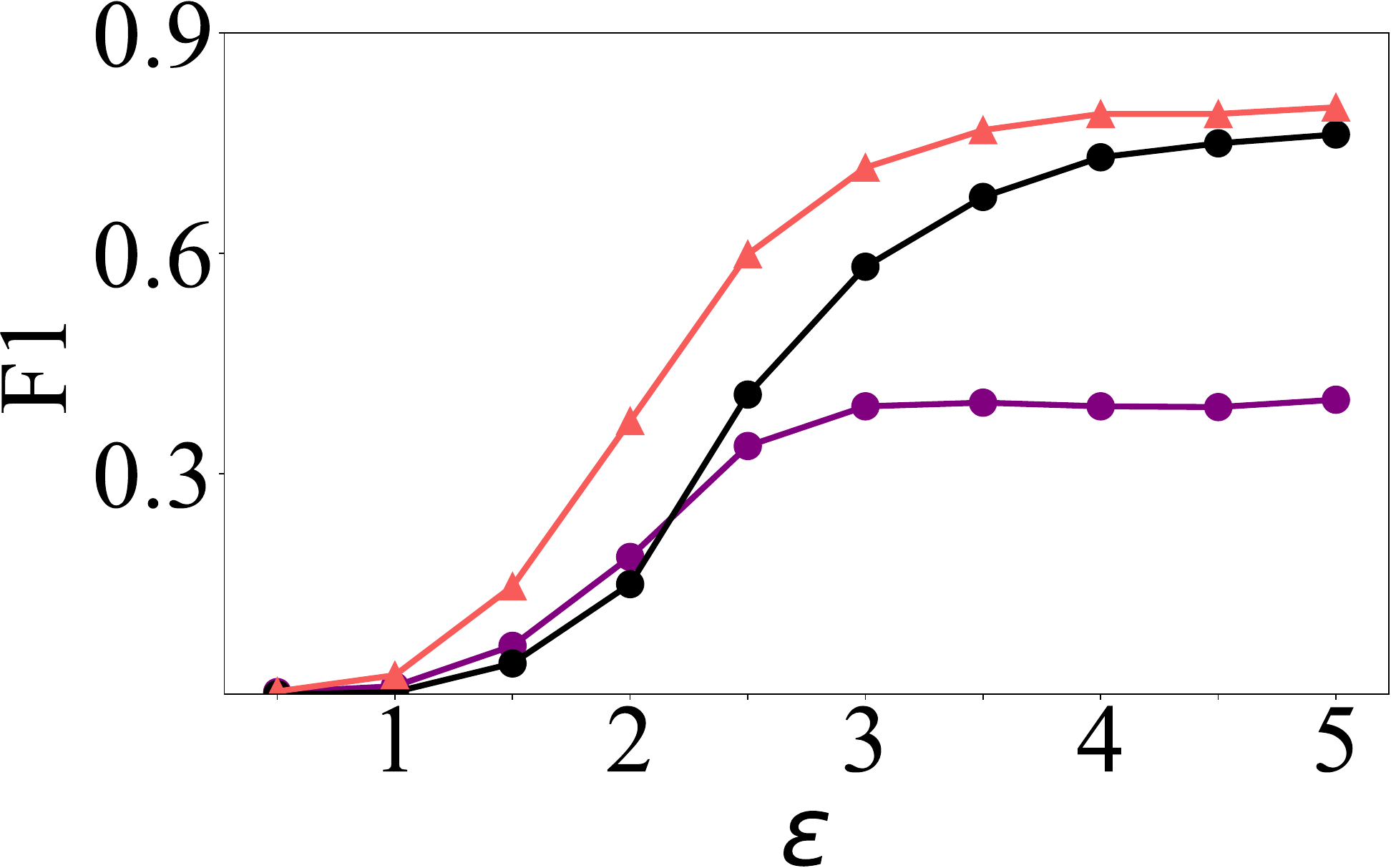}
        }
        \subfloat[TYS, $k$=40]{
    		\includegraphics[width=0.32\columnwidth]{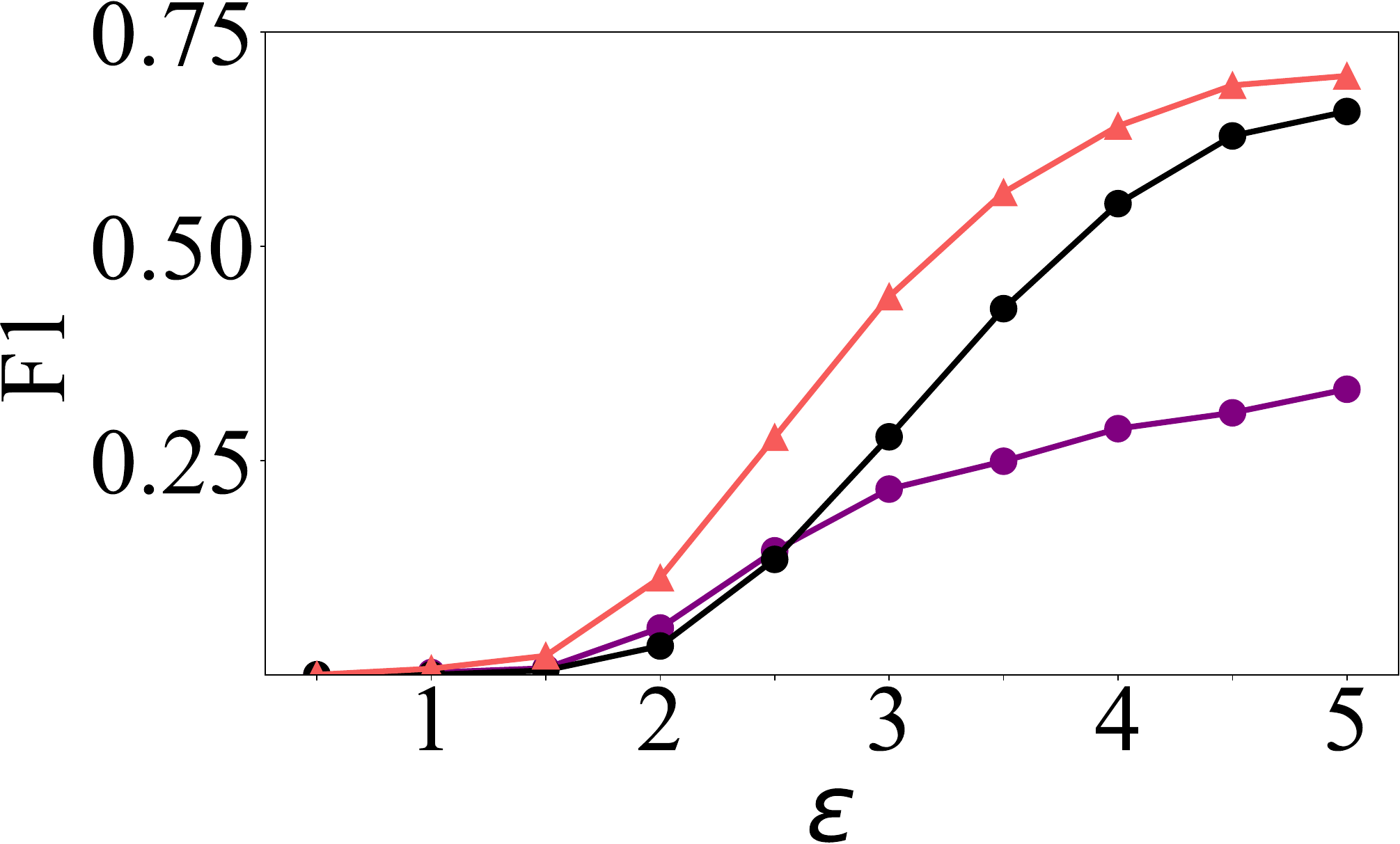}
    	}\hfill
     \vspace{-0.1in}
        \subfloat[UBA, $k$=10]{
        	\includegraphics[width=0.32\columnwidth]{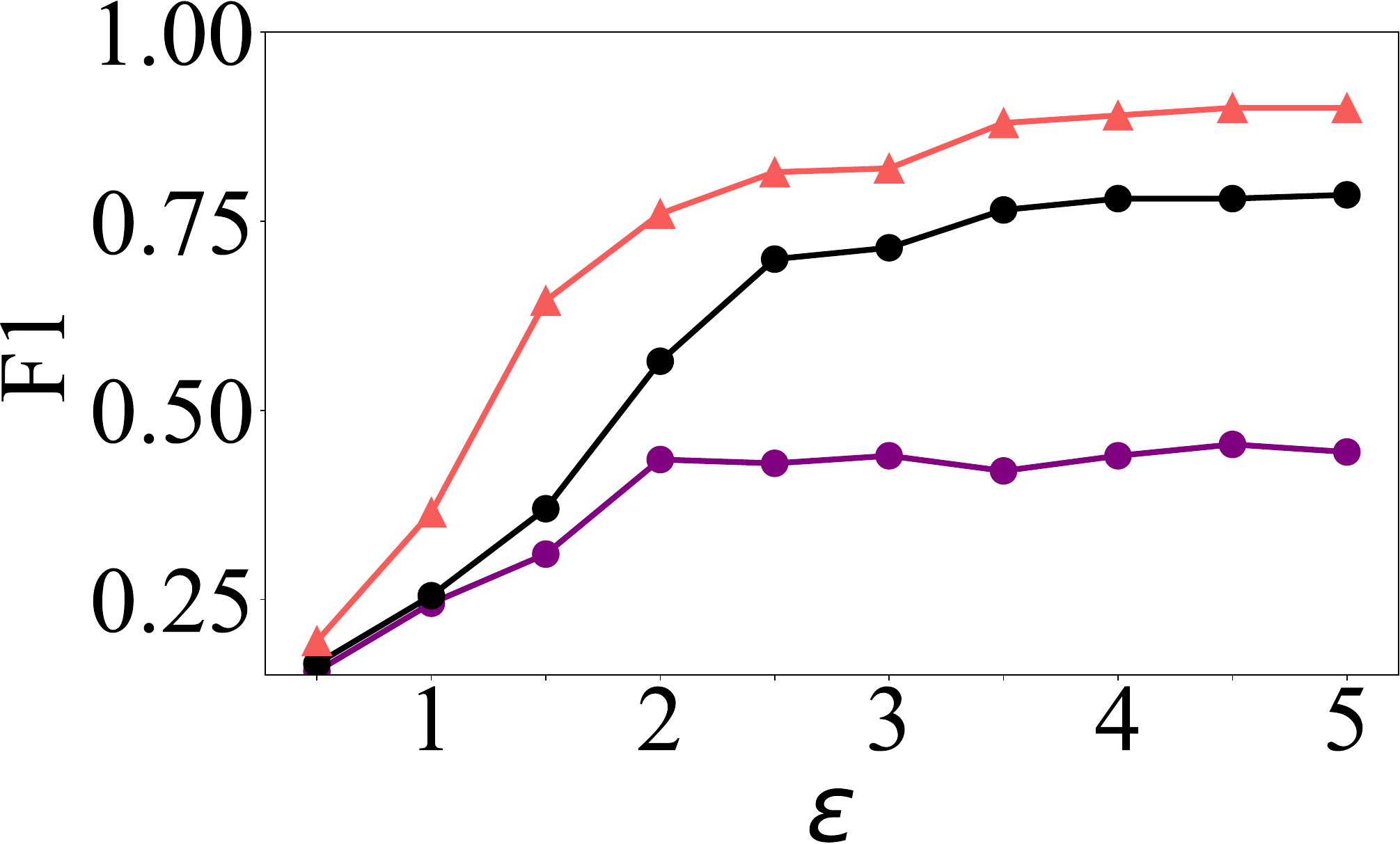}
        }
    	\subfloat[UBA, $k$=20]{
    		\includegraphics[width=0.31\columnwidth]{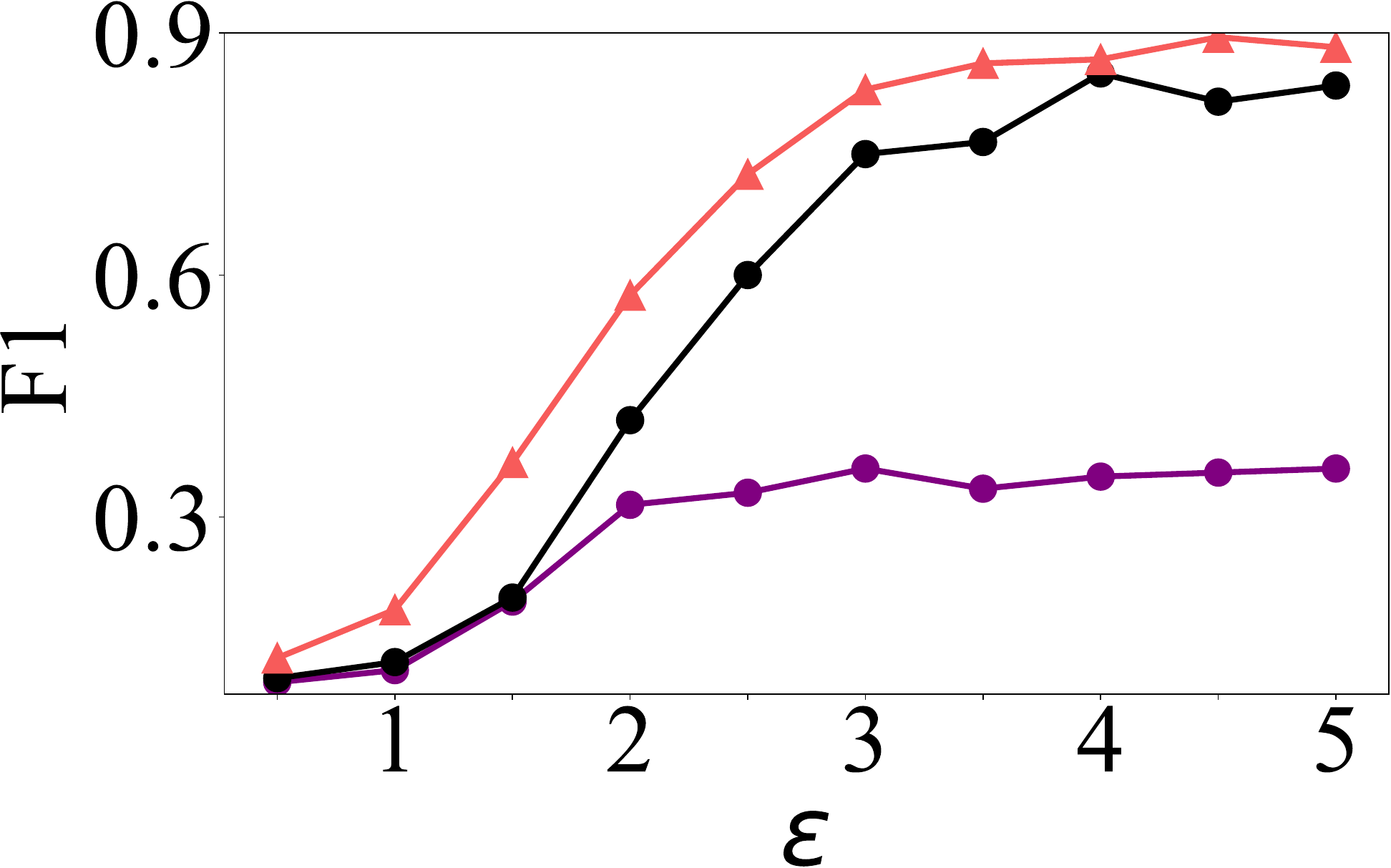}
    	}
    	\subfloat[UBA, $k$=40]{
    		\includegraphics[width=0.31\columnwidth]{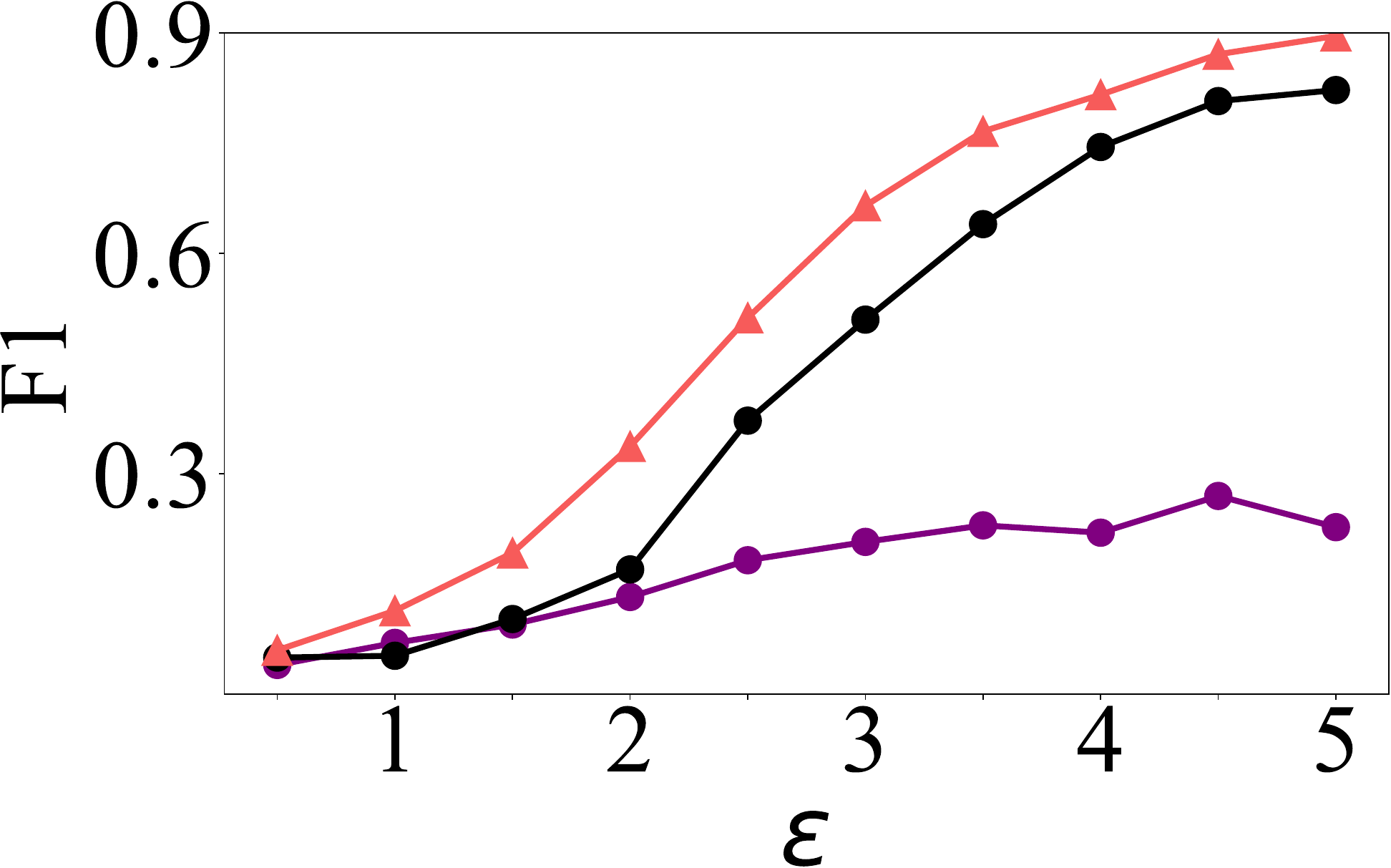}
    	}\hfill
	    \vspace{-0.1in}
        \subfloat[SYN, $k$=10]{
        	\includegraphics[width=0.31\columnwidth]{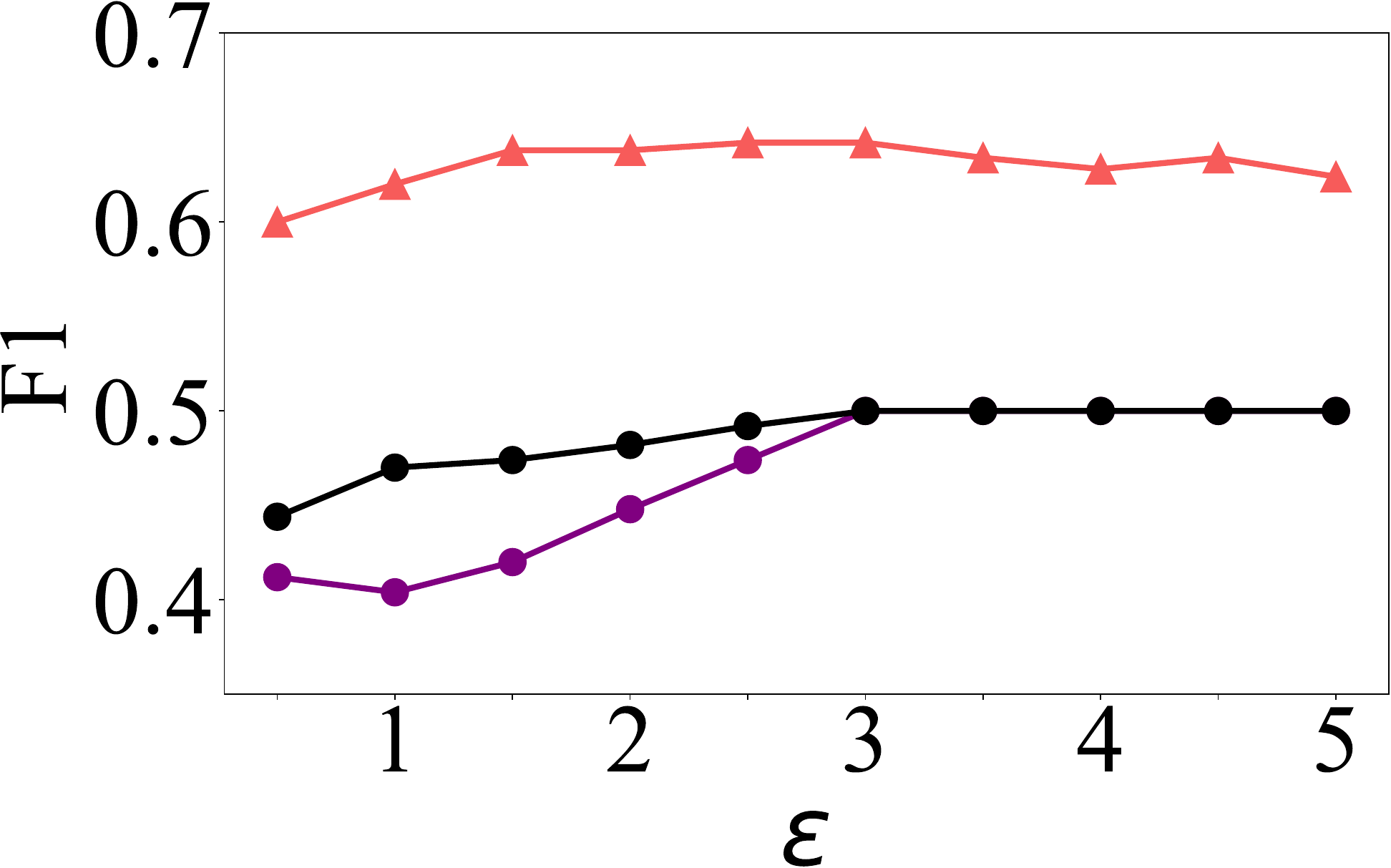}
        }
            \subfloat[SYN, $k$=20]{
    		\includegraphics[width=0.31\columnwidth]{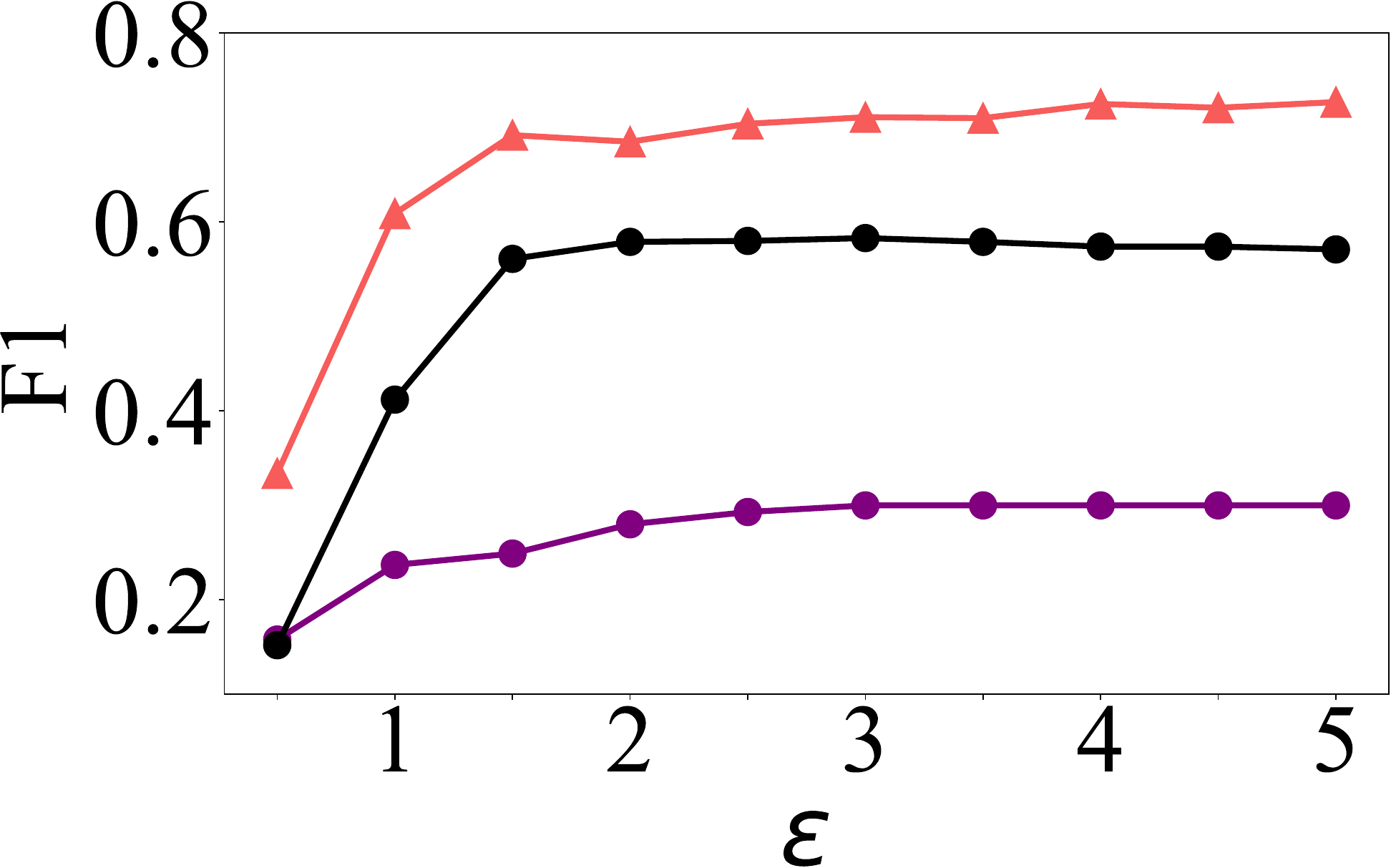}
    	}
    	\subfloat[SYN, $k$=40]{
    		\includegraphics[width=0.31\columnwidth]{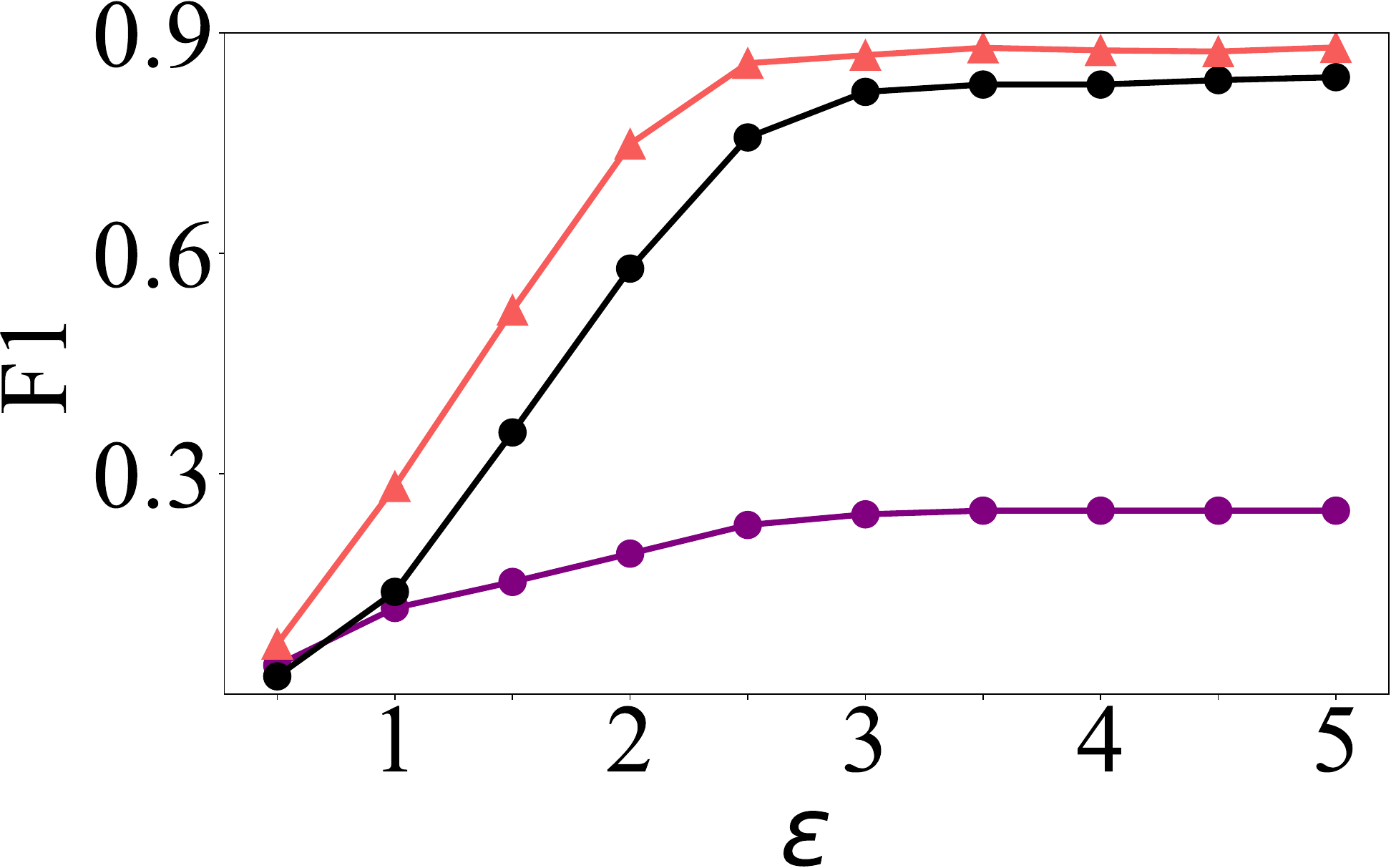}
    	}\hfill
     \vspace{-0.1in}
    \caption{F1 scores vs. privacy budget $\epsilon$ under different $k$.}
    \vspace{-0.15in}
    \label{fig_F1}
\end{figure}

\begin{figure}[t]
    \centering
    \includegraphics[width=0.5\columnwidth]{experimental_figures/fig_legend.pdf} \\ 
    \vspace{-0.15in}
        \subfloat[RDB, $k$=10]{
    	\includegraphics[width=0.31\columnwidth]{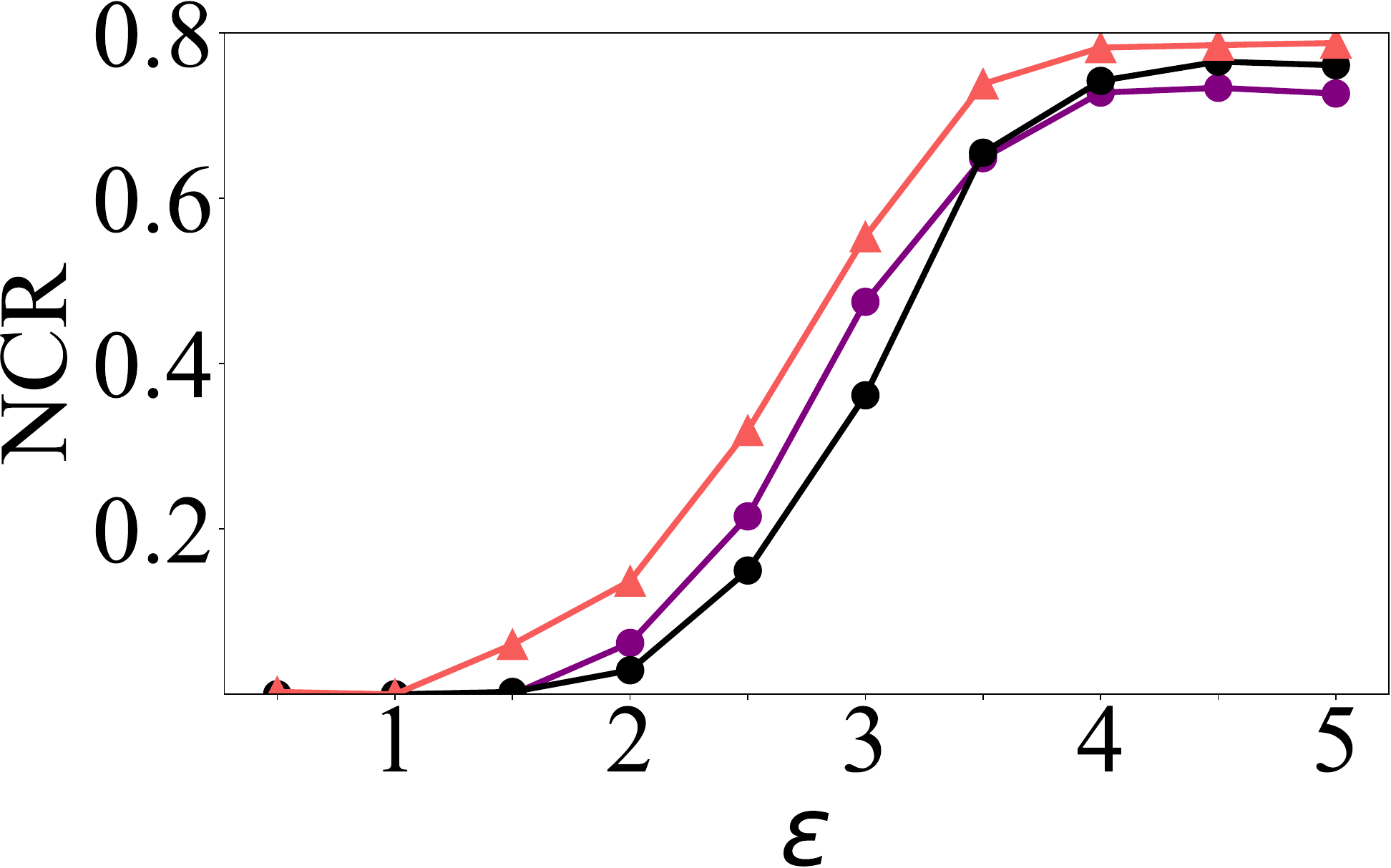}
    }
    \subfloat[RDB, $k$=20]{
    		\includegraphics[width=0.31\columnwidth]{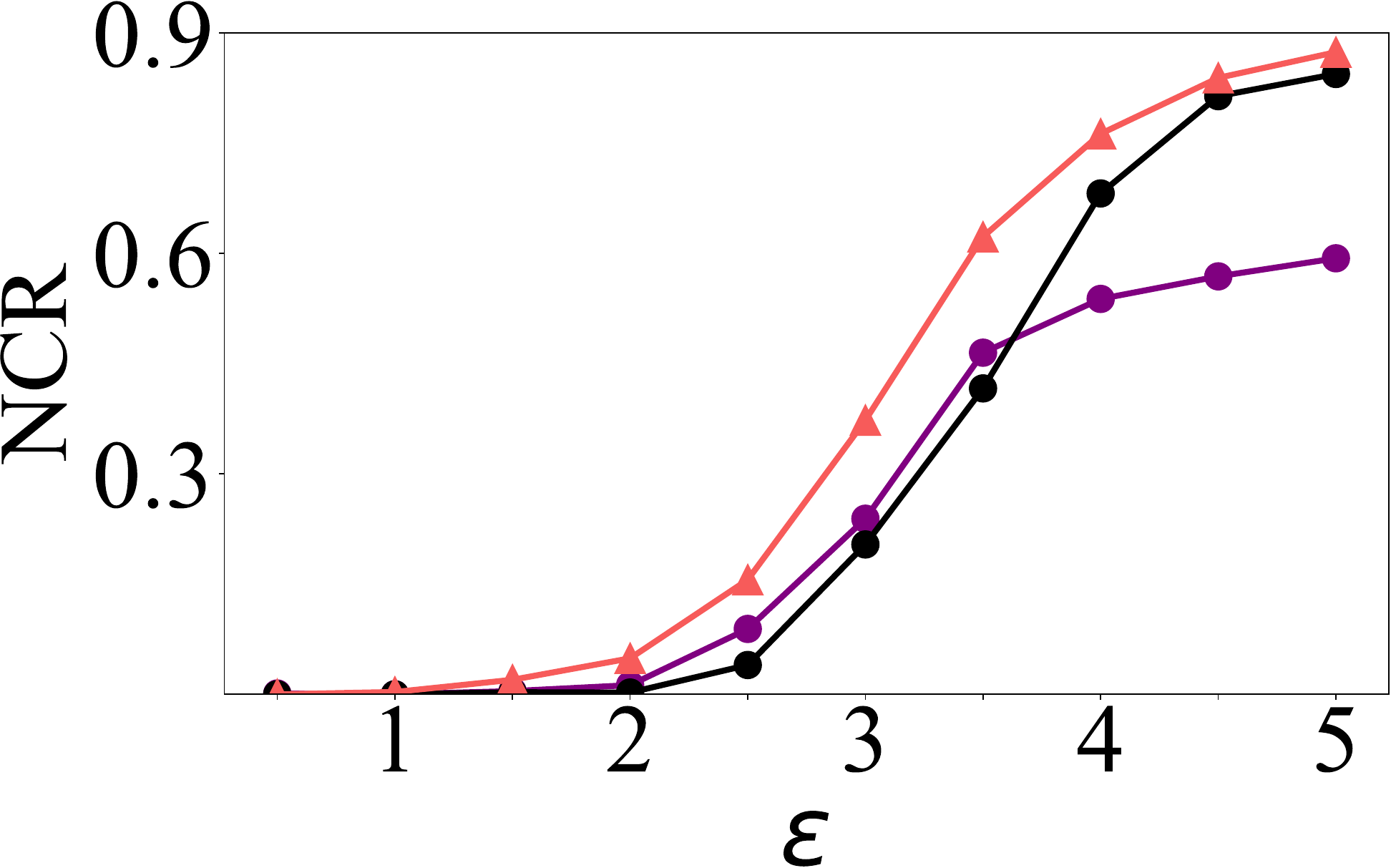}
    	}
     \subfloat[RDB, $k$=40]{
    		\includegraphics[width=0.31\columnwidth]{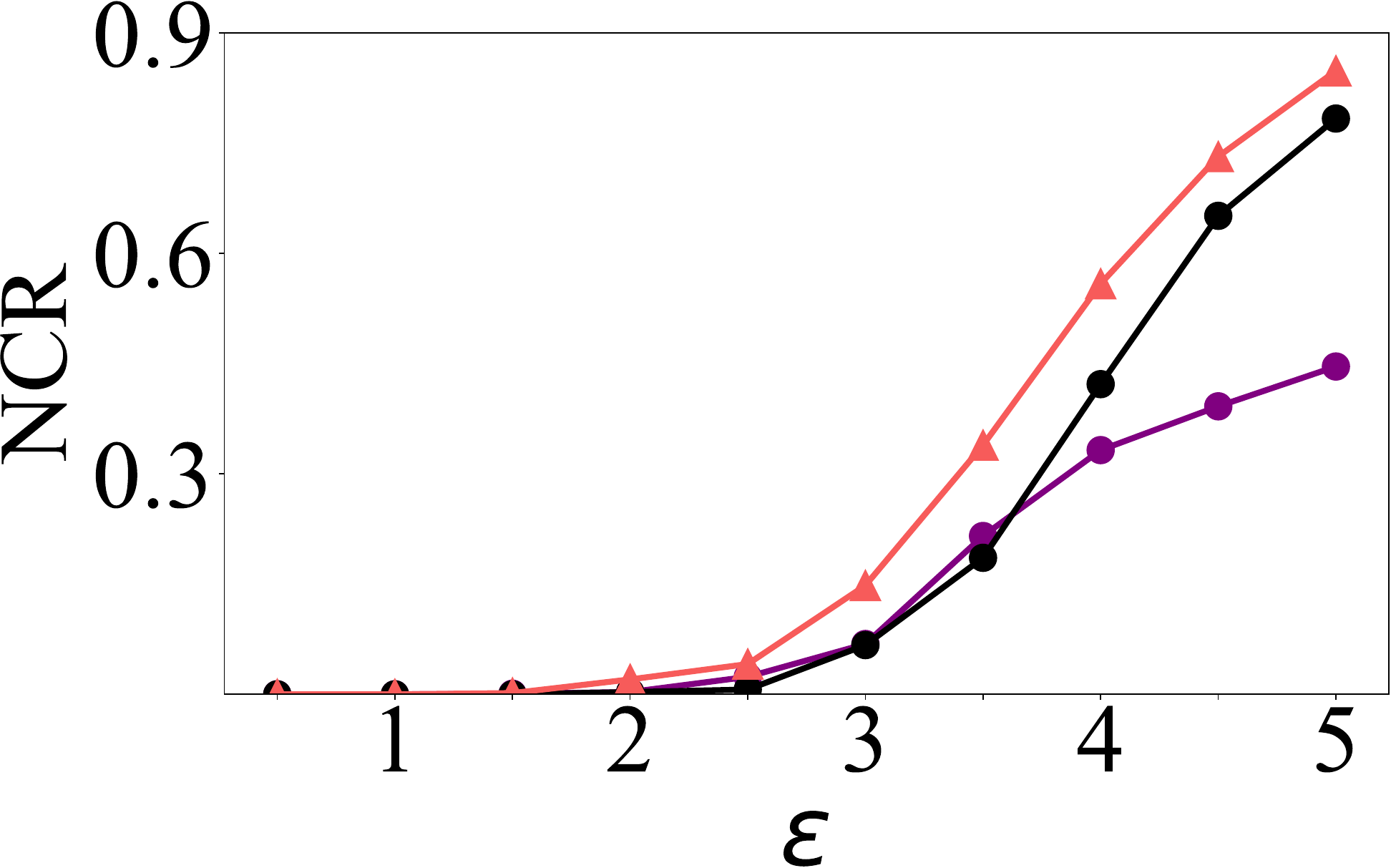}
    	}\hfill
     \vspace{-0.1in}
	\subfloat[YCM, $k$=10]{
		\includegraphics[width=0.31\columnwidth]{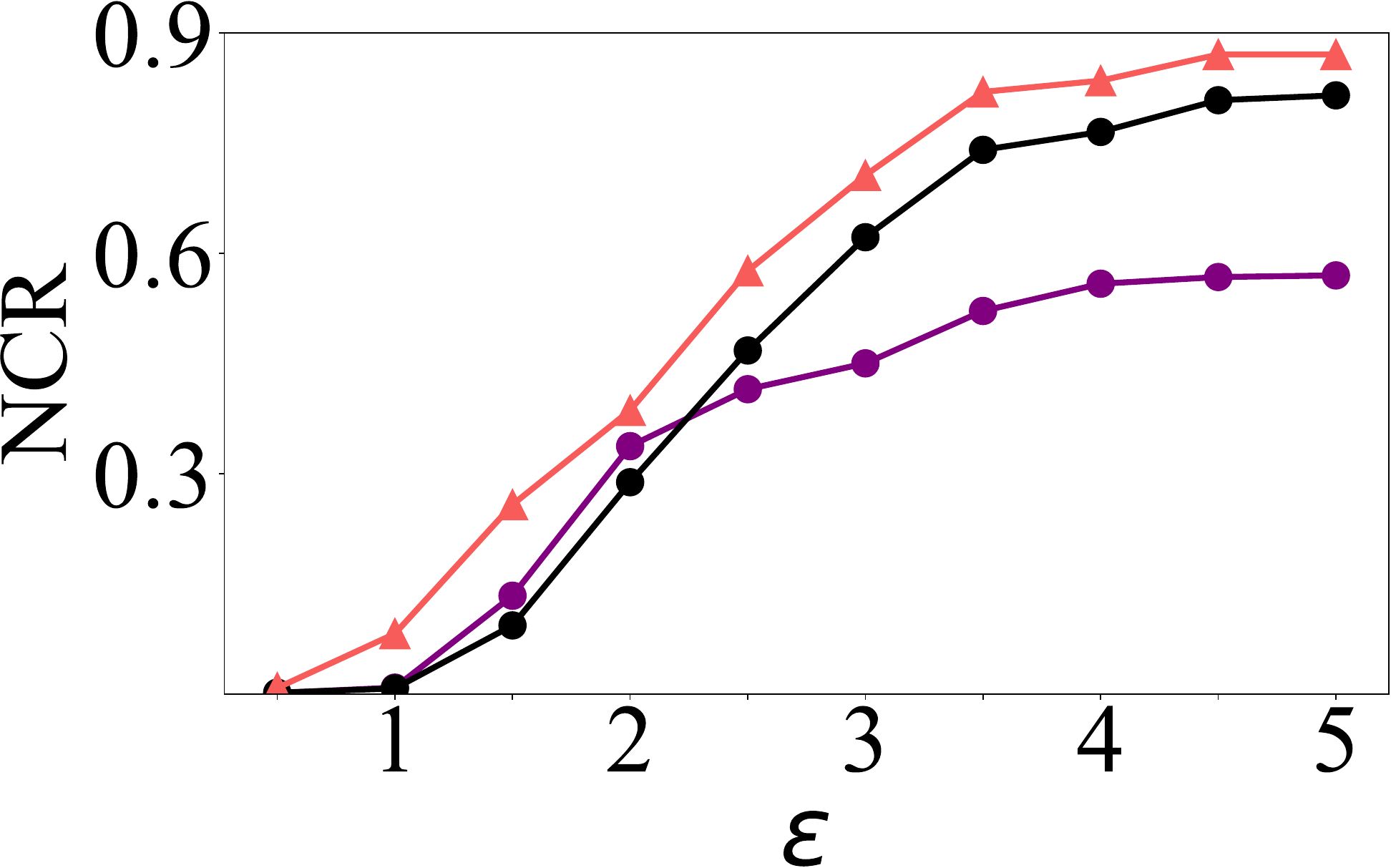}
	}
        \subfloat[YCM, $k$=20]{
    		\includegraphics[width=0.31\columnwidth]{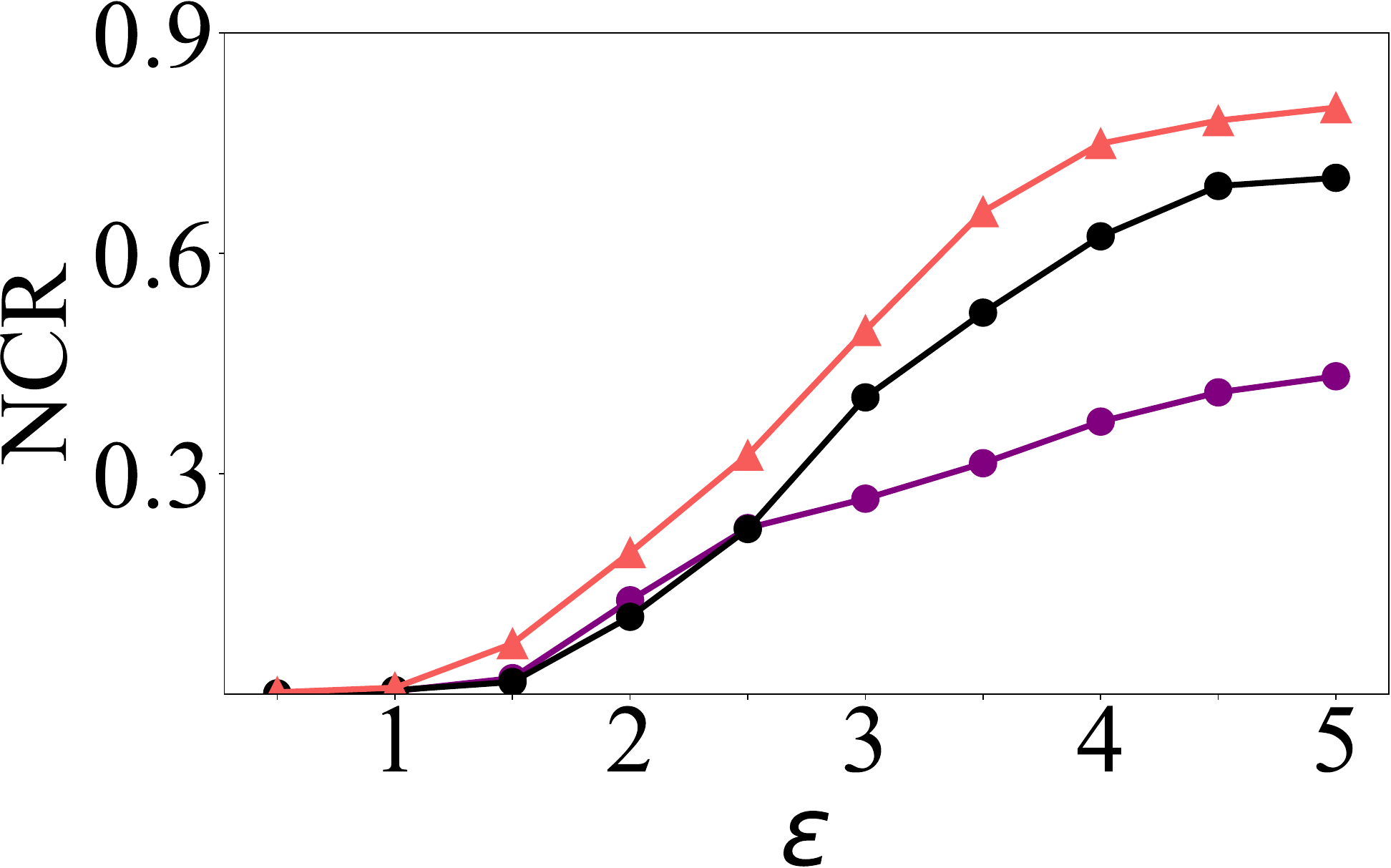}
    	}
     \subfloat[YCM, $k$=40]{
        	\includegraphics[width=0.31\columnwidth]{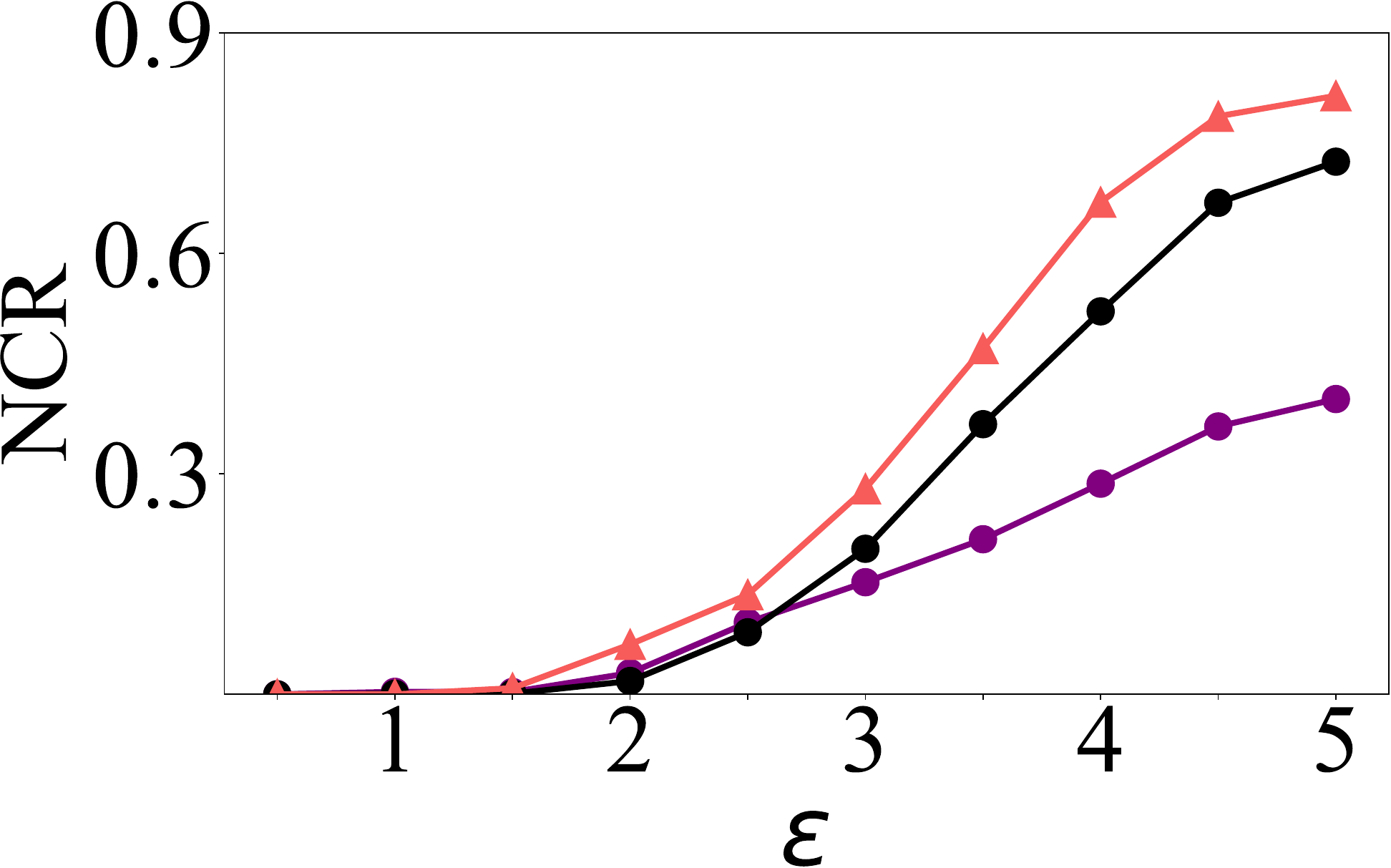}
        }\hfill
     \vspace{-0.1in}
	\subfloat[TYS, $k$=10]{
		\includegraphics[width=0.31\columnwidth]{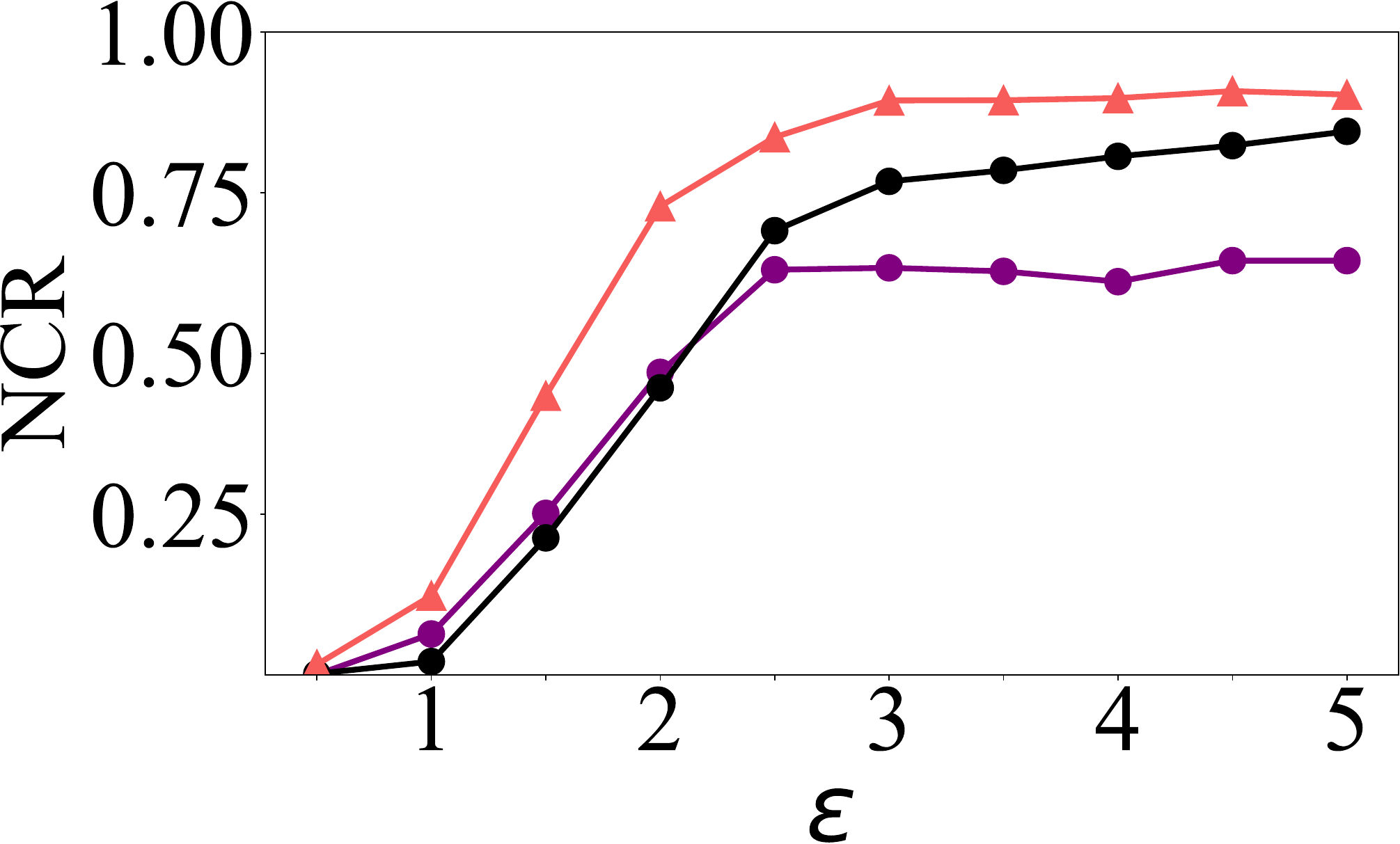}
	}
        \subfloat[TYS, $k$=20]{
        	\includegraphics[width=0.31\columnwidth]{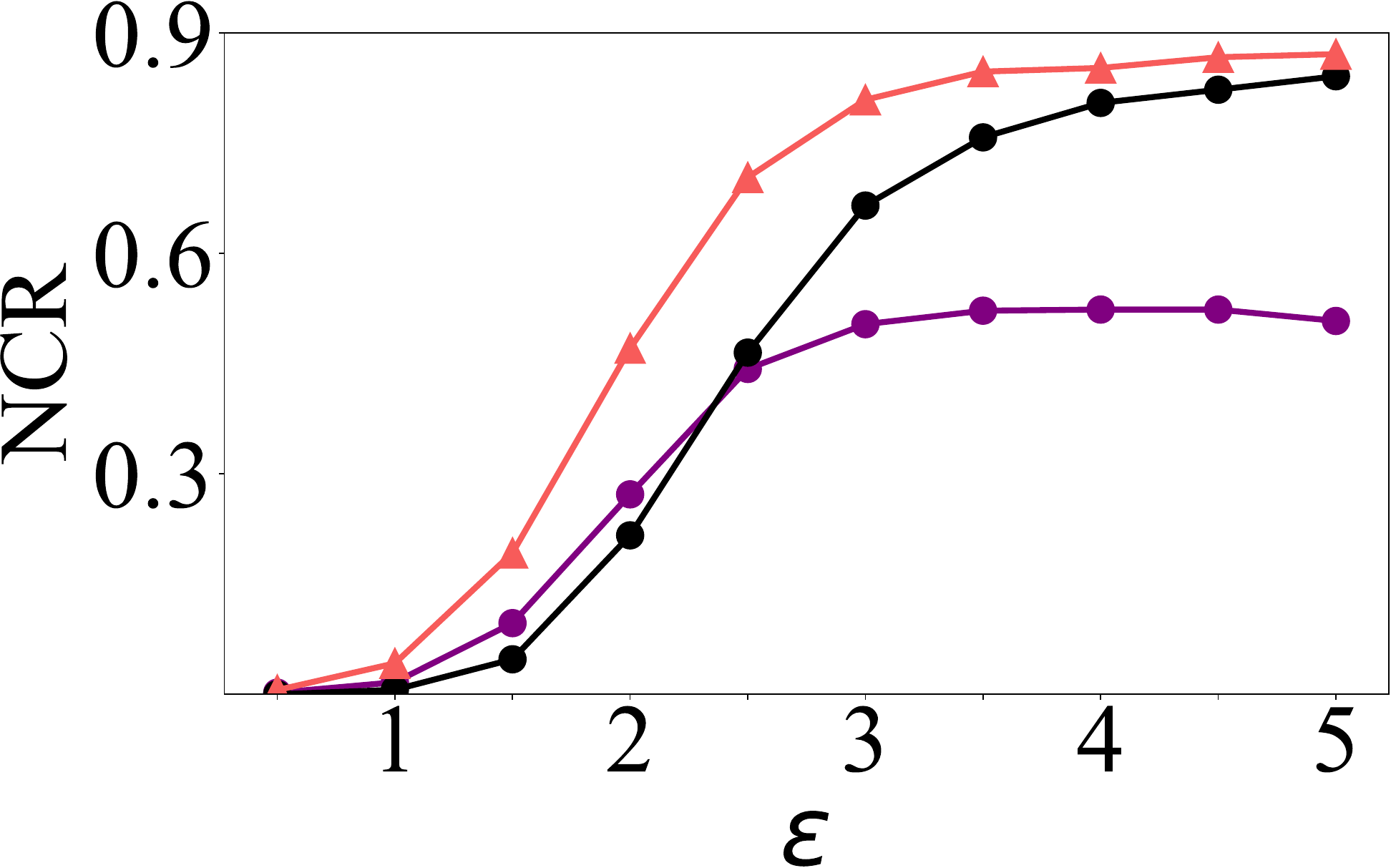}
        }
        \subfloat[TYS, $k$=40]{
    		\includegraphics[width=0.31\columnwidth]{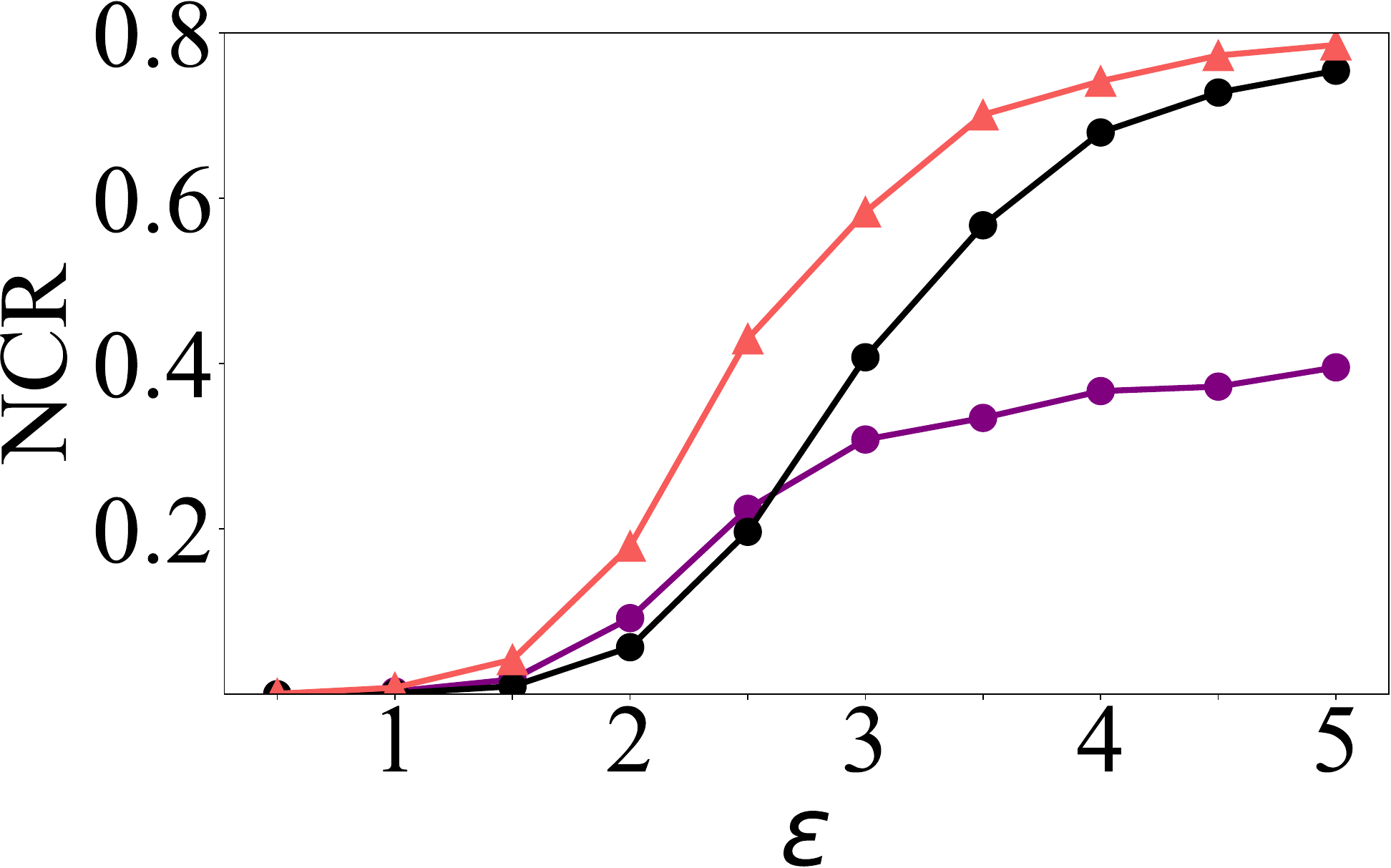}
    	}\hfill
     \vspace{-0.1in}
        \subfloat[UBA, $k$=10]{
            	\includegraphics[width=0.31\columnwidth]{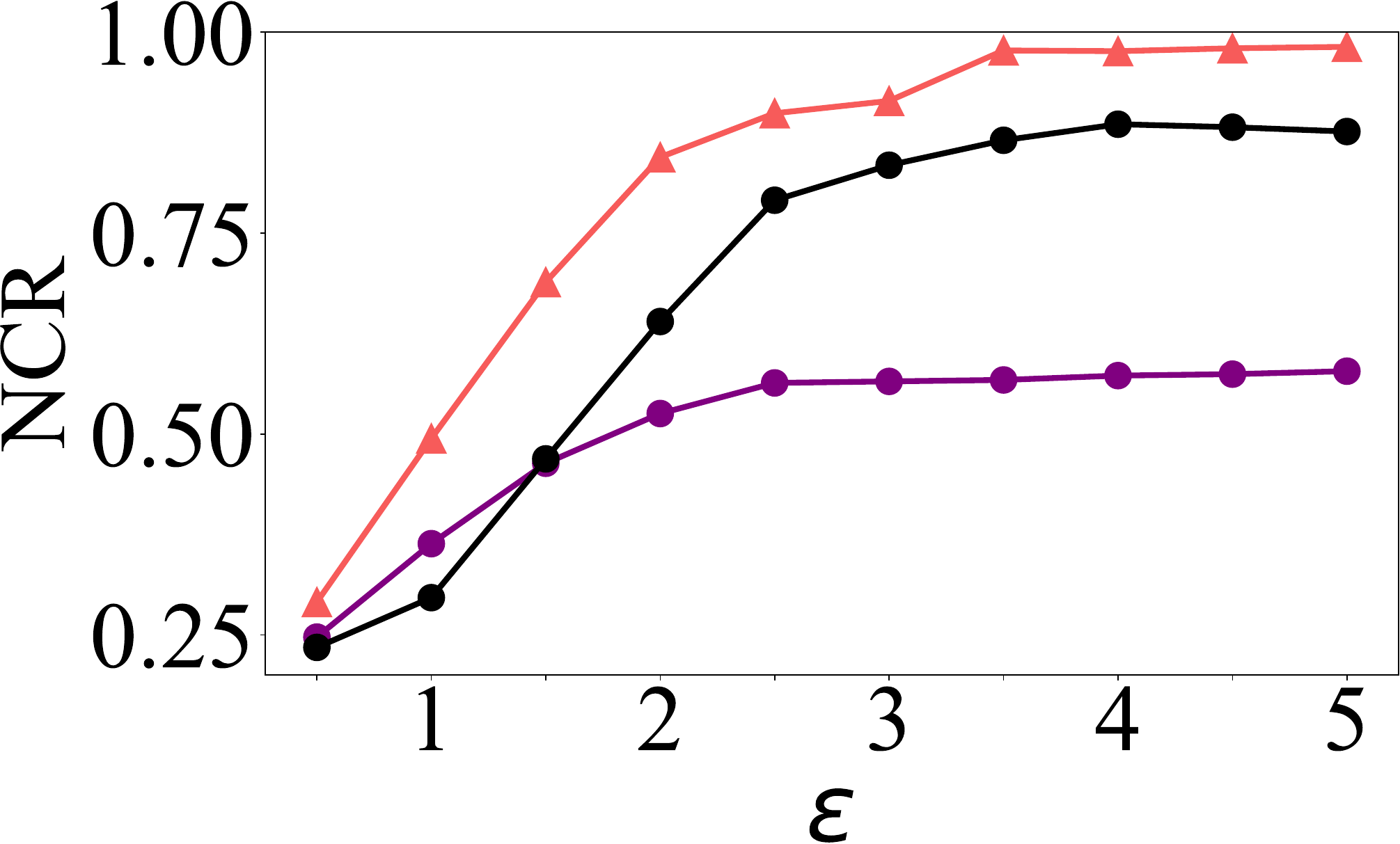}
            }
    	\subfloat[UBA, $k$=20]{
    		\includegraphics[width=0.315\columnwidth]{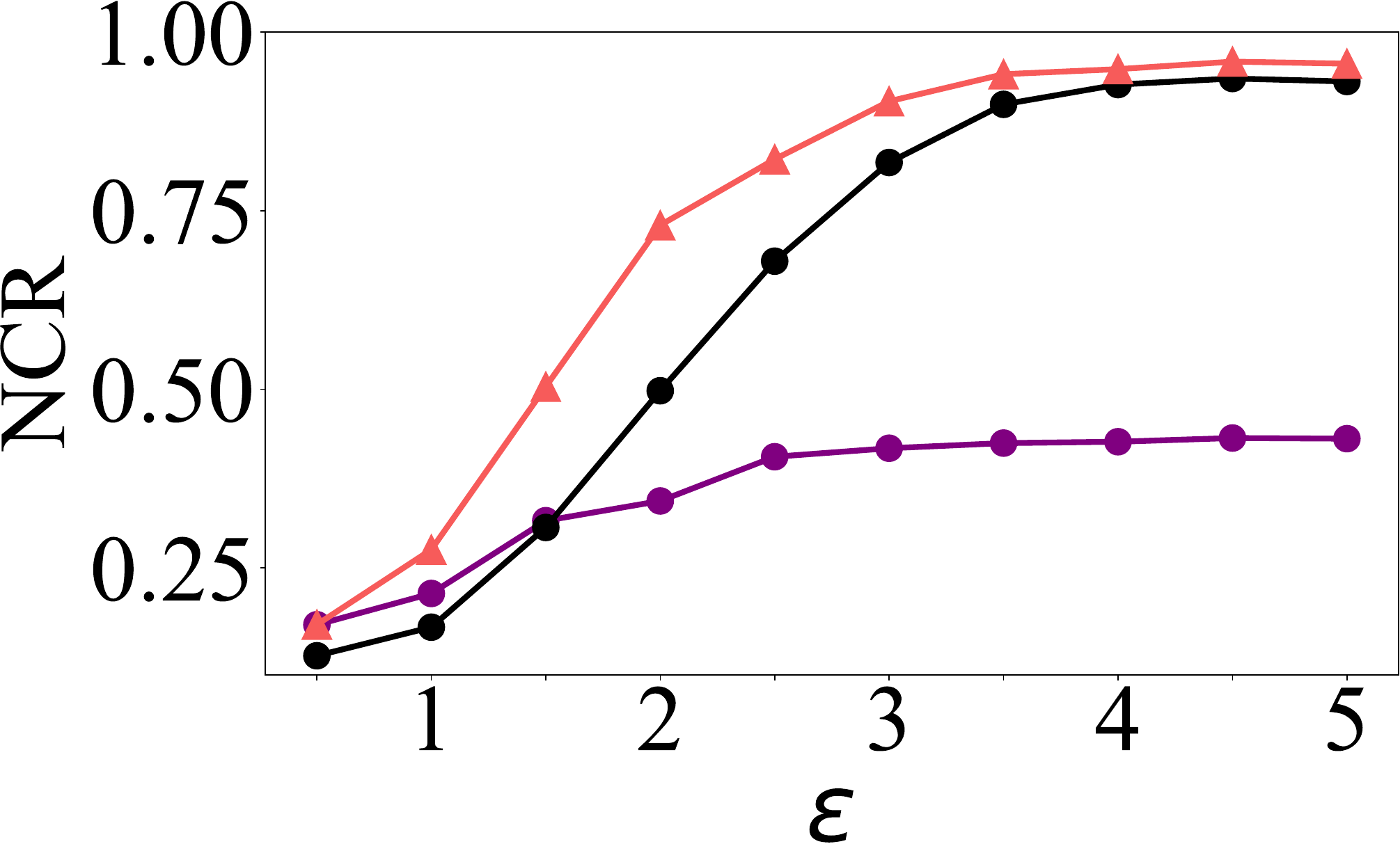}
    	}
    	\subfloat[UBA, $k$=40]{
    		\includegraphics[width=0.32\columnwidth]{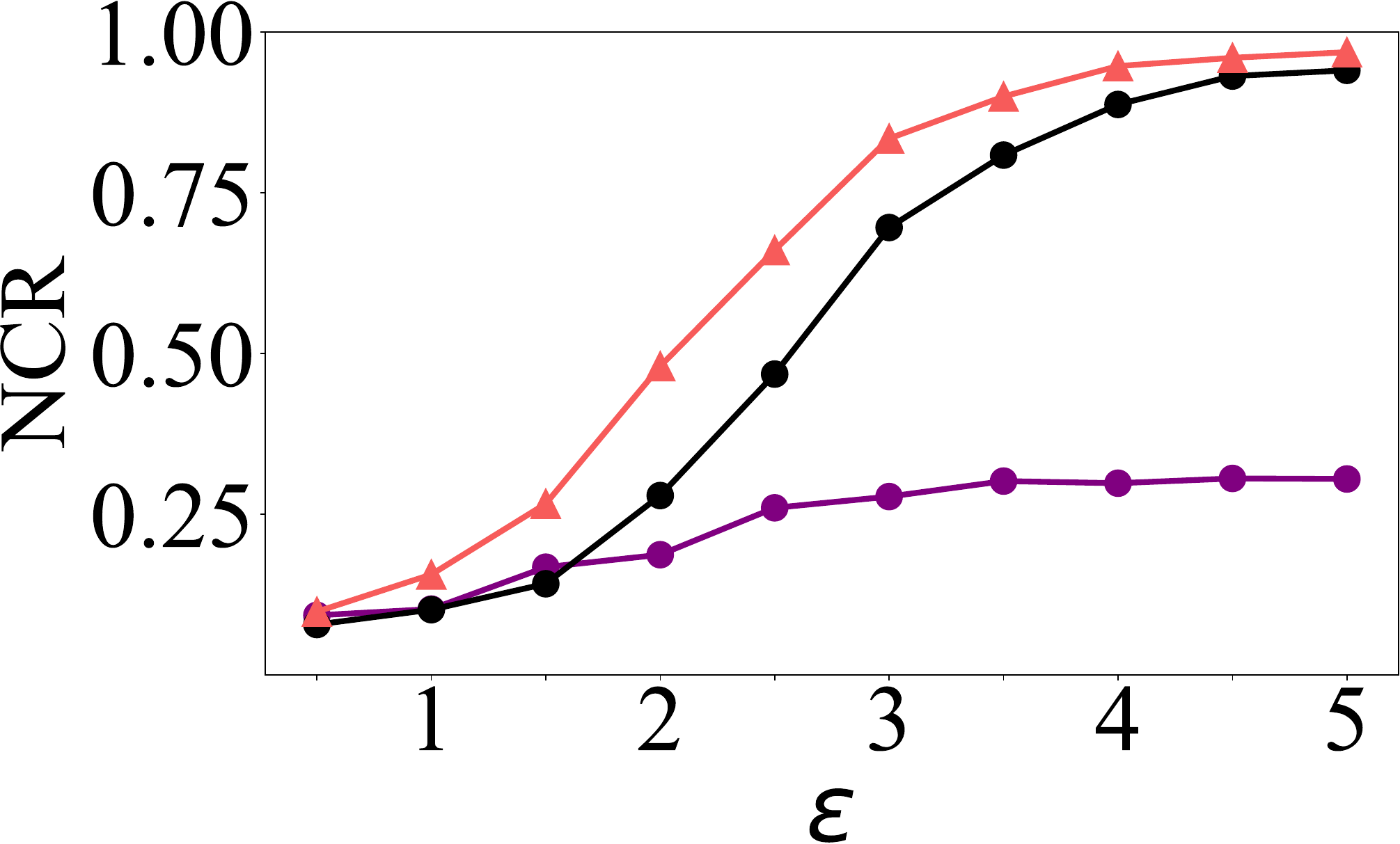}
    	}\hfill
    \vspace{-0.1in}
        \subfloat[SYN, $k$=10]{
        	\includegraphics[width=0.31\columnwidth]{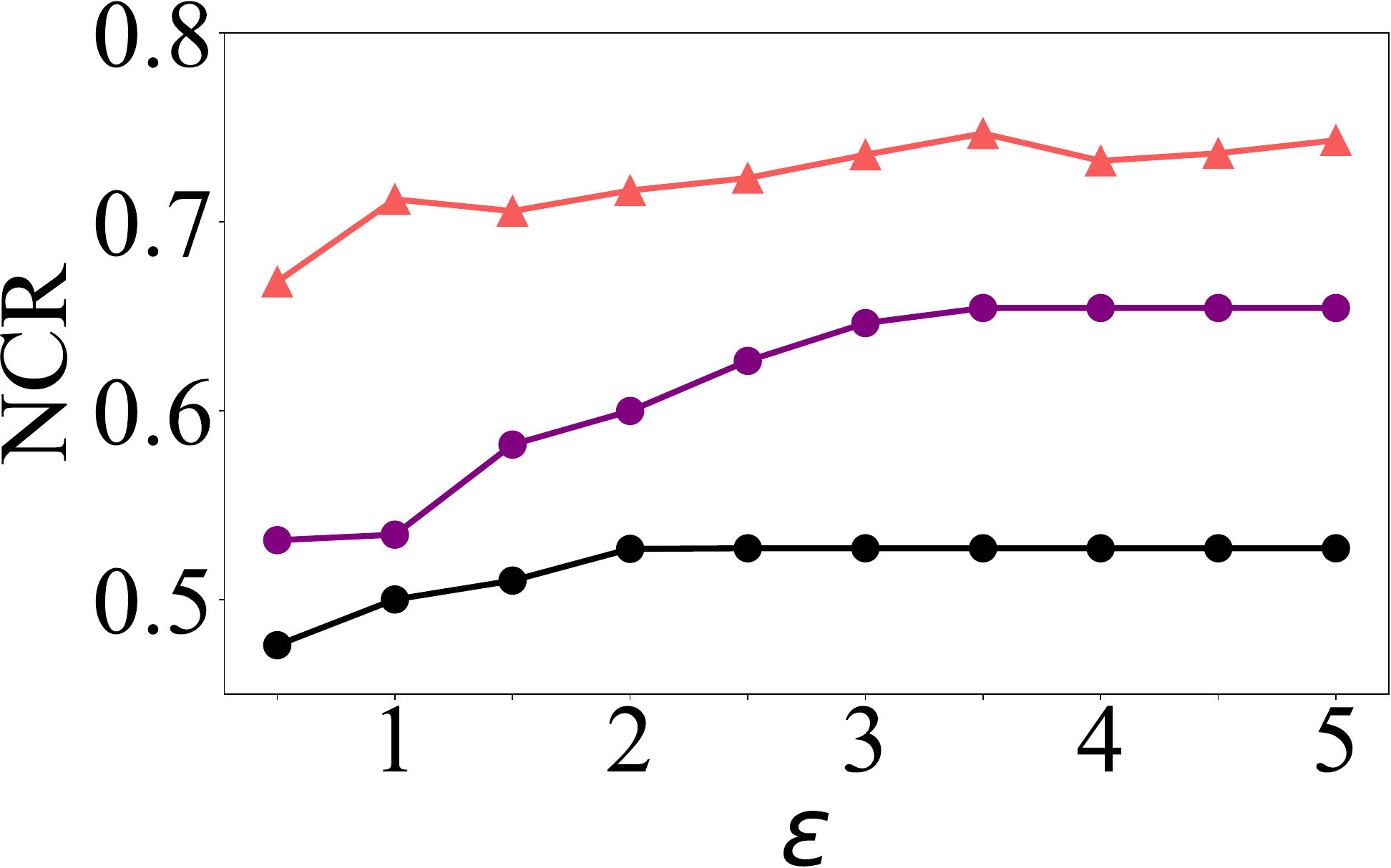}
        }
    	\subfloat[SYN, $k$=20]{
    		\includegraphics[width=0.31\columnwidth]{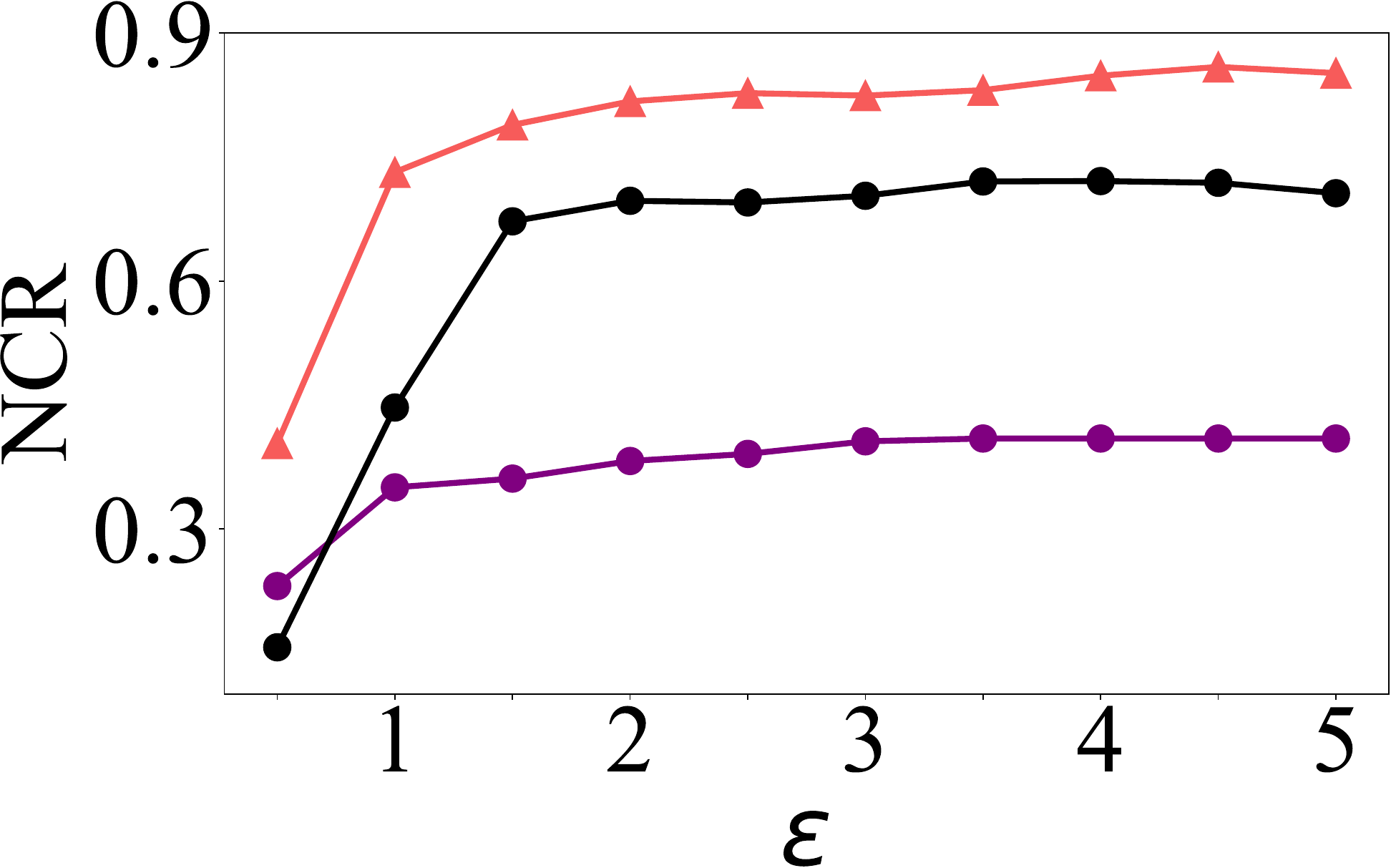}
    	}
    	\subfloat[SYN, $k$=40]{
    		\includegraphics[width=0.32\columnwidth]{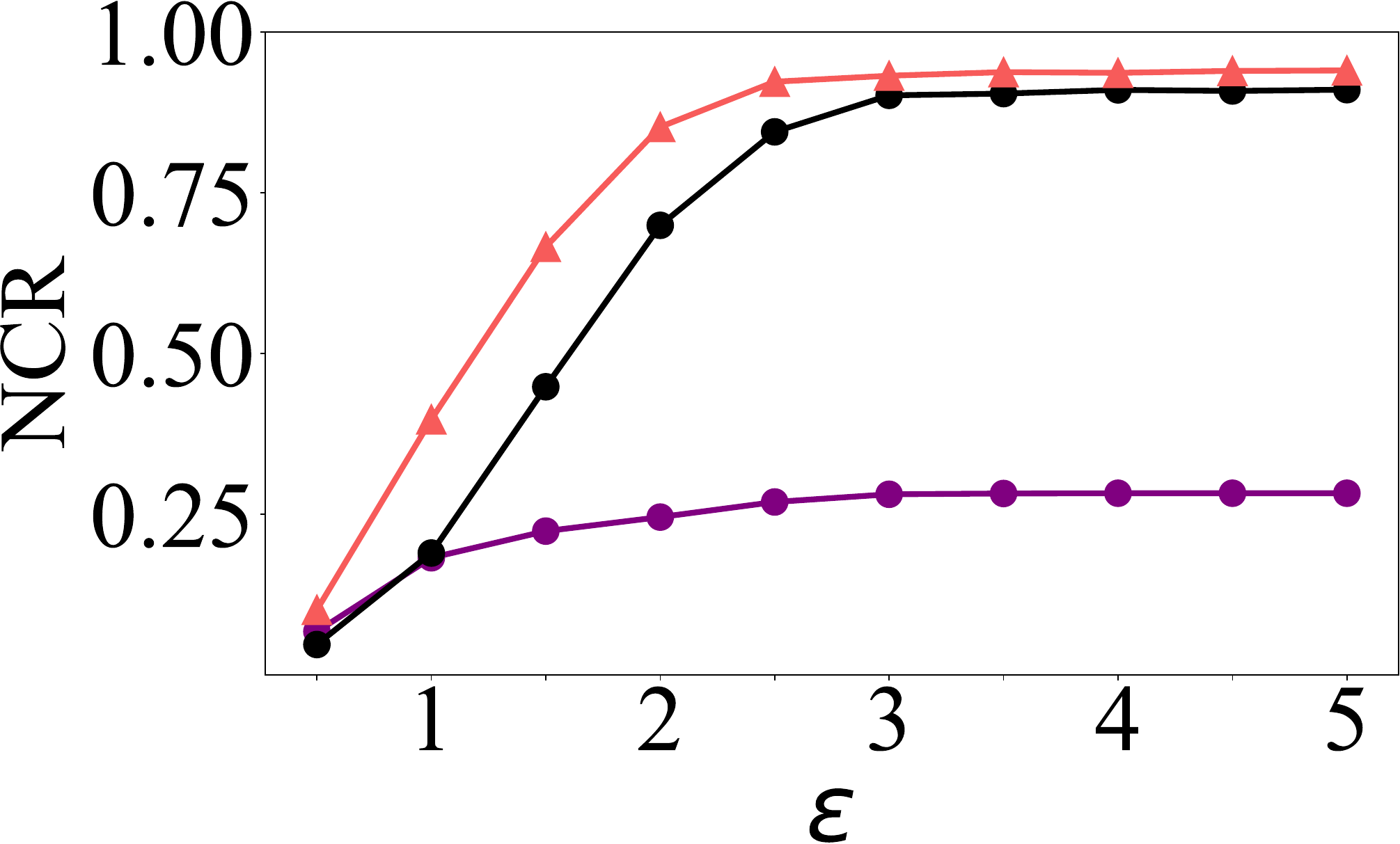}
    	}\hfill
    \vspace{-0.1in}
    \caption{NCR scores vs. privacy budget $\epsilon$ under different $k$.}
    \vspace{-0.15in}
    \label{fig_NCR}
\end{figure}

{\bf Baselines and Parameter Setting}. 
The most related study focuses on the cross-party setting is \cite{SHH2023}. Since it does not fully satisfy LDP, we substitute the GRRX mechanism proposed by the authors with the $k$-RR mechanism to ensure the LDP guarantee, denoted as GTF. Another baseline, FedPEM, is described in Algorithm~\ref{FAPEM}. It directly employs PEM, the state-of-the-art mechanism in single party settings, in each party, and then uploads all the identified local heavy hitters and their counts to the server, which then aggregates and reports the overall results. 
\textcolor{black}{For a fair comparison, we use $k$-RR as the FO (we also investigate the impacts of different FOs, see Section~\ref{subsec:different FOs}) and set maximum length $m=48$. For $k$-RR, we assign a dummy item to out-of-domain items. We set the granularity $g=24$ (i.e., step size $m/g=2$) since $g$ should not be too large (see Section~\ref{subsec:ablation study}). We recommend $g=24$ or $12$ depending on the population size. A large population can support a fine granularity while reducing the accumulated sampling errors. $\beta$ determines the proportion of users for pruning candidates at each level, functioning as a test set. Therefore, we set $\beta=0.1$, following the common practice in machine learning. We heuristically set the level of the shared shallow trie $g=\lfloor0.25g\rfloor$ and assign 10\% users for the estimations in this phase as a warm start to align all parties.} All the mechanisms are executed 50 times and the average is plotted. 

{\bf Metrics}. 
We evaluate the utility of the mechanisms using two metrics following the previous study~\cite{WLJ2019}. The first one is the F1 score~\cite{MRS2008}, which is the harmonic mean of precision and recall: 
$F1 = 2pr / (p + r), $
where $p = |R^k_T \cap R^k| / |R^k|$, $ \quad r = |R^k_T \cap R^k| / |R^k_T|$, $R^k_T$ is the top $k$ ground truths, and $R^k$ is the estimated heavy hitters. 

The other metric is the Normalized Cumulative Rank (NCR)~\cite{WLJ2019}, 
which assigns a higher penalty to missing the most frequent value than to others, by employing a quality function: 
$NCR = \sum_{v \in R^k} q(v)$ $/$$\sum_{v \in R^k_T} q(v), $ 
where $q(v)=k-rank(v)$ if $v$ is ranked within top $k$, otherwise $q(v)=0$. Both F1 and NCR scores fall within the range of $[0,1]$, where larger scores indicate better performance.

\begin{figure*}[h]
    \centering
    \includegraphics[width=0.25\textwidth]{experimental_figures/fig_legend.pdf} \\
    \vspace{-0.15in}
    \begin{minipage}{0.9\textwidth}
        \subfloat[RDB, OUE, $k$=10]{
        	\includegraphics[width=0.188\textwidth]{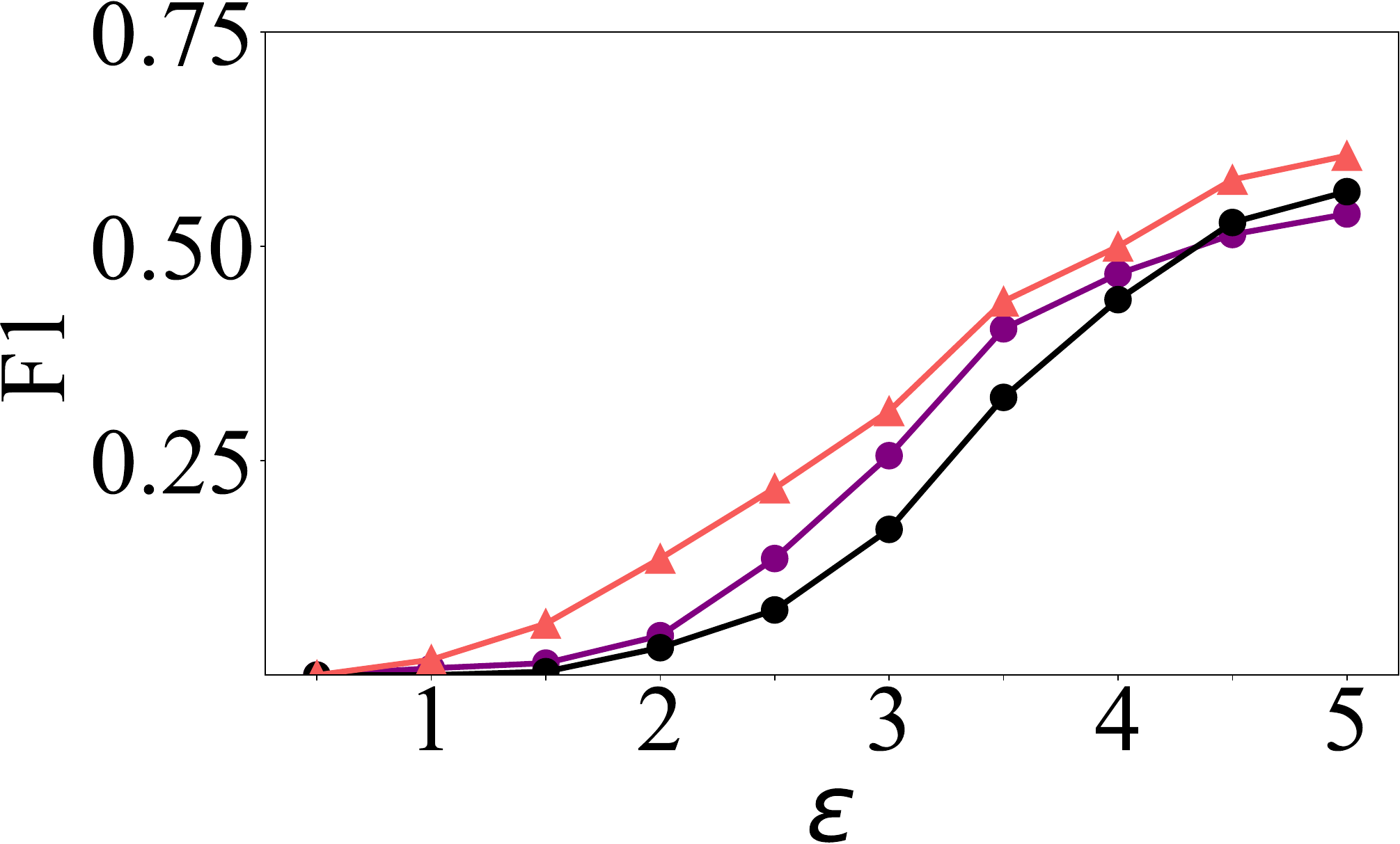}
        }\hfill
        \subfloat[YCM, OUE, $k$=10]{
    		\includegraphics[width=0.185\textwidth]{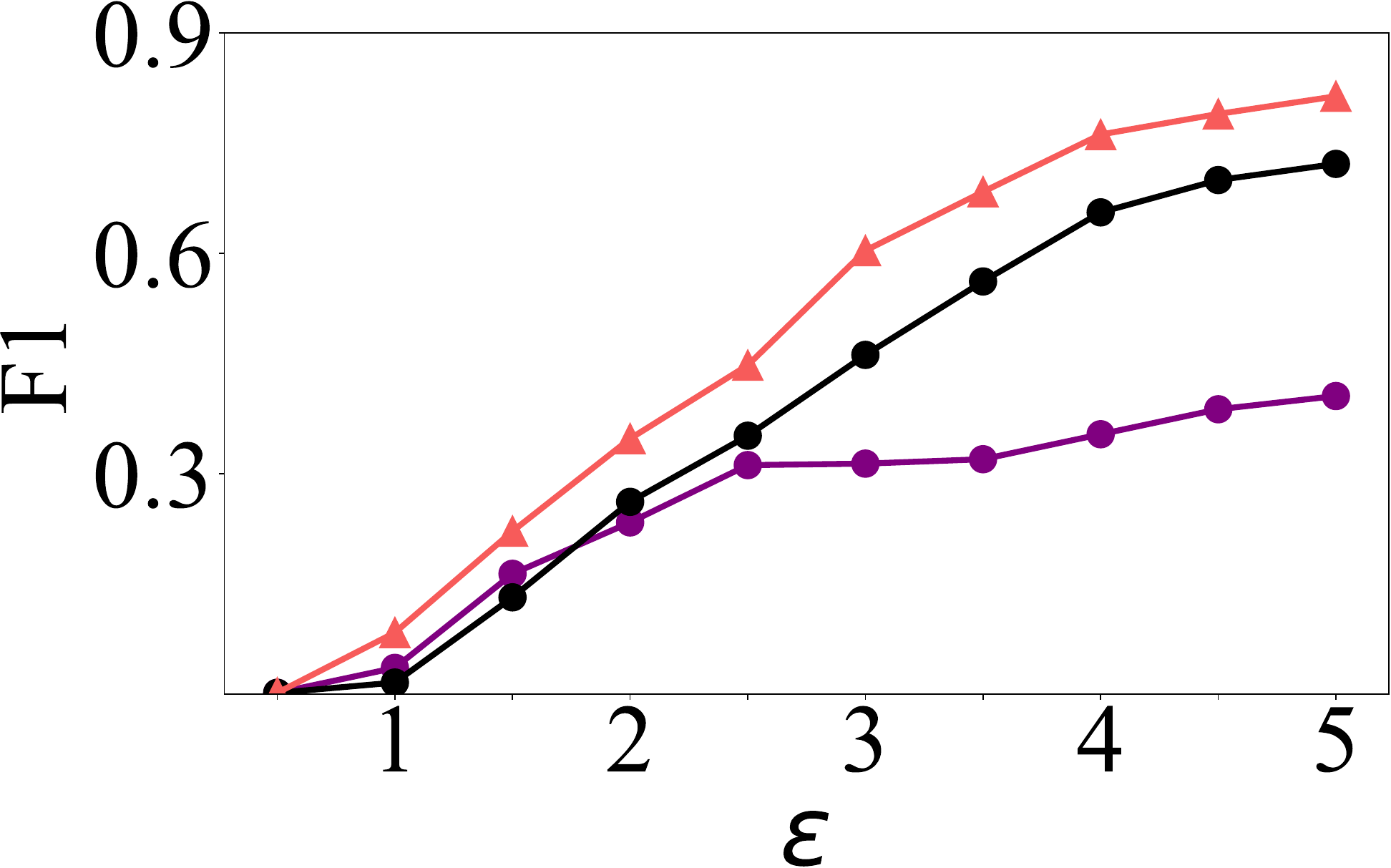}
    	}\hfill
        \subfloat[TYS, OUE, $k$=10]{
    		\includegraphics[width=0.185\textwidth]{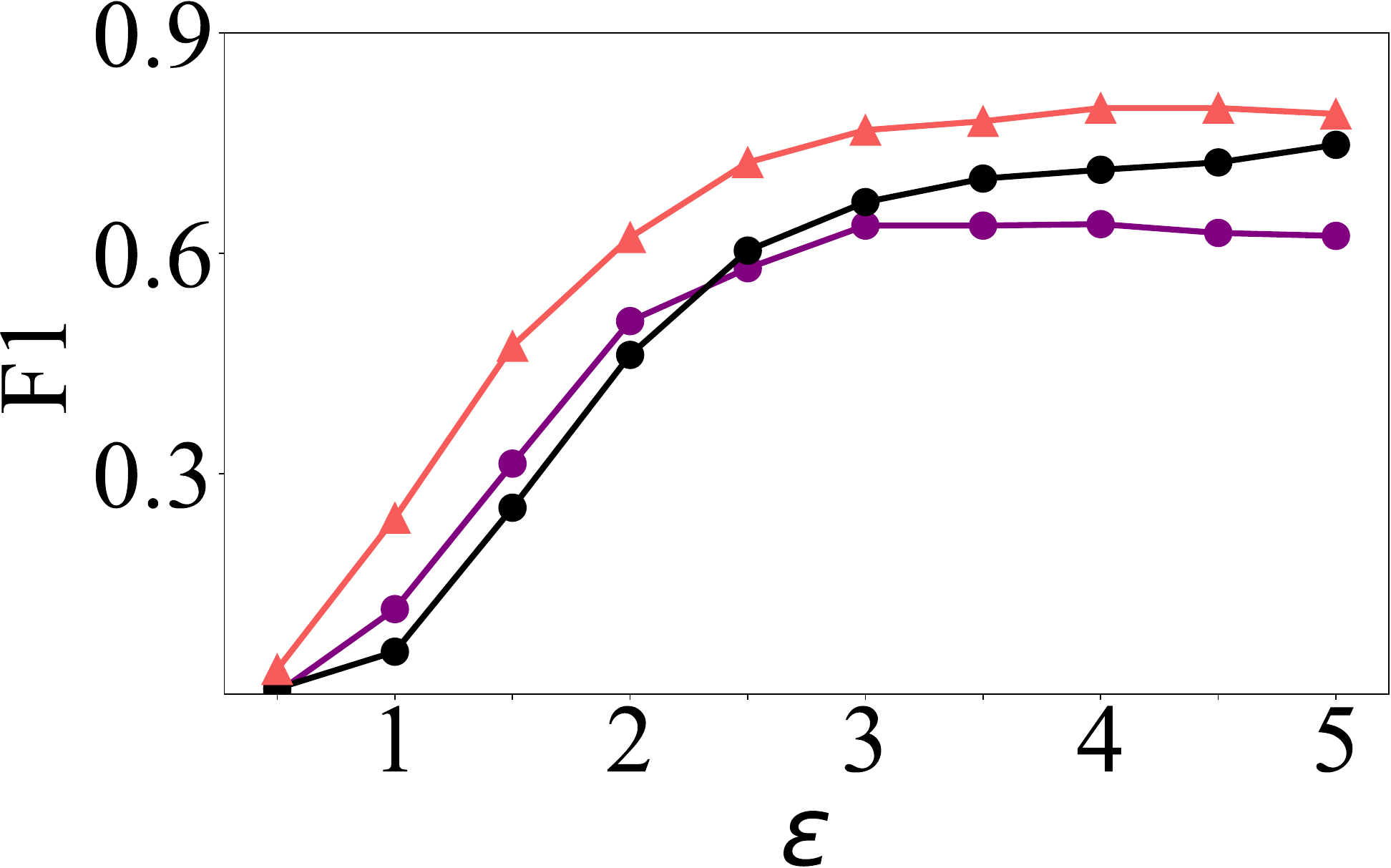}
    	}\hfill
        \subfloat[UBA, OUE, $k$=10]{
    		\includegraphics[width=0.185\textwidth]{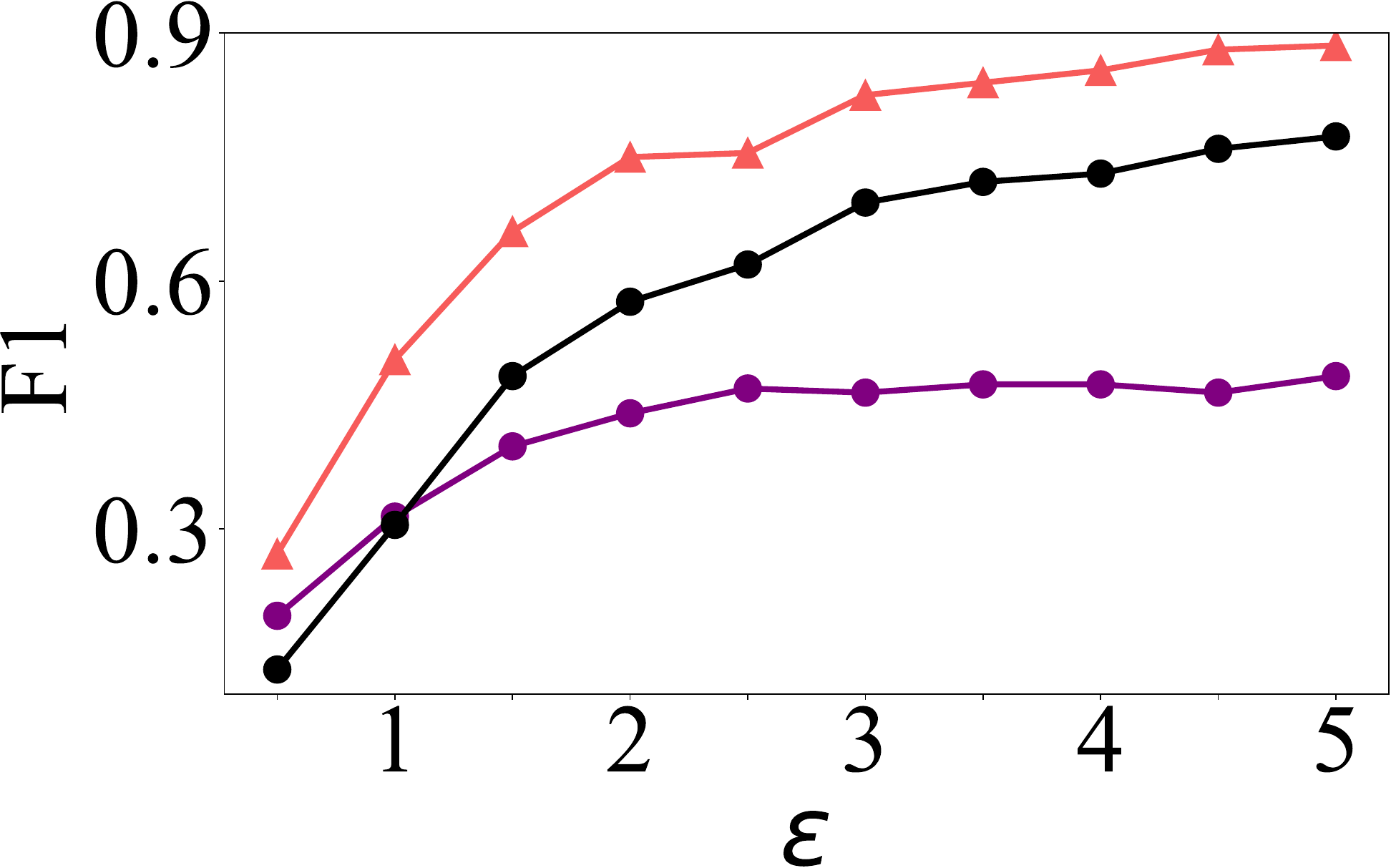}
    	}\hfill
        \subfloat[SYN, OUE, $k$=10]{
    		\includegraphics[width=0.185\textwidth]{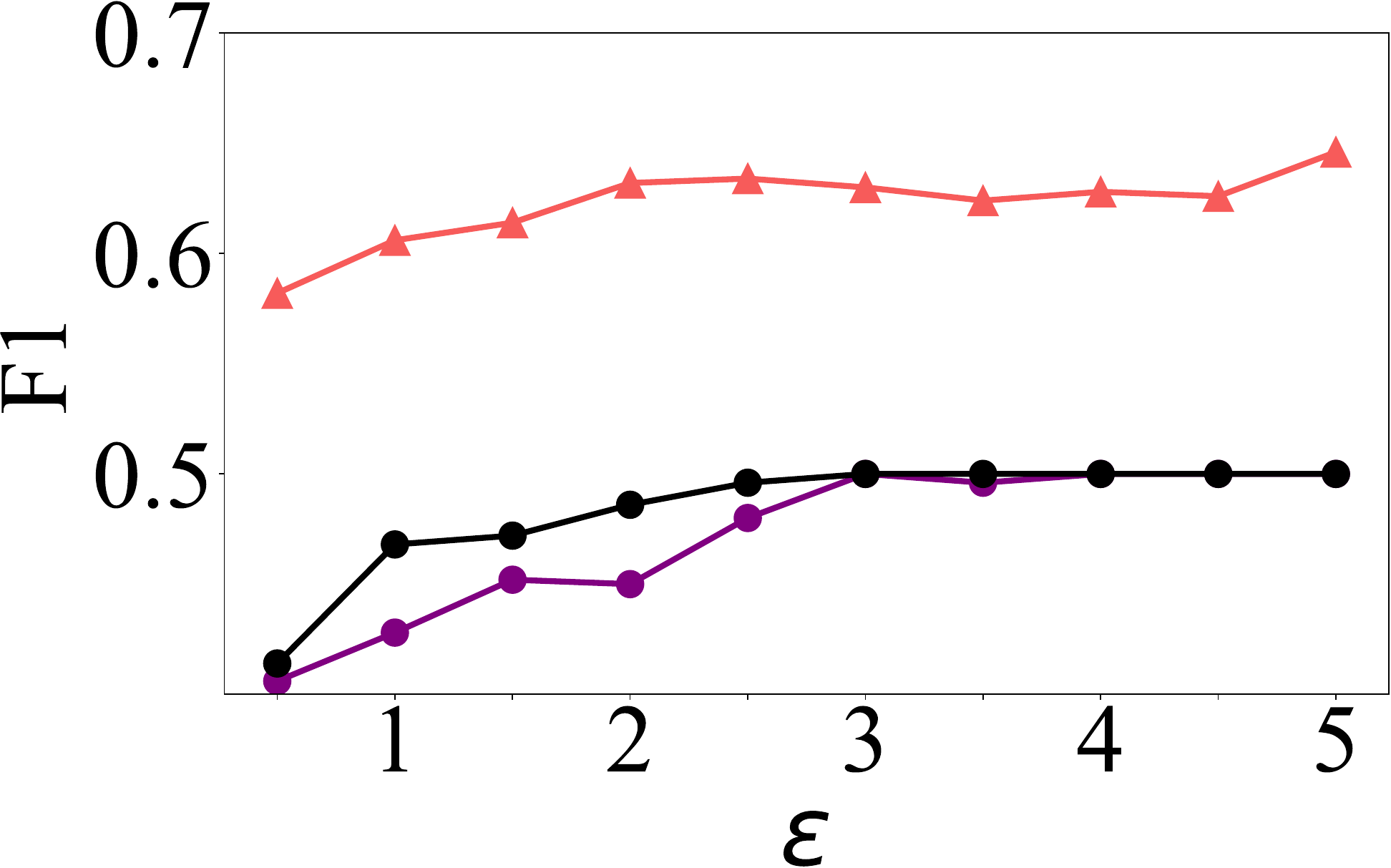}
    	}\hfill
    
    \vspace{-0.1in}
     \subfloat[RDB, OLH, $k$=10]{
        	\includegraphics[width=0.185\textwidth]{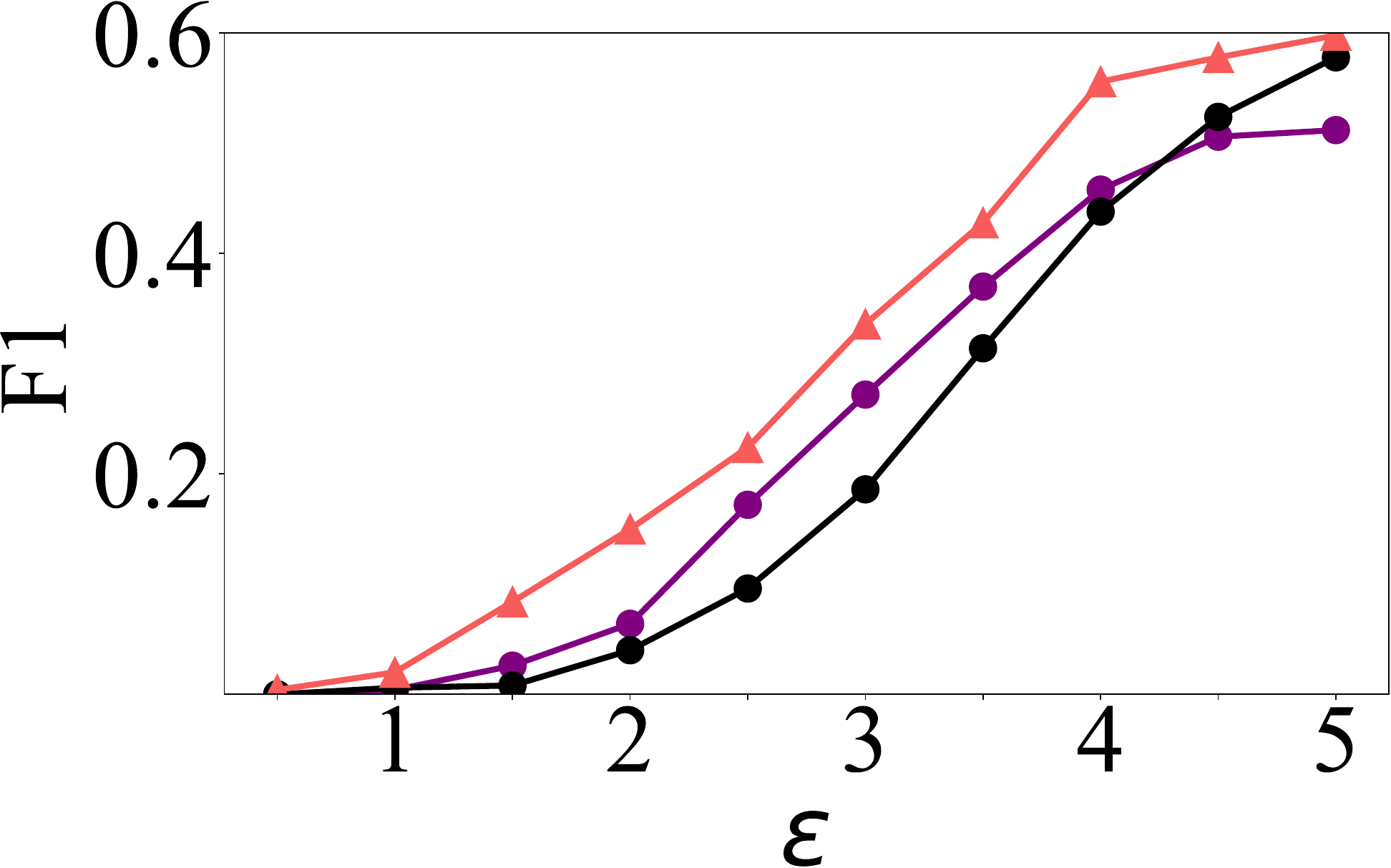}
        }\hfill
        \subfloat[YCM, OLH, $k$=10]{
    		\includegraphics[width=0.185\textwidth]{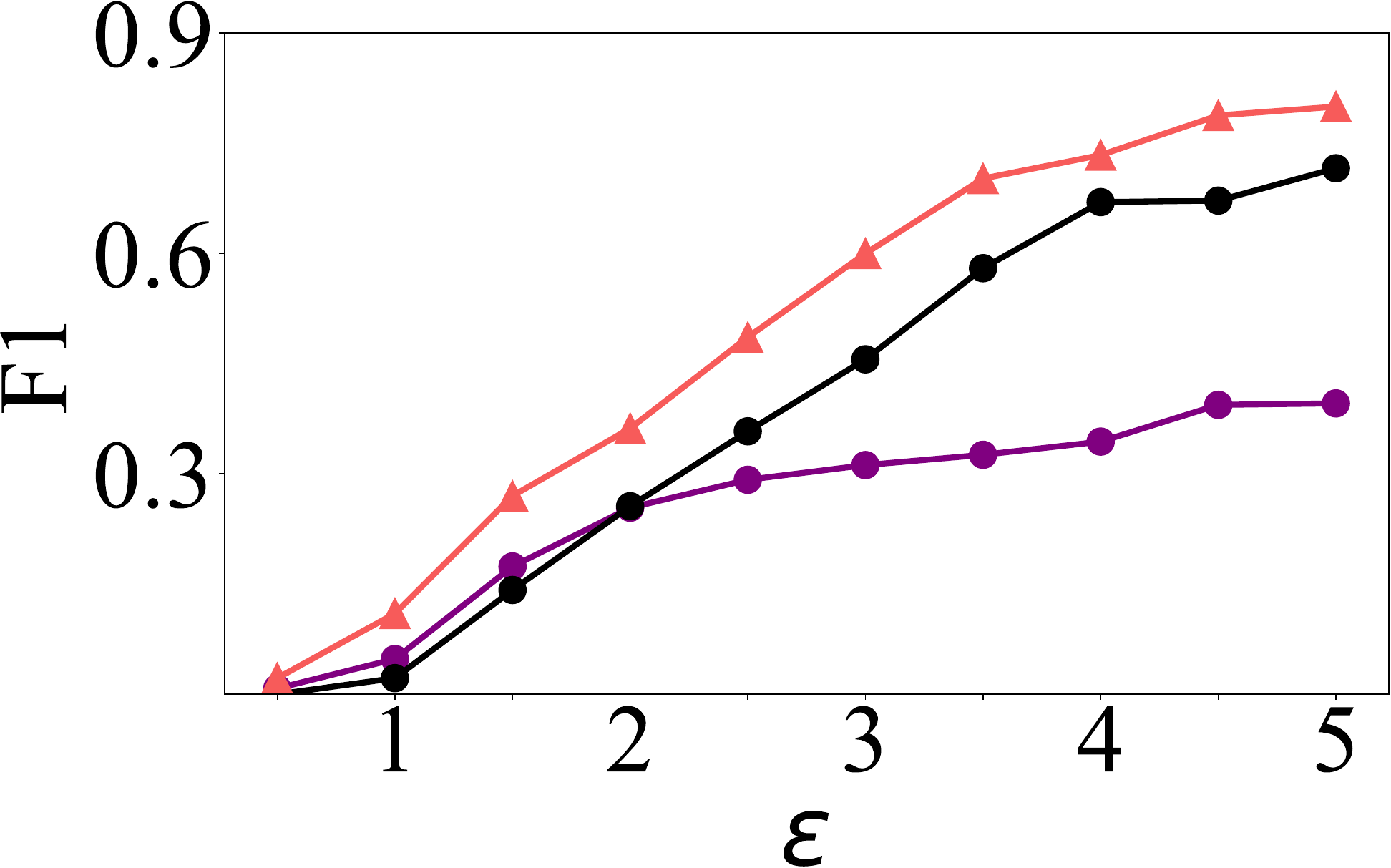}
    	}\hfill
        \subfloat[TYS, OLH, $k$=10]{
    		\includegraphics[width=0.185\textwidth]{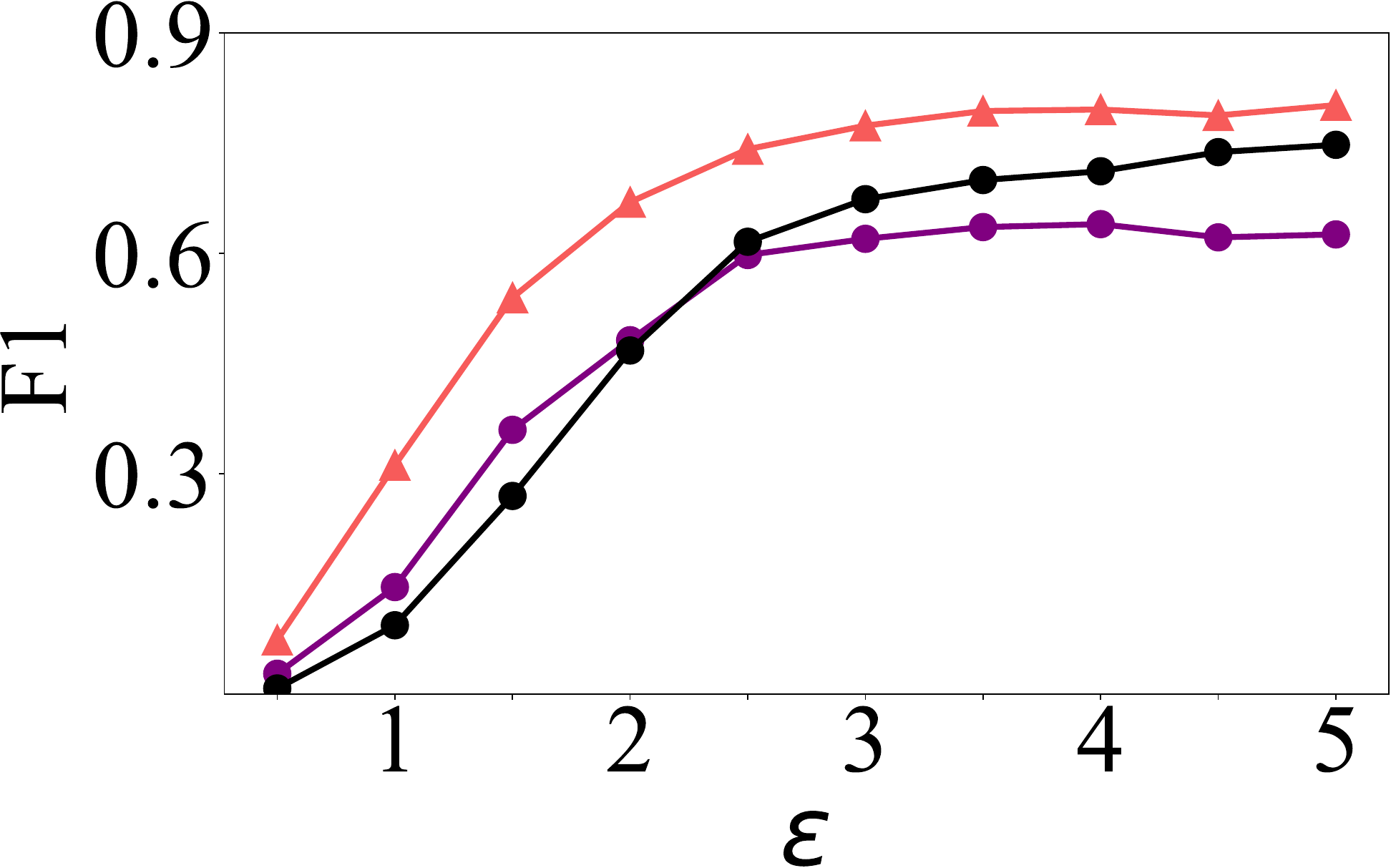}
    	}\hfill
        \subfloat[UBA, OLH, $k$=10]{
    		\includegraphics[width=0.185\textwidth]{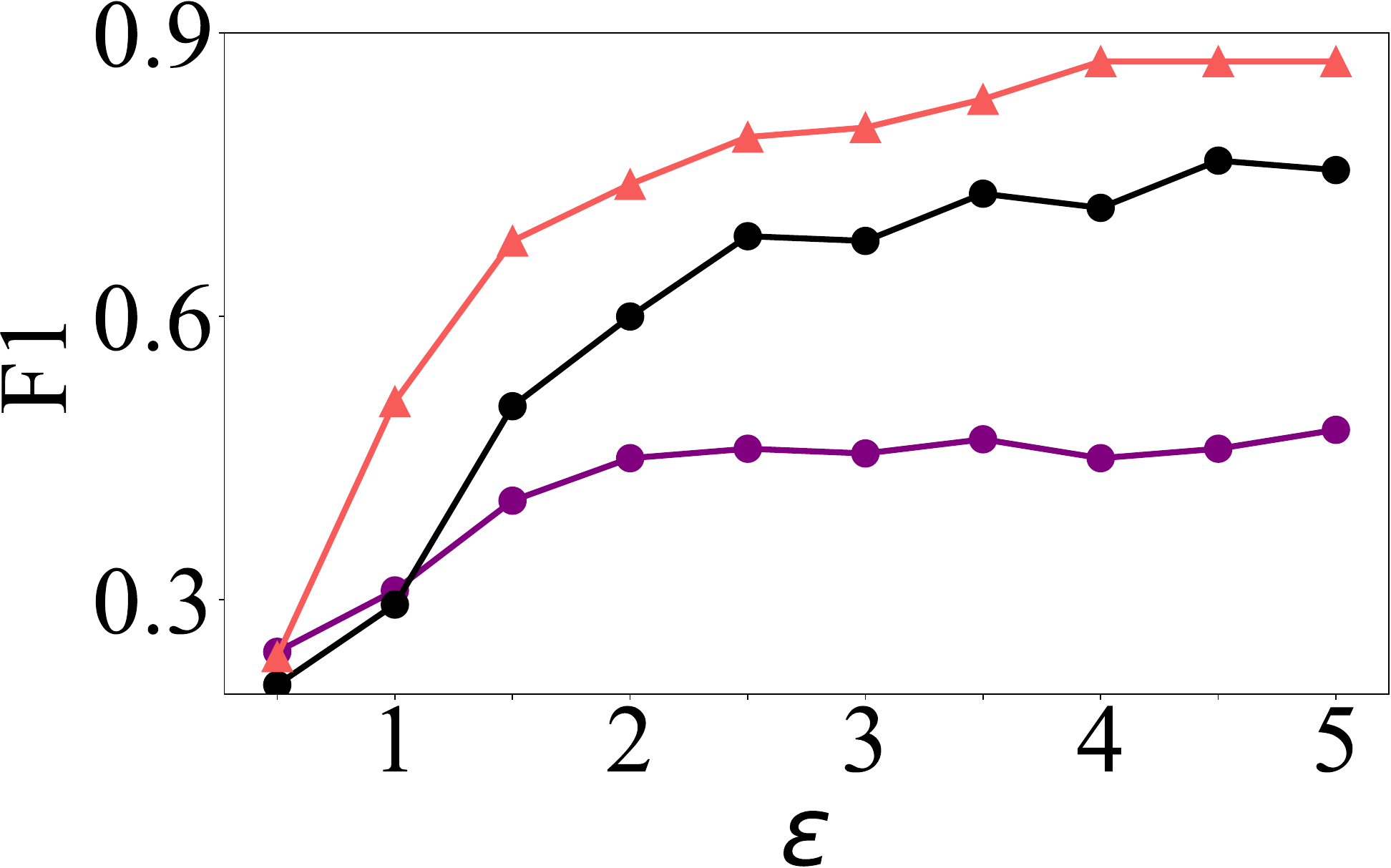}
    	}\hfill
        \subfloat[SYN, OLH, $k$=10]{
    		\includegraphics[width=0.185\textwidth]{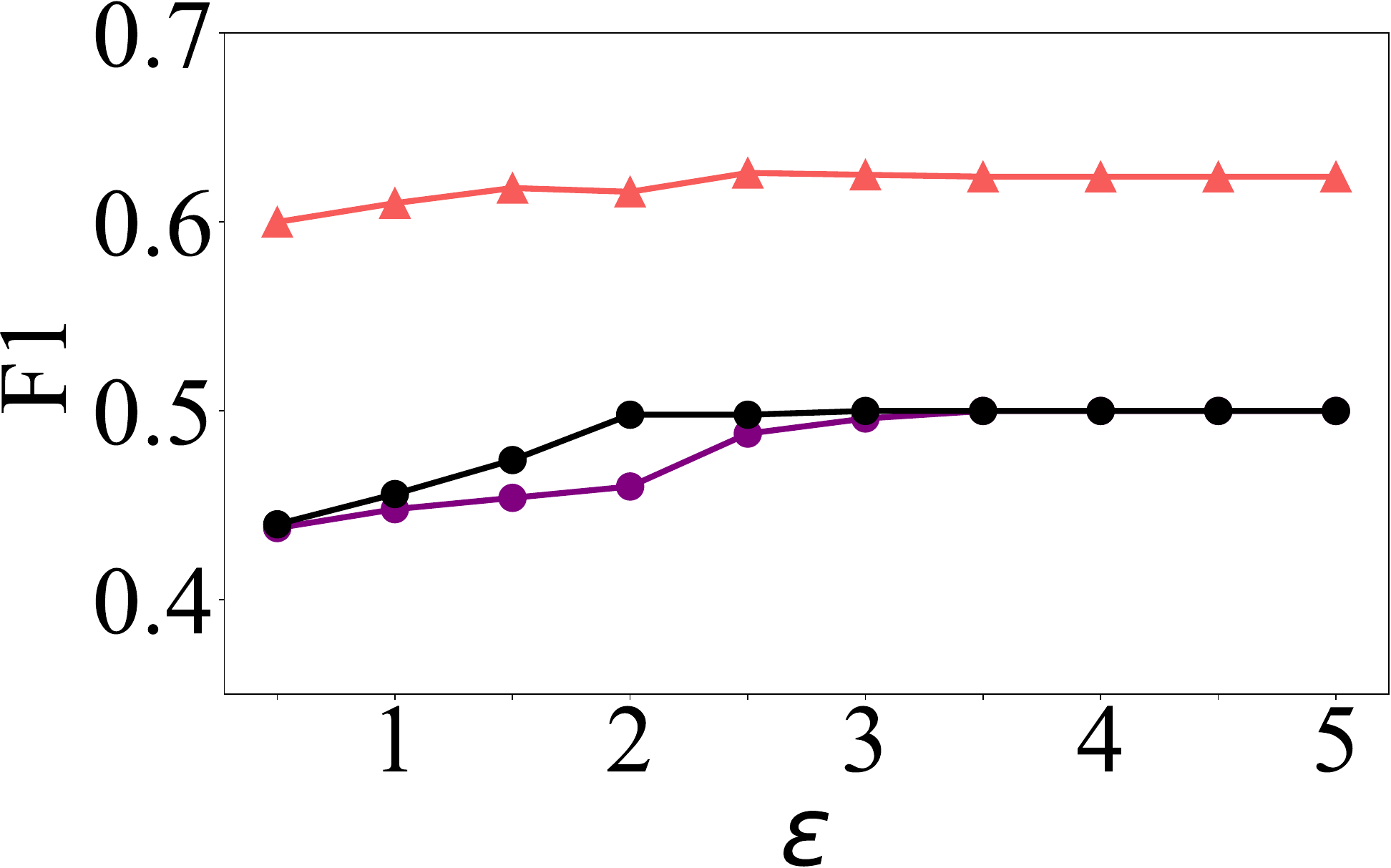}
    	}\hfill
    \end{minipage}
    \vspace{-0.1in}
    \caption{F1 scores vs. privacy budget $\epsilon$ under OUE and OLH.}
    \label{fig_different_FO}
\end{figure*}

\begin{figure*}[t]
    \centering
    \includegraphics[width=0.6\textwidth]{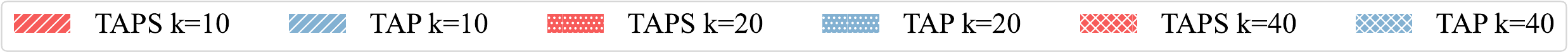} \\
    \vspace{-0.15in}
    \begin{minipage}{0.9\textwidth}
        \subfloat[RDB]{
        	\includegraphics[width=0.185\textwidth]{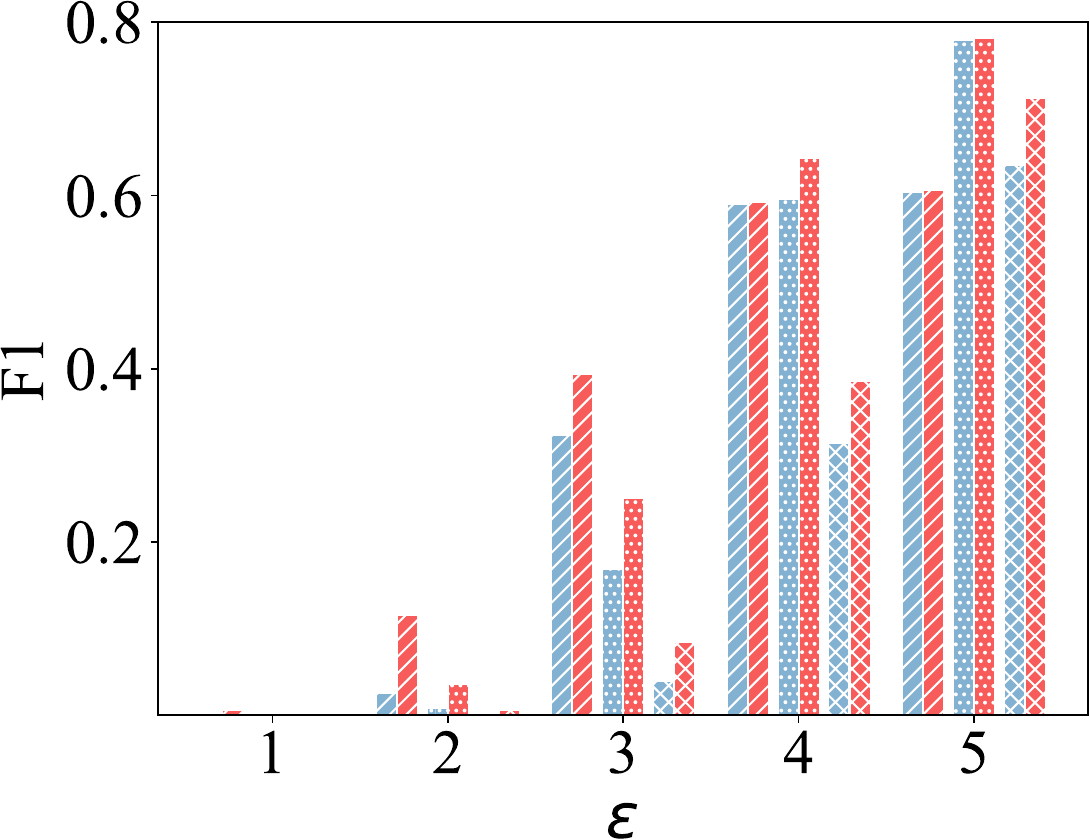}
        }\hfill
        \subfloat[YCM]{
    		\includegraphics[width=0.185\textwidth]{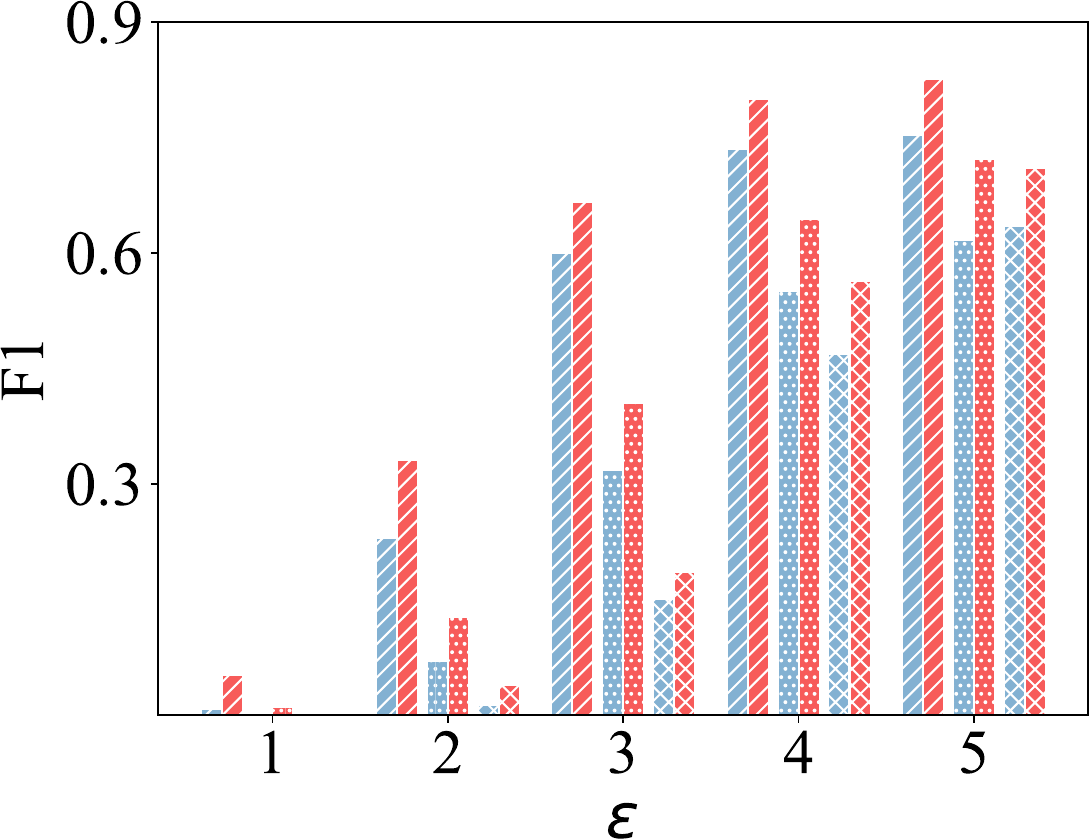}
    	}\hfill
        \subfloat[TYS]{
    		\includegraphics[width=0.185\textwidth]{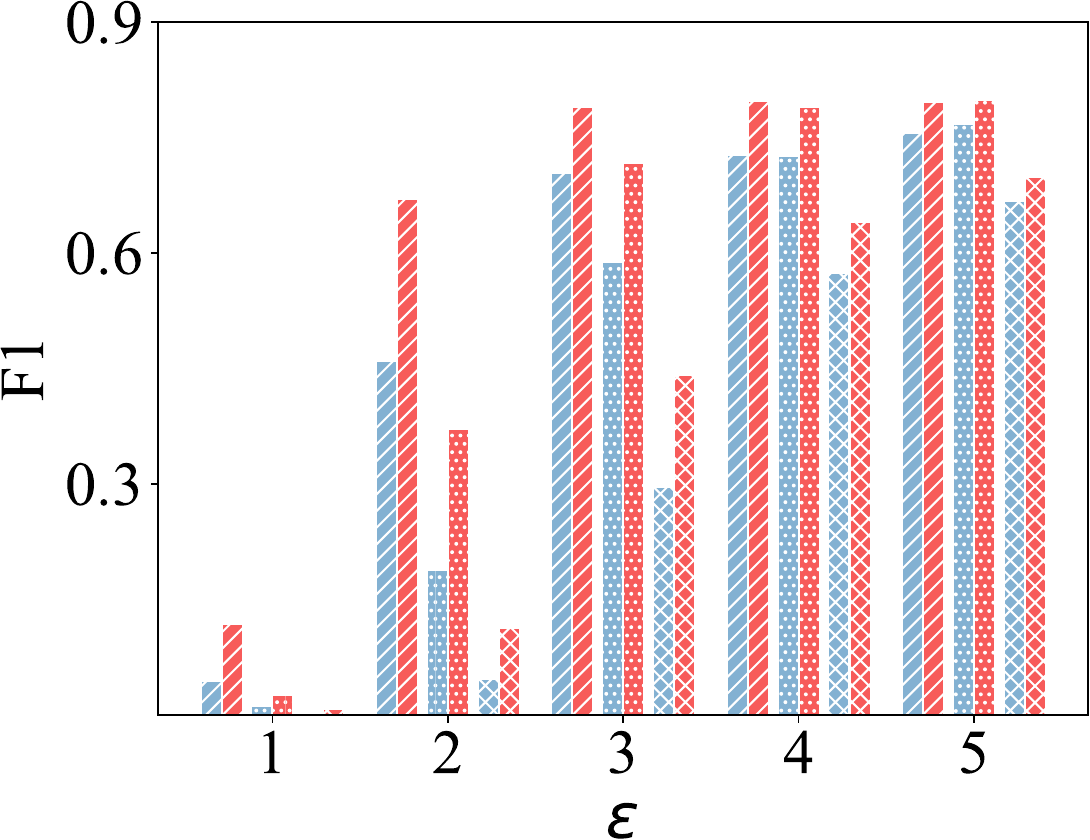}
    	}\hfill
        \subfloat[UBA]{
    		\includegraphics[width=0.185\textwidth]{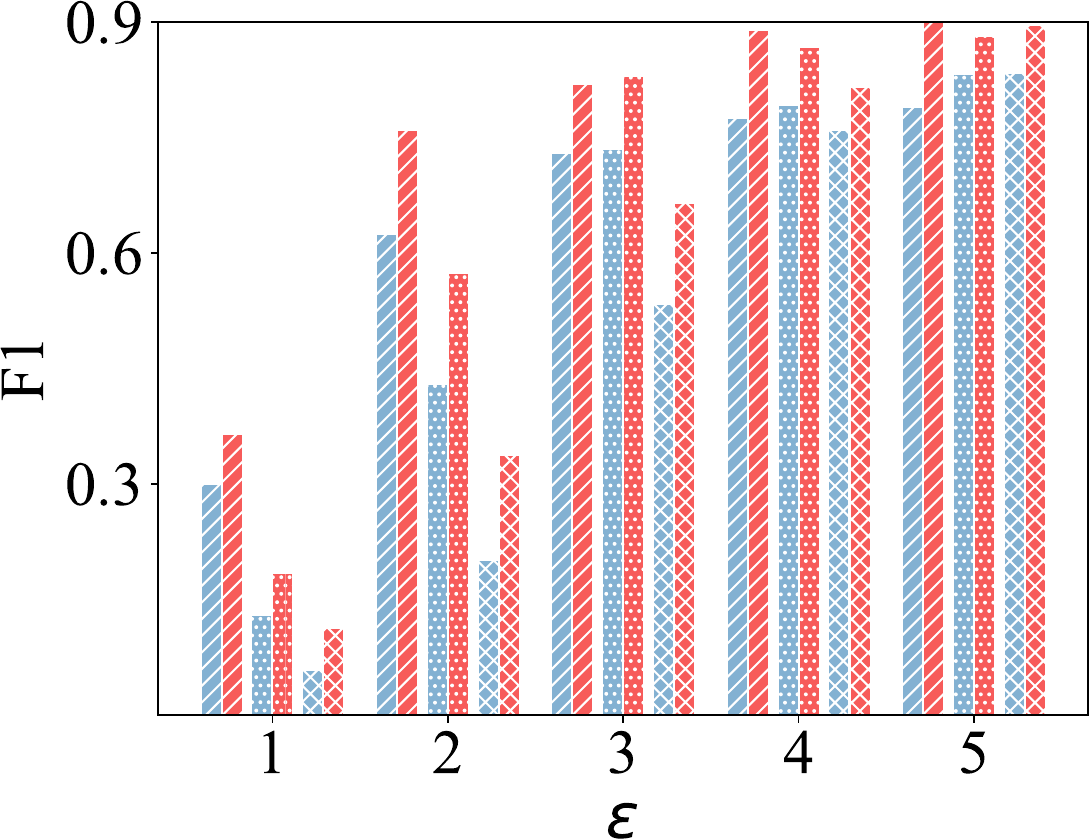}
    	}\hfill
        \subfloat[SYN]{
    		\includegraphics[width=0.185\textwidth]{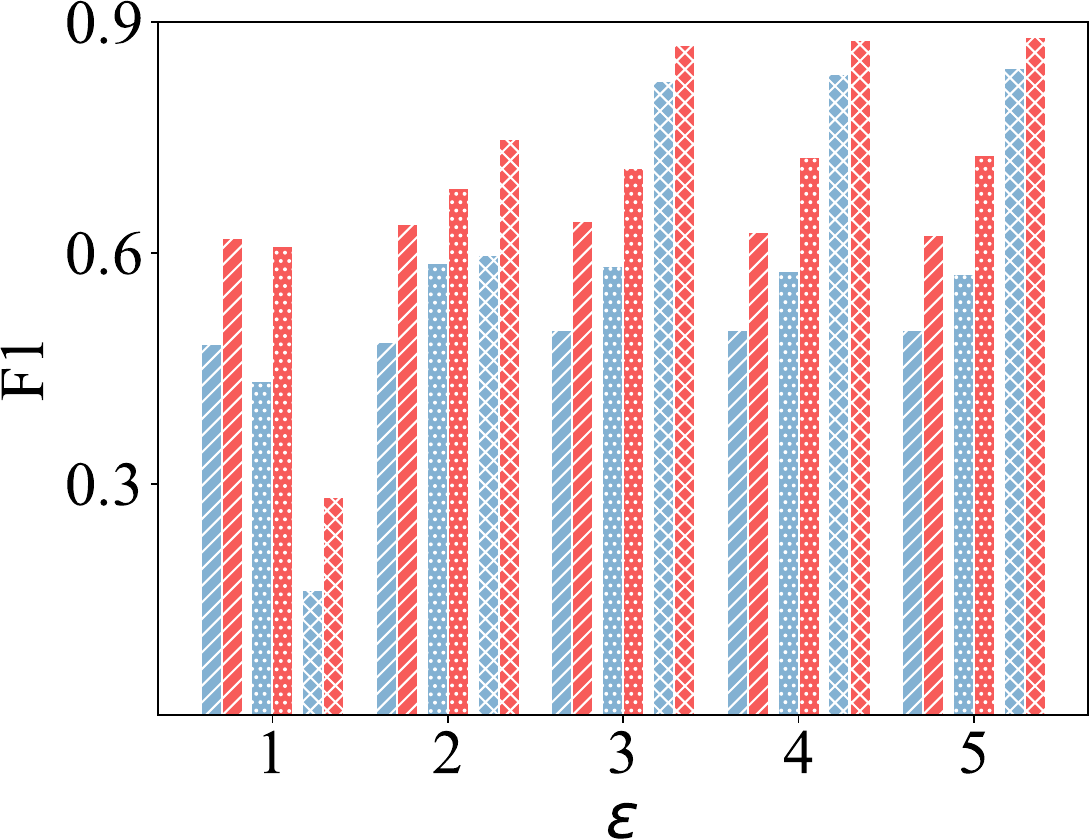}
    	}\hfill
    \end{minipage}
    \vspace{-0.15in}
    \caption{F1 scores of TAPS with or without the consensus-based pruning under varying privacy budget $\epsilon$.}
    \vspace{-0.1in}
    \label{fig_without_pruning}
\end{figure*}

\subsection{Overall Results on Utility}

\begin{table*}[t]
\centering
\caption{F1 scores with varying step sizes ($\epsilon=4$ and $k=10$).}
\vspace{-0.15in}
\resizebox{\textwidth}{!}{
\begin{tabular}{c|c|c|c|c|c|c|c|c|c|c|c|c|c|c|c|c|c|c|c|c|c|c|c|c}
\hline
\multicolumn{1}{c|}{Dataset} & Step size & GTF & FedPEM & TAPS & Dataset & Step size & GTF & FedPEM & TAPS & Dataset & Step size & GTF & FedPEM & TAPS & Dataset & Step size & GTF & FedPEM & TAPS & Dataset & Step size & GTF & FedPEM & TAPS \\ \hline
\multirow{3}{*}{RDB} & 2 & 0.544 & 0.552 & \textbf{0.592} & \multirow{3}{*}{YCM} & 2 & 0.402 & 0.736 & \textbf{0.80} & \multirow{3}{*}{TYS} & 2 & 0.58  & 0.736 & \textbf{0.798} & \multirow{3}{*}{UBA} & 2 & 0.44  & 0.78  & \textbf{0.89} & \multirow{3}{*}{SYN} & 2 & 0.50  & 0.50  & \textbf{0.628}\\
                     & 4 & 0.392 & 0.632 & \textbf{0.756} &   & 4 & 0.392 & 0.632 & \textbf{0.756} &   & 4 & 0.642 & 0.69  & \textbf{0.78}  &   & 4 & 0.50  & 0.755  & \textbf{0.87}  &
                     & 4 & 0.50  & 0.50  & \textbf{0.64}\\
                     & 6 & 0.36  & 0.256 & \textbf{0.382} &   & 6 & 0.346 & 0.572 & \textbf{0.656} &   & 6 & 0.634 & 0.668 & \textbf{0.748} &   & 6 & 0.415 & 0.685  & \textbf{0.815}  &
                     & 6 & 0.498 & 0.50  & \textbf{0.63}  \\ \hline
\end{tabular}}
\vspace{-0.1in}
\label{table_different_stepsize}
\end{table*}

We first evaluate the performance of the proposed TAPS mechanism over five datasets based on F1 score. We set privacy budget $\epsilon$ from $1$ to $5$, and the number of frequent items $k$ from $10$ to $40$. As shown in Figure~\ref{fig_F1}, 
although GTF can filter some noisy items when the privacy budget $\epsilon$ is small, it ignores the impacts of different quantities across parties and prunes too many inherent frequent but similar items, resulting in a poor performance. In the case of FedPEM, due to the independent estimation, it estimates many redundant prefixes or items and thus cannot lead to good performance. In contrast, our TAPS mechanism significantly outperforms baselines in all cases, even when $k$ is relatively large and $\epsilon$ is small and as the number of parties grows. For example, when $\epsilon = 1$, TAPS achieves an F1 score of 0.283 on the SYN dataset, compared to 0.1395 for the best baseline, demonstrating an improvement of over 100\%. This highlights the effectiveness of our adaptive extension, shared trie construction, and the consensus-based pruning.

With regard to the NCR metric, the TAPS mechanism always outperforms the baselines, as shown in Figure~\ref{fig_NCR}. The GTF mechanism achieves a higher NCR score in SYN when $k=10$ compare to its F1 score.  This discrepancy is due to the presence of some items in SYN whose frequencies are extremely high in certain parties. Therefore, with a relatively small $k=10$ and a large number of parties in SYN, GTF can achieve a higher NCR score. However, as $k$ increases, it drops many true negatives and thus diminishes its utility. Unlike the baselines, our TAPS mechanism benefits from the proposed strategies and consistently identify accurate heavy hitters, further demonstrating its effectiveness.

\subsection{Impact Study on Different FOs}\label{subsec:different FOs}
We further investigate the impact of different FOs on the performance. Specifically, for all the solutions, we substitute $k$-RR with two typical FOs, namely OUE~\cite{WBL2017} and OLH~\cite{WBL2017}, respectively. Note that for OUE, we assign a dummy position in the encoded vector for the out-of-domain items or prefixes. As shown in Figure~\ref{fig_different_FO}, the results are mostly consistent with those using $k$-RR. The proposed TAPS mechanism outperforms the baselines, which demonstrates the advantages and robustness of the proposed strategies. This indicates that TAPS is a general solution for federated hitter identification, and can be well adaptable to different FOs. As such, we can freely choose an appropriate FO according to the requirements of communication constraints and domain size of specific application scenario~\cite{WBL2017} to achieve better performance.

\begin{table*}[t]
\centering
\caption{Scalability evaluation under varying user population in UBA ($\epsilon=4$ and $k=10$).}
\vspace{-0.15in}

\resizebox{0.7\textwidth}{!}{
\begin{tabular}{c|c|c|c|c|c|c|c|c|c|c|c|c|c}
\hline
& \multicolumn{3}{c|}{F1 Scores} & \multicolumn{5}{c|}{Communication Costs} & \multicolumn{5}{c}{Running Time} \\ 
\hline
\multicolumn{1}{c|}{User proportion} & GTF & FedPEM & TAPS & GTF & FedPEM & TAPS & OUE & OLH & GTF & FedPEM & TAPS  & OUE & OLH \\ \hline
25\% & 0.48  & 0.785  & 0.9   & 18 kb & 18 kb & 69 kb   & $>$ 2 PiB   &$> 10^{4}$ kb & 2.1105 s  & 2.0688 s  & 12.3642 s & $> 72$ h  & $> 72$ h \\ 
50\% & 0.44  & 0.78   & 0.89  & 18 kb    & 18 kb    & 70 kb   &$>$ 2 PiB    &$> 10^{4}$ kb & 5.9571 s  & 6.2007 s  & 40.5691 s & $> 72$ h & $> 72$ h \\ 
75\% & 0.42  & 0.76   & 0.88 & 17 kb    & 17 kb    &54 kb   &$>$ 2 PiB    &$> 10^{5}$ kb & 11.19 s   & 11.9729 s & 79.9358 s & $> 72$ h & $> 72$ h \\ 
100\% & 0.44 & 0.78   & 0.89 & 17 kb    & 17 kb    &75 kb   &$>$ 2 PiB    &$> 10^{5}$ kb & 17.5221 s & 19.8624 s & 128.0133 s & $> 72$ h & $> 72$ h \\ 
\hline
\end{tabular}
}
\vspace{-0.1in}
\label{table_scalability}
\end{table*}

\begin{table}[t]
\centering
\caption{F1 scores of TAPS with fixed/adaptive extension numbers ($\epsilon=4$ and $k=10$). 
}
\vspace{-0.17in}

\resizebox{0.65\columnwidth}{!}{\begin{tabular}{c|c|c|c|c|c}
\hline
& $t=\lfloor k/2 \rfloor$ & $t=k$ & $t=2k$ & $t=3k$ & Adaptive t \\ \hline
RDB & 0.424   & 0.574   & 0.582   & 0.546   & \textbf{0.598}\\ \hline
YCM & 0.644   & 0.786   & 0.782   & 0.762   & \textbf{0.80}\\ \hline
TYS  & 0.753  & 0.769   & 0.781   & 0.785   & \textbf{0.80}\\ \hline
UBA  & 0.89    & 0.89     & 0.885     & 0.87     & \textbf{0.89}\\ \hline
SYN  & 0.624    & 0.614     & 0.602     & 0.60     & \textbf{0.628}\\ \hline

\end{tabular}}
\vspace{-0.1in}
\label{table_adaptive_extension_number}
\end{table}

\begin{table}[t]
\centering
\caption{F1 scores of TAPS with or without the shared trie ($\epsilon=4$ and $k=10$). 
}
\vspace{-0.17in}

\resizebox{0.65\columnwidth}{!}{\begin{tabular}{c|c|c|c|c|c}
\hline
     & RDB & YCM & TYS & UBA & SYN \\ \hline
TAPS (w/o shared trie) & 0.584 & 0.785 & 0.783 & 0.878 & 0.61  \\ \hline
TAPS                   & \textbf{0.598}  & \textbf{0.80} & \textbf{0.80} & \textbf{0.89} & \textbf{0.628}                   \\ \hline

\end{tabular}}
\vspace{-0.1in}
\label{table_wo_shared_trie}
\end{table}

\begin{table}[t]
\centering
\caption{Average recall scores in addressing statistical heterogeneity ($\epsilon=4$, $k=10$, \small{$\uparrow$}: improvement over the best baseline). 
}
\vspace{-0.17in}

\resizebox{0.55\columnwidth}{!}{\begin{tabular}{c|c|c|c|c}
\hline
&\# parties & GTF & FedPEM & TAPS \\ \hline
             RDB & 2 & 0.304 & 0.302 & \textbf{0.413 \small{($\uparrow
$ 35.9\%)}} \\ \hline
             YCM & 4 & 0.262 & 0.255 & \textbf{0.368 \small{($\uparrow
$ 40.5\%)}} \\ \hline
             TYS & 6 & 0.296 & 0.283 & \textbf{0.333 \small{($\uparrow
$ 12.5\%)}} \\ \hline
             UBA & 6 & 0.295 & 0.294 & \textbf{0.325 \small{($\uparrow
$ 10.2\%)}} \\ \hline
             SYN & 8 & 0.162 & 0.162 & \textbf{0.179 \small{($\uparrow
$ 10.5\%)}} \\ \hline
\end{tabular}}
\vspace{-0.1in}
\label{table_statistical_heterogeneity}
\end{table}

\begin{table}[t]
\centering
\caption{F1 scores with varying data heterogeneity controlled by $\beta$ in SYN ($\epsilon=4$ and $k=10$). 
}
\vspace{-0.17in}

\resizebox{0.5\columnwidth}{!}{\begin{tabular}{c|c|c|c}
\hline
& GTF & FedPEM & TAPS \\ \hline
             $Dir_N(\beta=0.2)$ & 0.384  & 0.40  & \textbf{0.538} \\ \hline
             $Dir_N(\beta=0.5)$& 0.50 & 0.50   & \textbf{0.628} \\ \hline
             $Dir_N(\beta=0.8)$& 0.50 & 0.552 & \textbf{0.624} \\ \hline
\end{tabular}}
\vspace{-0.1in}
\label{table_data_heterogeneity}
\end{table}

\subsection{Ablation Study}\label{subsec:ablation study}
We explore the impact of various parameters or components on performance, including different step sizes, adaptive trie extension, and shared shallow trie construction. Then, we assess our TAPS mechanism in addressing statistical heterogeneity. We also evaluate the effectiveness of our consensus-based pruning strategy and the robustness of TAPS under varying degrees of data heterogeneity.

We first study the performance of TAPS and baselines under varying step sizes---different extension lengths $\lfloor m/g\rfloor\in\{2,4,6\}$, where $m$ is the maximum length and $g$ denotes the granularity (i.e., iterations). Note that we fix the privacy budget $\epsilon=4$ and query $k=10$. We can observe that in Table~\ref{table_different_stepsize}, TAPS consistently performs the best with its F1 scores shown in bold, which demonstrate the effectiveness of the proposed adaptive trie extension, shared shallow trie construction, and pruning strategies. Besides, the larger extension length also enhances the benefits of pruning a lot.

Next, we investigate the adaptive trie extension strategy and set $\epsilon=4$ and $k=10$. Unless otherwise specified, we set the step size/extension length to $2$ in the following experiments. We evaluate the TAPS mechanism with different fixed extension numbers from $\{\lfloor k/2\rfloor,k,2k,3k\}$. In Table~\ref{table_adaptive_extension_number}, we observe that the optimal fixed extension number varies across different datasets; however, our adaptive strategy consistently outperforms the fixed options, highlighting the benefits from adaptive trie construction. 
Additionally, we further assess the impact of the shared shallow trie construction on the accuracy of federated heavy hitters. Specifically, we remove the shared shallow trie construction step in TAPS and then compare the results with those by TAPS. As shown in Table~\ref{table_wo_shared_trie}, the TAPS with the shared shallow trie construction performs better across all datasets, which confirms its benefit.

We then evaluate the performance of TAPS on statistical heterogeneity by analyzing the average recall scores of global ground truths identified as local heavy hitters across different parties. As shown in Table~\ref{table_statistical_heterogeneity}, TAPS outperforms all baselines, achieving at least a $10\%$ improvement over the best baseline across all datasets. This can be attributed to our shared shallow trie and adaptive extension strategies, which help align the heavy hitter targets. Additionally, the consensus-based pruning strategy mitigates statistical heterogeneity by leveraging prior knowledge from other parties.

We further validate the benefits of the consensus-based pruning strategy. As shown in Figure~\ref{fig_without_pruning}, the TAPS mechanism consistently outperforms its non-pruning version (TAP) across various datasets and queries $k$, underscores the value of our pruning strategy.

Then we study the impact of data heterogeneity on our TAPS mechanism's performance over the SYN dataset by varying the imbalance level of domain skewness. This is achieved by varying the Dirichlet distribution parameter $\beta\in\{0.2,0.5,0.8\}$, where a smaller $\beta$ indicates more imbalance. As shown in Table~\ref{table_data_heterogeneity} under $\epsilon=4$ and $k=10$, TAPS significantly outperforms baselines across various degrees of non-IID conditions. Both the adaptive trie extension and the consensus-based pruning strategies effectively mitigate the negative impacts of non-IID, leading to superior performance.

\subsection{Impact Study on Scalability}\label{subsec:scalability}
We finally examine the scalability of our TAPS mechanism across varying user populations in the largest UBA dataset. We randomly sample different subsets of users, $25\%$, $50\%$, $75\%$, and $100\%$, in UBA to evaluate the F1 scores, communication costs, and running times. As shown in Table~\ref{table_scalability}, our TAPS mechanism consistently outperforms the baselines across these varied populations, while maintaining competitive and acceptable communication and computation costs. The superior performance is because the shared shallow trie construction and the adaptive extension strategies help to preserve necessary prefixes by aggregating at a shallow level, and the consensus-based pruning strategy assists to leverage prior knowledge from other party to provide more informative frequency information. 
Although the communication cost of TAPS is slightly higher than that of GTF and FedPEM, primarily due to the pruning strategy, it remains significantly lower than that of directly using OUE or OLH. Regarding the computation costs, TAPS requires more running time than GTF and FedPEM since sequential estimation consumes more time than the baselines that can be executed in parallel in different parties. However, it is still substantially less than that required by directly using OLH, demonstrating the practicality and robustness of TAPS across different user scales in real-world settings.

\section{Conclusion}\label{conclusion}
In this paper, we propose a novel target-aligning prefix tree mechanism satisfying $\epsilon$-LDP, for federated heavy hitter analytics. We treat the estimation as a two-phase process, which applies the shared shallow trie construction and adaptive trie extension strategies. The proposed mechanism is further optimized by integrating a consensus-based pruning strategy, which enables each party to utilize (noisy) prior knowledge from the previous party to adaptively prune the candidate domain. To the best of our knowledge, this is the first solution to the federated heavy hitter analytics in a cross-party setting under the stringent $\epsilon$-LDP. Extensive experiments on both synthetic and real-world datasets demonstrate the effectiveness of our mechanism.

\section*{Acknowledgment}
This work was supported by the National Natural Science Foundation of China (Grant No: 92270123, 62102334, 62372122 and 62072390), and the Research Grants Council, Hong Kong SAR, China (Grant No:  15209922, 15208923, 15210023 and 25207224).

\bibliographystyle{ACM-Reference-Format}
\bibliography{mybibliography}

\end{document}